\documentclass[journal]{IEEEtran}
\usepackage{cite}
\ifCLASSINFOpdf
  \usepackage[pdftex]{graphicx}
\else
  \usepackage[dvips]{graphicx}
\fi
\usepackage{amsmath}
\ifCLASSOPTIONcompsoc   
	\usepackage[caption=false,font=normalsize,labelfont=sf,textfont=sf]{subfig}
\else
	\usepackage[caption=false,font=footnotesize]{subfig}
\fi
\usepackage{bm}
\usepackage{amsmath}
\usepackage{amsthm}
\usepackage{amsfonts}
\usepackage{amssymb}
\usepackage{color}
\usepackage{algorithm}
\usepackage{setspace}
\usepackage{tabularx,booktabs}
\usepackage{multirow}
\usepackage{algorithmicx}
\usepackage{enumitem}
\usepackage{algpseudocode}
\makeatletter
\def\BState{\State\hskip-\ALG@thistlm}
\makeatother

\usepackage{dsfont}
\usepackage{psfrag}

\usepackage{caption}
\errorcontextlines\maxdimen

\makeatletter
    \newcommand*{\algrule}[1][\algorithmicindent]{\makebox[#1][l]{\hspace*{.5em}\thealgruleextra\vrule height \thealgruleheight depth \thealgruledepth}}%
\newcommand*{\thealgruleextra}{}
\newcommand*{\thealgruleheight}{.75\baselineskip}
\newcommand*{\thealgruledepth}{.25\baselineskip}

\newcount\ALG@printindent@tempcnta
\def\ALG@printindent{%
    \ifnum \theALG@nested>0% is there anything to print
        \ifx\ALG@text\ALG@x@notext% is this an end group without any text?
        \else
            \unskip
            \addvspace{-1pt}% FUDGE to make the rules line up
            \ALG@printindent@tempcnta=1
            \loop
                \algrule[\csname ALG@ind@\the\ALG@printindent@tempcnta\endcsname]%
                \advance \ALG@printindent@tempcnta 1
            \ifnum \ALG@printindent@tempcnta<\numexpr\theALG@nested+1\relax% can't do <=, so add one to RHS and use < instead
            \repeat
        \fi
    \fi
    }%
\usepackage{etoolbox}
\patchcmd{\ALG@doentity}{\noindent\hskip\ALG@tlm}{\ALG@printindent}{}{\errmessage{failed to patch}}
\makeatother

\newbox\statebox
\newcommand{\myState}[1]{%
    \setbox\statebox=\vbox{#1}%
    \edef\thealgruleheight{\dimexpr \the\ht\statebox+1pt\relax}%
    \edef\thealgruledepth{\dimexpr \the\dp\statebox+1pt\relax}%
    \ifdim\thealgruleheight<.75\baselineskip
        \def\thealgruleheight{\dimexpr .75\baselineskip+1pt\relax}%
    \fi
    \ifdim\thealgruledepth<.25\baselineskip
        \def\thealgruledepth{\dimexpr .25\baselineskip+1pt\relax}%
    \fi
    \State #1%
    \def\thealgruleheight{\dimexpr .75\baselineskip+1pt\relax}%
    \def\thealgruledepth{\dimexpr .25\baselineskip+1pt\relax}%
}
\DeclareMathOperator{\diag}{diag}

\newcommand{\norm}[1]{\left\lVert#1\right\rVert}
\newcommand{\tpose}{\mathsf{T}}
\newcommand{\hermi}{\mathsf{H}}

\definecolor{Dblue}{rgb}{0,0,1}
\definecolor{Dbrown}{rgb}{0.59,0.4,0}
\definecolor{Dred}{rgb}{0,0,0}
\definecolor{Dgreen}{rgb}{0,0.4,0}
\definecolor{Nblue}{rgb}{0,0,0.65}

\def \dred {\color{Dred}}

\newtheorem{prop}{Proposition}

\begin{document}
%\bstctlcite{IEEEexample:BSTcontrol}
\bstctlcite{IEEEexample:BSTcontrol}
%
% paper title
% Titles are generally capitalized except for words such as a, an, and, as,
% at, but, by, for, in, nor, of, on, or, the, to and up, which are usually
% not capitalized unless they are the first or last word of the title.
% Linebreaks \\ can be used within to get better formatting as desired.
% Do not put math or special symbols in the title.
\title{Sparsity-Assisted Signal Denoising and \\ Pattern Recognition in Time-Series Data}
%
%
% author names and IEEE memberships
% note positions of commas and nonbreaking spaces ( ~ ) LaTeX will not break
% a structure at a ~ so this keeps an author's name from being broken across
% two lines.
% use \thanks{} to gain access to the first footnote area
% a separate \thanks must be used for each paragraph as LaTeX2e's \thanks
% was not built to handle multiple paragraphs
%

\author{G.V.~Prateek,~\IEEEmembership{Student Member,~IEEE,}
        Yo-El~Ju,
        and~Arye~Nehorai,~\IEEEmembership{Life~Fellow,~IEEE}% <-this % stops a space
%\thanks{M. Shell was with the Department
%of Electrical and Computer Engineering, Georgia Institute of Technology, Atlanta,
%GA, 30332 USA e-mail: (see http://www.michaelshell.org/contact.html).}% <-this % stops a space
%\thanks{J. Doe and J. Doe are with Anonymous University.}% <-this % stops a space
%\thanks{Manuscript received April 19, 2005; revised August 26, 2015.}
\thanks{Research reported in this publication was supported by National Institutes of Health awards K23-NS089922 (YSJ, GVP), R01AG059507 (AN, YSJ, GVP), and UL1RR024992 Sub-Award KL2-TR000450 (YSJ). The content is solely the responsibility of the authors and does not necessarily represent the official views of the National Institutes of Health.}
}

\maketitle

% As a general rule, do not put math, special symbols or citations
% in the abstract or keywords.
\begin{abstract}
We address the problem of signal denoising and pattern recognition in processing batch-mode time-series data by combining linear time-invariant filters, orthogonal multiresolution representations, and sparsity-based methods. We propose a novel approach to designing higher-order zero-phase low-pass, high-pass, and band-pass infinite impulse response filters as matrices, using spectral transformation of the state-space representation of digital filters. We also propose a proximal gradient-based technique to factorize a special class of zero-phase high-pass and band-pass digital filters so that the factorization product preserves the zero-phase property of the filter and also incorporates a sparse-derivative component of the input in the signal model. {\dred To demonstrate applications of our novel filter designs, we validate and propose new signal models to simultaneously denoise and identify patterns of interest. We begin by using our proposed filter design to test an existing signal model that simultaneously combines linear time invariant (LTI) filters and sparsity-based methods. We develop a new signal model called sparsity-assisted signal denoising (SASD) by combining our proposed filter designs with the existing signal model. Using simulated data, we demonstrate the robustness of the SASD signal model across different orders of filter and noise levels. Thereafter, we propose and derive a new signal model called sparsity-assisted pattern recognition (SAPR). In SAPR, we combine LTI band-pass filters and sparsity-based methods with orthogonal multiresolution representations, such as wavelets, to detect specific patterns in the input signal. Finally, we combine the signal denoising and pattern recognition tasks, and derive a new signal model called the sparsity-assisted signal denoising and pattern recognition (SASDPR). We illustrate the capabilities of the SAPR and SASDPR frameworks using sleep-electroencephalography data to detect K-complexes and sleep spindles, respectively.}
\end{abstract}

% Note that keywords are not normally used for peerreview papers.

% For peer review papers, you can put extra information on the cover
% page as needed:
% \ifCLASSOPTIONpeerreview
% \begin{center} \bfseries EDICS Category: 3-BBND \end{center}
% \fi
%
% For peerreview papers, this IEEEtran command inserts a page break and
% creates the second title. It will be ignored for other modes.
\IEEEpeerreviewmaketitle

%%%%%%%%%%%%%%%%%%%%%%%%%%%%%%%%%%%%%%%%%%%%%%%%%%%%%%%%%%%%%%%%%%%%%%%%%%%%%%%%%%%%%%%%%%%%%%%%%%%%%%%%%%%%%%%%%%%%%%%%%%%%%%%%%%%%%%%%%%%%%%%%%%%%%%%%%%%
%%%%%%%%%%%%%%%%%%%%%%%%%%%%%%%%%%%%%%%%%%%%%%%%%%%%%%%%%%%%%%%%%%%%%%%%%%%%%%%%%%%%%%%%%%%%%%%%%%%%%%%%%%%%%%%%%%%%%%%%%%%%%%%%%%%%%%%%%%%%%%%%%%%%%%%%%%%
%%%%%%%%%%%%%%%%%%%%%%%%%%%%%%%%%%%%%%%%%%%%%%%%%%%%%%%%%%%%%%%%%%%%%%%%%%%%%%%%%%%%%%%%%%%%%%%%%%%%%%%%%%%%%%%%%%%%%%%%%%%%%%%%%%%%%%%%%%%%%%%%%%%%%%%%%%%
\section{Introduction}
%\raggedbottom
% The very first letter is a 2 line initial drop letter followed
% by the rest of the first word in caps.
% 
% form to use if the first word consists of a single letter:
% \IEEEPARstart{A}{demo} file is ....
% 
% form to use if you need the single drop letter followed by
% normal text (unknown if ever used by the IEEE):
% \IEEEPARstart{A}{}demo file is ....
% 
% Some journals put the first two words in caps:
% \IEEEPARstart{T}{his demo} file is ....
% 
% Here we have the typical use of a "T" for an initial drop letter
% and "HIS" in caps to complete the first word.
\IEEEPARstart{S}{ignal} denoising and pattern recognition of time-series data are widely used in many scientific fields, including physics, engineering, medicine, economics, acoustics, biology, and psychology. For example, specific signal patterns in the electroencephalogram (EEG) data, which are useful in clinical diagnosis and cognitive neuroscience, are challenging to detect and distinguish from artifacts. The traditional method for solving the signal denoising problem involves the use of linear time-invariant (LTI) filters. In machine learning, the conventional way to recognize patterns in time-series involves three steps: feature extraction, feature selection, and classification.

LTI filters are easy to implement and are also efficient, especially when the frequency band of the signal of interest is known. However, if the signal contains discontinuities, LTI filters over-smooth the region of discontinuities. In contrast, fast iterative methods such as total variation denoising (TVD) \cite{chambolle2004,rudin1992,figueiredo2006} preserve discontinuities or singular points and are suitable for piecewise-constant signals. Though TVD is fast and effective, it often exhibits staircase-like artifacts, especially in regions where the signals are locally approximated by higher order polynomials \cite{rodriguez2009,karahanoglu2011,bredies2010,hu2012}. Further, any abrupt change, such as discontinuities or spikes in the signal, spreads out over the whole frequency axis. As a result, the signal is no longer sparse in the frequency domain. Thus, sparsity-based methods such as compressed sensing \cite{chen2001,candes2011} with dictionary elements from an oversampled discrete Fourier transform (DFT) matrix cannot reconstruct the signal perfectly.

An alternative approach uses wavelets, which offer an orthogonal multiresolution representation of the signals and have several advantages over traditional Fourier methods in analyzing signals that contain discontinuities \cite{mallat1989}. Wavelet-based denoising developed using adaptive thresholding of wavelet coefficients can simultaneously denoise and preserve the singularity points of the signal \cite{donoho1994,donoho1995}. The main drawback of this approach is that it introduces pseudo-Gibbs artifacts at the singular points due to oscillations that are more local and of smaller amplitude near signal's discontinuities. Several prior methods addressed the pseudo-Gibbs phenomenon explicitly. These include wavelet transform modulus sum \cite{hsung1999}, wavelet-domain hidden Markov models \cite{crouse1998}, wavelet footprints \cite{dragotti2003}, and total variation-based wavelet denoising \cite{durand2003}. 

In addition to explicit wavelet-based techniques, denoising is also achieved by decomposing the underlying signal as the sum of two or more components which include a local polynomial signal or low-frequency signal together with a sparse or sparse-derivative signal, or both \cite{selesnick2012_poly,selesnick2014,selesnick2015,selesnic2017}. For example, the sparsity-assisted signal smoothening (SASS) algorithm \cite{selesnick2014,selesnick2015,selesnic2017} simultaneously combines LTI filtering and TVD to denoise a wide class of signals. The effectiveness of SASS is mainly due to the computationally efficient designing of zero-phase noncausal high-pass and low-pass recursive filters as banded matrices. {\dred The main purpose of the zero-phase property is to eliminate phase distortion introduced by causal linear time invariant filters. In other words, the zero-phase property denoises the signal and also preserves its shape.} However, using these recursive filters as matrices introduces three challenges: a)~filter response types are limited to low-pass and high-pass filters; b)~the orders of the filter numerator and denominator polynomials must be equal; and c)~filters with polynomial orders larger than six are highly unstable because the banded matrices are no longer invertible, thus limiting the steepness of the transition bands.

To address these shortcomings, we develop a novel approach to designing zero-phase noncausal filters as matrices. Our method is inspired by spectral transformation of the state-space representation of digital filters \cite{mullis1976_assp,constantinides1970} and by forward-backward filtering \cite{gustafsson1996}. The spectral transformation property expands the filter response types to include low-pass, high-pass, and band-pass filters. Moreover, these filters do not require the orders of the numerator and denominator polynomials to be equal. Furthermore, the forward-backward filtering approach to designing zero-phase noncausal filters as matrices does not require a matrix inversion step, thereby making it feasible to design filters of higher orders. The maximum achievable order of the filter depends only on the positive definiteness condition of the reachability and observability Gramians obtained from the state-space representation of the digital filter. In addition, we also develop a proximal gradient-based method to factorize a special class of zero-phase high-pass and band-pass digital filters which contain at least one zero at $z = 1$, so that the factorization product preserves the zero-phase property of the filters and also incorporates a sparse-derivative component of the input into the signal model. {\dred The key differences between the zero-phase filters designed in \cite{selesnick2014,selesnick2015,selesnic2017} and our method are as follows: a) the zero-phase filters designed as matrices are stable and not sparse, b) the orders of the filter depend on the positive definiteness condition of the reachability and observability Gramians, and c) filter response types include stable low-pass, high-pass, and band-pass filters. Because our zero-phase filters are not sparse, they are computationally expensive when compared with the sparse and banded zero-phase filters designed in \cite{selesnick2014,selesnick2015,selesnic2017}. However, our zero-phase filters are stable, and the filter response type includes narrow band-pass filters which enables the development of new signal models for pattern recognition.}

{\dred To demonstrate applications of our novel filter designs, we validate and propose new signal models to simultaneously denoise and identify patterns of interest. We use our proposed filter design to test an existing signal model that simultaneously combines linear time invariant (LTI) filters and sparsity-based methods \cite{selesnic2017}. We develop a new signal model called sparsity-assisted signal denoising (SASD) by combining our proposed filter designs with the existing signal model. Because the zero-phase filters in the SASD signal model are stable, they demonstrate consistent results on changing the orders of the filter. Thereafter, we propose and derive a new signal model called sparsity-assisted pattern recognition (SAPR). In SAPR, we combine LTI band-pass filters and sparsity-based methods with orthogonal multiresolution representations, such as wavelets, to detect specific patterns in the input signal. Finally, we combine the signal denoising and pattern recognition tasks, and derive a new signal model called the sparsity-assisted signal denoising and pattern recognition (SASDPR). In SAPR and SASDPR, we use the prior knowledge of the pattern of interest and design narrow zero-phase band pass filter so the pass band of the filter response covers the frequency band of the pattern of interest. The optimization framework in SAPR and SASDPR are analogous to a pattern recognition problem in machine learning. However, the three key tasks, i.e., feature extraction (orthogonal multiresolution representation), feature selection (via sparsity-inducing norms), and classification (zero-phase band-pass filtering), happen simultaneously. 

To demonstrate the capabilities of the SAPR and SASDPR, we provide an illustrative example of detecting K-complexes and sleep spindles, respectively, in sleep-EEG data. K-complexes, sleep spindles, and slow-wave sleep constitute physiological markers of non-rapid eye movement (NREM) sleep. Recent findings suggest that there exists a bidirectional relationship between NREM sleep and amyloid-beta pathophysiology that may contribute to Alzheimer disease (AD) \cite{steriade1998,mander2016,ju2013,ju2017}. To enhance the understanding of how NREM sleep affects AD pathophysiology, it is necessary to develop accurate methods to automatically detect NREM and other EEG features in large EEG datasets with linked phenotypic measurements and AD biomarker characterization. Our proposed method of detecting K-complexes and sleep spindles demonstrate an improved performance relative to the existing method \cite{parekh2015}.}

Sections \ref{sec:state_space_rep} and \ref{sec:spec_transform} review the concepts required for state-space representation and spectral transformation of digital filters, respectively. Section \ref{sec:zero_phase_mats} presents a novel approach for designing higher order zero-phase low-pass, high-pass, and band-pass digital filters as matrices. Section \ref{sec:factorize_zero_phase_mats} develops a proximal gradient-based algorithm to factorize zero-phase high-pass and band-pass filters so that the factorization procedure preserves the zero-phase property of the filters and also enables the incorporation of discontinuities in the signal as a sparse derivative. {\dred Sections \ref{subsec:sig_denoise}, \ref{subsec:pattern_reco}, and \ref{sec:sig_sdpr} formulate the problems of signal denoising and pattern recognition as a convex optimization problem, and derive iterative procedures to solve them.}

\textit{Notations}: The following general notation will be used throughout the paper. Bold uppercase and lowercase letters denote a matrix and vector, respectively. Uppercase letters that are not bold denote scalars. For any matrix $\bm{A}$, $\bm{A}^\tpose$, $\bm{A}^{-1}$, and $\mathrm{Tr}\left\{\bm{A}\right\}$, denote the transpose, inverse, and trace of $\bm{A}$, respectively. $\bm{I}_N$ represents an $N\times N$ identity matrix. The norms $\norm{\cdot}_2$, $\norm{\cdot}_1$, and $\norm{\cdot}_\mathrm{F}$ indicate the $\ell_2$, $\ell_1$, and Frobenius norms, respectively. The vectorization of matrix $\bm{A}$ (column-wise unfolding of the matrix) is represented as $\mathrm{vec}(\bm{A})$.
%%%%%%%%%%%%%%%%%%%%%%%%%%%%%%%%%%%%%%%%%%%%%%%%%%%%%%%%%%%%%%%%%%%%%%%%%%%%%%%%%%%%%%%%%%%%%%%%%%%%%%%%%%%%%%%%%%%%%%%%%%%%%%%%%%%%%%%%%%%%%%%%%%%%%%%%%%%
%%%%%%%%%%%%%%%%%%%%%%%%%%%%%%%%%%%%%%%%%%%%%%%%%%%%%%%%%%%%%%%%%%%%%%%%%%%%%%%%%%%%%%%%%%%%%%%%%%%%%%%%%%%%%%%%%%%%%%%%%%%%%%%%%%%%%%%%%%%%%%%%%%%%%%%%%%%
%%%%%%%%%%%%%%%%%%%%%%%%%%%%%%%%%%%%%%%%%%%%%%%%%%%%%%%%%%%%%%%%%%%%%%%%%%%%%%%%%%%%%%%%%%%%%%%%%%%%%%%%%%%%%%%%%%%%%%%%%%%%%%%%%%%%%%%%%%%%%%%%%%%%%%%%%%%
\section{Preliminaries\label{sec:prelim}}
%%%%%%%%%%%%%%%%%%%%%%%%%%%%%%%%%%%%%%%%%%%%%%%%%%%%%%%%%%%%%%%%%%%%%%%%%%%%%%%%%%%%%%%%%%%%%%%%%%%%%%%%%%%%%%%%%%%%%%%%%%%%%%%%%%%%%%%%%%%%%%%%%%%%%%%%%%%
\subsection{State-Space Representation \label{sec:state_space_rep}}
%%%%%%%%%%%%%%%%%%%%%%%%%%%%%%%%%%%%%%%%%%%%%%%%%%%%%%%%%%%%%%%%%%%%%%%%%%%%%%%%%%%%%%%%%%%%%%%%%%%%%%%%%%%%%%%%%%%%%%%%%%%%%%%%%%%%%%%%%%%%%%%%%%%%%%%%%%%
Consider the $M$-th order stable discrete linear time invariant system transfer function
\begin{align}
H(z) = \frac{ \sum_{i=0}^M b_i z^{-i}}{1 +  \sum_{i=1}^M a_i z^{-i}} = \frac{B(z)}{A(z)}\label{eqn:tf},
\end{align}
where $a_i$ and $b_i$ are the filter coefficients of the numerator and denominator polynomials $B(z)$ and $A(z)$, respectively. The transfer function in (\ref{eqn:tf}) can be realized with completely controllable and observable state model equations in a recursive form as in \cite{moore1981}:
\begin{equation}
\begin{aligned}
\bm{s}(k+1) &= \bm{A}_{\mathrm{f}} \bm{s}(k) + \bm{B}_{\mathrm{f}} u(k) \\ \label{eqn:state_eqn}
y(k) 		&= \bm{C}_{\mathrm{f}} \bm{s}(k) + D_{\mathrm{f}} u(k), 
\end{aligned}
\end{equation}
where $u(k) \in \mathbb{R}$, $y(k) \in \mathbb{R}$, and $\bm{s}(k) \in \mathbb{R}^{M \times 1}$ denote the scalar input, scalar output, and state vector, respectively, and $\bm{A}_{\mathrm{f}} \in \mathbb{R}^{M \times M}$, $\bm{B}_{\mathrm{f}} \in \mathbb{R}^{M \times 1}$, and $\bm{C}_{\mathrm{f}} \in \mathbb{R}^{1 \times M}$, are real constant state matrices, and ${D}_{\mathrm{f}} \in \mathbb{R}^{1 \times 1}$ is a scalar. The recursive form of the digital filter in (\ref{eqn:state_eqn}) represents the forward filtering equations of the transfer function $H(z)$ and is denoted by the subscript ${\mathrm{f}}$. The transfer function $H(z)$ can be expressed in terms of the state matrices as
\begin{equation}
\begin{aligned}
H(z) = \sum_{k=0}^{\infty} h(k) z^{-k} = {D}_{\mathrm{f}} + \bm{C}_{\mathrm{f}}( z\bm{I}_M - \bm{A}_{\mathrm{f}} )^{-1} \bm{B}_{\mathrm{f}},
\end{aligned}
\end{equation}
with the correspondence
\begin{align}
h(k) = \begin{cases} D_{\mathrm{f}}, &\quad k = 0; \\ \bm{C}_{\mathrm{f}} \bm{A}_{\mathrm{f}}^{k-1} \bm{B}_{\mathrm{f}}, &\quad k = 1,2,\ldots. \end{cases}
\end{align}
Therefore, the filter $H(z)$ can also be represented in state-space as $(\bm{A}_{\mathrm{f}},\bm{B}_{\mathrm{f}},\bm{C}_{\mathrm{f}},D_{\mathrm{f}})$. If $\bm{u} = [u(0),u(1),\ldots,u(N-1)]^\tpose \in \mathbb{R}^{N \times 1}$ is the input vector and $\bm{s}(0)$ is the initial state vector, then the outputs of the filter, denoted by $\bm{y} = [y(0),y(1),\ldots,y(N-1)]^\tpose \in \mathbb{R}^{N \times 1}$, can be expressed as in \cite{gustafsson1996}:
\begin{align}
\bm{y} = \bm{H}_{\mathrm{f}} \bm{u} + \bm{O}_{\mathrm{f}} \bm{s}(0), \label{eqn:forward_filter}
\end{align}	
where $\bm{H}_{\mathrm{f}} \in \mathbb{R}^{N \times N}$ is a lower-triangular Toeplitz matrix of impulse response coefficients, expressed as

\begin{align}
\bm{H}_{\mathrm{f}} = \begin{bmatrix} D_{\mathrm{f}}  &   0    & \ldots &   0  \\
							  \bm{C}_{\mathrm{f}}\bm{B}_{\mathrm{f}}  &  D_{\mathrm{f}}  &    0   & \vdots \\
							\vdots  &        & \ddots &      \\
							 \bm{C}_{\mathrm{f}}\bm{A}_{\mathrm{f}}^{N-2}\bm{B}_{\mathrm{f}} & \ldots &  \bm{C}_{\mathrm{f}}\bm{B}_{\mathrm{f}}  &  D_{\mathrm{f}}
			  \end{bmatrix}, \label{eqn:toeplitz_matrix}
\end{align}
and $\bm{O}_{\mathrm{f}} \in \mathbb{R}^{N \times M}$ is the observability matrix. We assume that the reachability and observability Gramians denoted by $\bm{W}_{\mathrm{r}} \in \mathbb{R}^{M \times M}$ and $\bm{W}_{\mathrm{o}} \in \mathbb{R}^{M \times M}$, respectively, are positive definite, and satisfy the algebraic Lyapunov equations given as in \cite{anderson1979}:
\begin{equation}
\begin{aligned}
\bm{W}_{\mathrm{r}} = \bm{A}_{\mathrm{f}} \bm{W}_{\mathrm{r}} \bm{A}_{\mathrm{f}}^\tpose + \bm{B}_{\mathrm{f}}\bm{B}_{\mathrm{f}}^\tpose, \mbox{ and } \bm{W}_{\mathrm{o}} = \bm{A}_{\mathrm{f}}^\tpose \bm{W}_{\mathrm{o}} \bm{A}_{\mathrm{f}} + \bm{C}_{\mathrm{f}}^\tpose \bm{C}_{\mathrm{f}}. \notag
\end{aligned}
\end{equation}
Let $\bm{T} \in \mathbb{R}^{M \times M}$ be a nonsingular matrix. It is well known that the transfer function in (\ref{eqn:tf}) is invariant under nonsingular transformations \cite{snelgrove1986}. Under the change of variables $\hat{\bm{x}}(k) = \bm{T}^{-1} \bm{x}(k)$, the parameterization of the state-variable can be written as
\begin{equation}
\begin{aligned}
\hat{\bm{A}_{\mathrm{f}}} = \bm{T}^{-1} \bm{A}_{\mathrm{f}} \bm{T}, \quad \hat{\bm{B}_{\mathrm{f}}} = \bm{T}^{-1} \bm{B}_{\mathrm{f}}, \quad \hat{\bm{C}_{\mathrm{f}}} = \bm{C}_{\mathrm{f}} \bm{T}, \label{eqn:tranform_tf}
\end{aligned}
\end{equation}
and the reachability and observability Gramians are
\begin{equation}
\hat{\bm{W}}_{\mathrm{r}} = \bm{T}^{-1} \bm{W}_{\mathrm{r}} \bm{T}^{-\tpose} \mbox{ and } \hat{\bm{W}}_{\mathrm{o}} = \bm{T}^\tpose \bm{W}_{\mathrm{o}} \bm{T},  \\ \label{eqn:transform_gramians}
\end{equation}
respectively. Such transformations are called \emph{similarity} transformations. Under similarity transformations, the transfer function $H(z)$ remains the same and is expressed in a different coordinate system. Further, the eigenvalues of an asymptotically stable system (or modes) are invariant, but the eigenvalues of the Gramians are not invariant. However, the eigenvalues of the product of the Gramian matrices are invariant because $\hat{\bm{W}}_{\mathrm{r}} \hat{\bm{W}}_{\mathrm{o}} = \bm{T}^{-1} {\bm{W}}_{\mathrm{r}} {\bm{W}}_{\mathrm{o}} \bm{T}$. 

Let ${\boldsymbol\varSigma} = \diag\{{\sigma}^2_1,\ldots,{\sigma}^2_M\} \in \mathbb{R}^{M \times M}$ denote a diagonal matrix whose diagonal entries are the eigenvalues of $\bm{W}_{\mathrm{r}} \bm{W}_{\mathrm{o}}$. A transformation $\bm{T}$ for which $\hat{\bm{W}}_{\mathrm{r}}$ and $\hat{\bm{W}}_{\mathrm{o}}$ in (\ref{eqn:transform_gramians}) are diagonal is called a \emph{principal axis realization} \cite{mullis1976_cs} or \emph{contragredient} transformation \cite{laub1987}. A special case of the principal axis realization transformations where $\hat{\bm{W}}_{\mathrm{r}} = \hat{\bm{W}}_{\mathrm{o}} = {\boldsymbol\varSigma}^{1/2}$ is called an \emph{internally balanced} transformation \cite{moore1981}. The necessary and sufficient conditions to obtain such transformations are derived in \cite[Proposition~10]{moore1981}, \cite[Theorem 1]{mullis1976_cs}, and \cite[Theorem~1]{laub1987}. The details of the algorithm to obtain an internally balanced transformation are presented in \cite[Section II]{laub1987}. In addition, the balanced realization with $\hat{\bm{W}}_{\mathrm{r}} = \hat{\bm{W}}_{\mathrm{o}} = {\boldsymbol\varSigma}^{1/2}$ has the minimum sensitivity to noise and thus is recommended as a starting realization if numerical algorithms are applied to rational functions \cite{hwang1977}\cite{thiele1986}.
%%%%%%%%%%%%%%%%%%%%%%%%%%%%%%%%%%%%%%%%%%%%%%%%%%%%%%%%%%%%%%%%%%%%%%%%%%%%%%%%%%%%%%%%%%%%%%%%%%%%%%%%%%%%%%%%%%%%%%%%%%%%%%%%%%%%%%%%%%%%%%%%%%%%%%%%%%%
%%%%%%%%%%%%%%%%%%%%%%%%%%%%%%%%%%%%%%%%%%%%%%%%%%%%%%%%%%%%%%%%%%%%%%%%%%%%%%%%%%%%%%%%%%%%%%%%%%%%%%%%%%%%%%%%%%%%%%%%%%%%%%%%%%%%%%%%%%%%%%%%%%%%%%%%%%%
\subsection{Spectral Transformations for Digital Filters \label{sec:spec_transform}}
%%%%%%%%%%%%%%%%%%%%%%%%%%%%%%%%%%%%%%%%%%%%%%%%%%%%%%%%%%%%%%%%%%%%%%%%%%%%%%%%%%%%%%%%%%%%%%%%%%%%%%%%%%%%%%%%%%%%%%%%%%%%%%%%%%%%%%%%%%%%%%%%%%%%%%%%%%%
Spectral transformation \cite{constantinides1970} provides a useful technique to construct low-pass, high-pass, band-pass, and band-stop filters. Given a prototype stable digital filter with a real rational transfer function $H(z)$ (preferably a low-pass filter), one constructs a composite transfer function of the form 
\begin{align}
G(z) = H(F(z)) = H(z)|_{z^{-1} \gets 1/F(z)},
\end{align}
where
\begin{align}
\frac{1}{F(z)} &= \pm \prod_{i=1}^{L} \left( \frac{1 - \alpha_i z}{z - \bar{\alpha}_i} \right) = \pm z^{-L} \frac{\pi(z^{-1})}{\pi(z)}. \label{eqn:unit_func}
\end{align}
Here, $\bar{\alpha}_i$ is the complex conjugate of $\alpha_i$, $|\alpha_i| < 1$, $\pi(z)$ is an $L$-th order polynomial in $z$, and the order of the filter is $G(z)$ is $LM$. The functions in (\ref{eqn:unit_func}) are called \emph{unit functions} \cite{constantinides1970}. Note that the unit function in (\ref{eqn:unit_func}), also represents an $L$-th order all-pass filter. The transformation to obtain the composite transfer function $G(z)$ involves substitution of $z^{-1}$ in $H(z)$ with the unit function in (\ref{eqn:unit_func}) where the mapping $z \mapsto F(z)$ is a mapping of the unit circle onto itself. Therefore, the regions of stability and instability of $H(z)$ are preserved in $G(z)$. The choice of the unit function depends on the frequency response of the composite filter.

Let $(\boldsymbol\alpha_{\mathrm{f}}, \boldsymbol\beta_{\mathrm{f}}, \boldsymbol\gamma_{\mathrm{f}}, \delta_{\mathrm{f}})$ denote the state-space representation of $1/F(z)$. Then, the transfer function $G(z)$ is
\begin{align}
G(z) = \mathcal{D}_{\mathrm{f}} + \bm{\mathcal{C}}_{\mathrm{f}}(z \bm{I}_{LM} - \bm{\mathcal{A}}_{\mathrm{f}})^{-1} \bm{\mathcal{B}}_{\mathrm{f}}, \label{eqn:spec_tf_1}
\end{align}
where
\begin{equation}
\begin{aligned} 
\bm{\mathcal{A}}_{\mathrm{f}} &= \bm{I}_M \otimes \boldsymbol\alpha_{\mathrm{f}} + \left[ \bm{A}_{\mathrm{f}} (\bm{I}_M - \delta_{\mathrm{f}})^{-1} \right] \otimes (\boldsymbol\beta_{\mathrm{f}} \boldsymbol\gamma_{\mathrm{f}}),  \\
\bm{\mathcal{B}}_{\mathrm{f}} &= \left[ (\bm{I}_M - \delta_{\mathrm{f}} \bm{\mathcal{A}}_{\mathrm{f}})^{-1} \bm{B}_{\mathrm{f}} \right] \otimes \boldsymbol\beta_{\mathrm{f}}, \\
\bm{\mathcal{C}}_{\mathrm{f}} &= \left[ \bm{C}_{\mathrm{f}}(\bm{I}_M - \delta_{\mathrm{f}} \bm{A}_{\mathrm{f}})^{-1} \right] \otimes \boldsymbol\gamma_{\mathrm{f}}, \\
\mathcal{D}_{\mathrm{f}}      &= D_{\mathrm{f}} + \delta_{\mathrm{f}} \bm{C}_{\mathrm{f}} (\bm{I}_M - \delta_{\mathrm{f}} \bm{A}_{\mathrm{f}})^{-1} \bm{B}_{\mathrm{f}}, \label{eqn:spec_tf_2}
\end{aligned}
\end{equation}
where $\otimes$ is the Kronecker product and $(\bm{\mathcal{A}}_{\mathrm{f}},\bm{\mathcal{B}}_{\mathrm{f}},\bm{\mathcal{C}}_{\mathrm{f}},\mathcal{D}_{\mathrm{f}})$ is the state-space representation of the transfer function $G(z)$ \cite[Lemma 1]{mullis1976_assp}. The details of the derivations of equations (\ref{eqn:spec_tf_1}) and (\ref{eqn:spec_tf_2}) are presented in \cite[Appendix A]{mullis1976_assp}. Let $\bm{\mathcal{W}}_{\mathrm{r}}$ and $\bm{\mathcal{W}}_{\mathrm{o}}$ denote the reachability and observability Gramian matrices, respectively, of the composite filter $G(z)$. Then, these Gramian matrices can be derived using the algebraic Lyapunov equations
\begin{equation}
\begin{aligned}
\bm{\mathcal{W}}_{\mathrm{r}} = \bm{\mathcal{A}}_{\mathrm{f}} \bm{\mathcal{W}}_{\mathrm{r}} \bm{\mathcal{A}}_{\mathrm{f}}^\tpose + \bm{\mathcal{B}}_{\mathrm{f}} \bm{\mathcal{B}}_{\mathrm{f}}^\tpose, \mbox{ and } \bm{\mathcal{W}}_{\mathrm{o}} = \bm{\mathcal{A}}_{\mathrm{f}}^\tpose \bm{\mathcal{W}}_{\mathrm{o}} \bm{\mathcal{A}}_{\mathrm{f}} + \bm{\mathcal{C}}_{\mathrm{f}}^\tpose \bm{\mathcal{C}}_{\mathrm{f}}. \notag
\end{aligned}
\end{equation}
The relationship between the Gramians of the two filters, $H(z)$ and the composite filter $G(z)$, is expressed as
\begin{equation}
\begin{aligned}
\bm{\mathcal{W}}_{\mathrm{r}} = \bm{W}_{\mathrm{r}} \otimes \bm{Q}, \mbox{ and } \bm{\mathcal{W}}_{\mathrm{o}} = \bm{W}_{\mathrm{o}} \otimes \bm{Q}^{-1},
\end{aligned}
\end{equation}
where $\bm{Q} \in \mathbb{R}^{L \times L}$ is the positive definite matrix \cite[Lemma 3]{mullis1976_assp}. Matrices $\bm{Q}$ and $\bm{Q}^{-1}$ are the reachability and observability matrices for the all-pass filter $1/F(z)$. Because the product of the Gramian matrices of a stable all-pass filter $1/F(z)$ is an Identity matrix, the second-order modes are all unity \cite[Corollary 1]{mullis1976_assp}. In addition, if $1/F(z)$ is internally balanced, then $\bm{Q} = \bm{I}_L$. Therefore, if $H(z)$ is a stable filter of order $M$ and $1/F(z)$ is a stable all-pass filter of order $L$, then the $LM$ second-order modes of the composite filter $G(z) = H(F(z))$ are simply $L$ copies of the $M$ second-order modes of $H(z)$ \cite[Theorem 2]{mullis1976_assp}. Moreover, if $1/F(z)$ is an internally balanced all-pass filter of order $L$, i.e., $\bm{Q} = \bm{I}_L$, and $\bm{T}$ be a transformation such that $H(z)$ is any internally balanced filter of order $M$, then the composite filter $G(z) = H(F(z))$ of order $LM$ is also internally balanced \cite{mullis1976_assp,koshita2008gramian}. The proof is simple and straightforward. If $H(z)$ and $1/F(z)$ are internally balanced, then the reachability and observability Gramians are $\boldsymbol\varSigma^{1/2}$ and $\bm{I}_L$, respectively. Therefore, $\bm{\mathcal{W}}_{\mathrm{r}} = \bm{\mathcal{W}}_{\mathrm{o}} = (\boldsymbol\varSigma^{1/2} \otimes \bm{I}_L)$. As the Gramians of the composite filter $G(z)$ are equal and are diagonal, $G(z)$ is automatically internally balanced, and thus demonstrates minimum sensitivity to noise.
%%%%%%%%%%%%%%%%%%%%%%%%%%%%%%%%%%%%%%%%%%%%%%%%%%%%%%%%%%%%%%%%%%%%%%%%%%%%%%%%%%%%%%%%%%%%%%%%%%%%%%%%%%%%%%%%%%%%%%%%%%%%%%%%%%%%%%%%%%%%%%%%%%%%%%%%%%%
%%%%%%%%%%%%%%%%%%%%%%%%%%%%%%%%%%%%%%%%%%%%%%%%%%%%%%%%%%%%%%%%%%%%%%%%%%%%%%%%%%%%%%%%%%%%%%%%%%%%%%%%%%%%%%%%%%%%%%%%%%%%%%%%%%%%%%%%%%%%%%%%%%%%%%%%%%%
%%%%%%%%%%%%%%%%%%%%%%%%%%%%%%%%%%%%%%%%%%%%%%%%%%%%%%%%%%%%%%%%%%%%%%%%%%%%%%%%%%%%%%%%%%%%%%%%%%%%%%%%%%%%%%%%%%%%%%%%%%%%%%%%%%%%%%%%%%%%%%%%%%%%%%%%%%%
\section{Infinite Impulse Response Filters as Matrices \label{sec:filt_as_mats}}
In this section, we present a novel approach to designing higher-order zero-phase low-pass, high-pass, and band-pass filters as matrices, using spectral transformation of the state-space representation of digital filters and forward-backward filtering. We also propose a proximal gradient-based method to factorize a special class of zero-phase high-pass and band-pass digital filters that contain at least one zero at $z = 1$. The factorization product almost completely preserves the zero-phase property of the filters and also incorporates any discontinuities, in the signal which are modeled as a sparse-derivative signal.
\begin{figure*}[t!]
\centering
\subfloat[Noncausal zero-phase fourth-order low-pass filter with cut-off frequency $\omega_c = 0.2\pi$. \label{fig:1a_lp_nczp}]{%
    \includegraphics[width=0.32\linewidth]{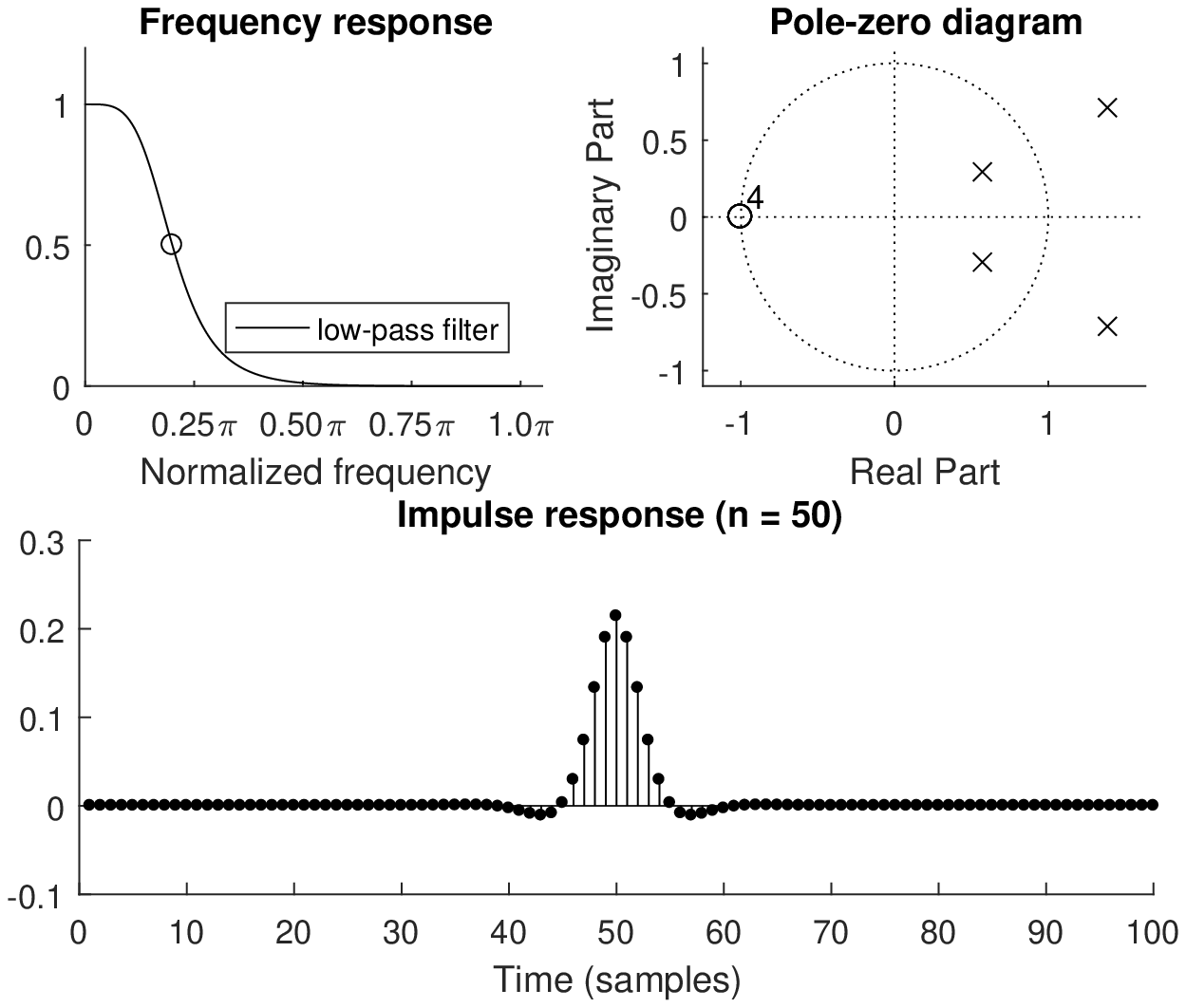}}
   \hfill
\subfloat[Noncausal zero-phase fourth-order high-pass filter with cut-off frequency $\omega_c = 0.2\pi$. \label{fig:1b_lp_nczp}]{%
     \includegraphics[width=0.32\linewidth]{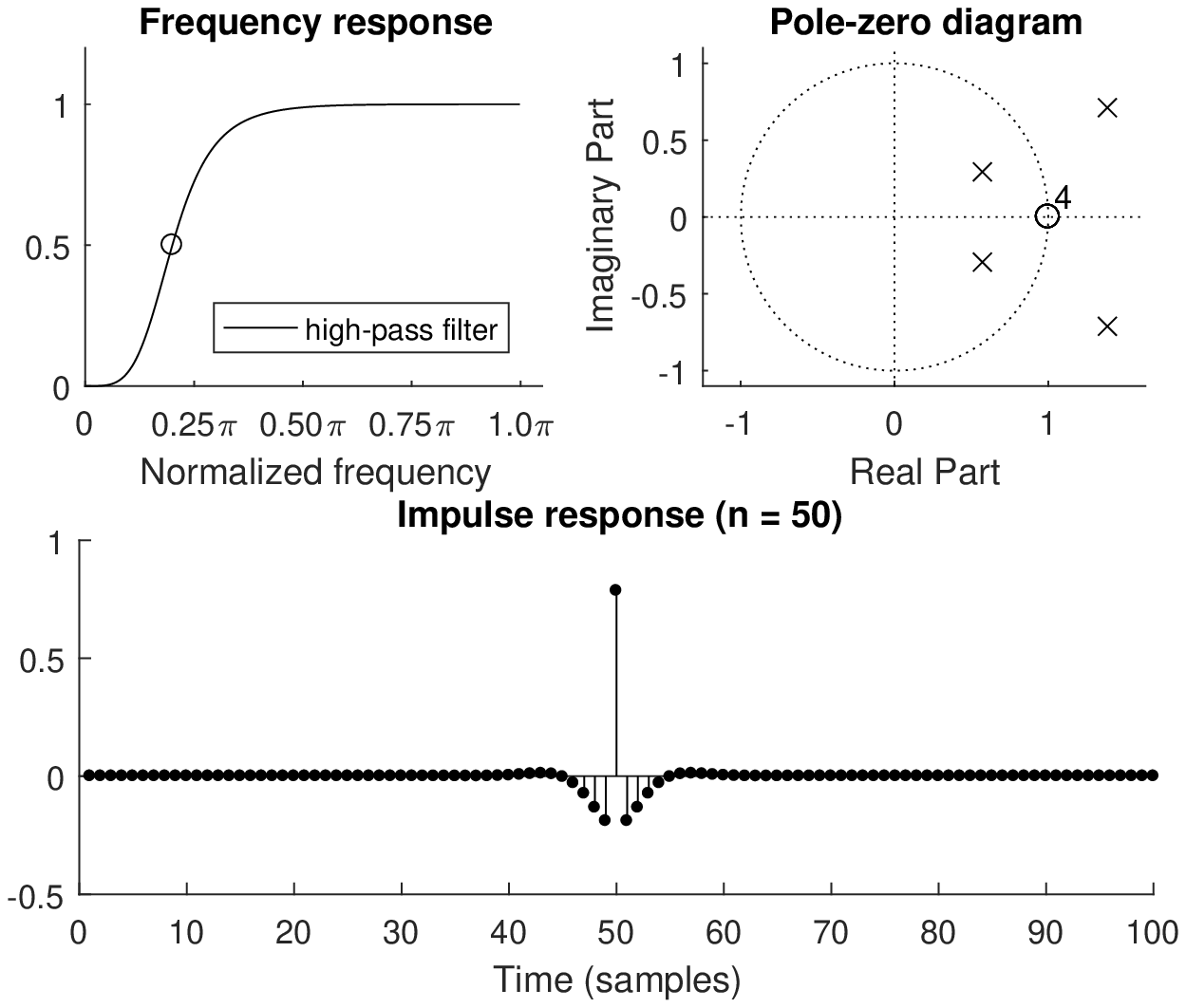}}
     \hfill
\subfloat[Noncausal zero-phase eighth-order band-pass filter $G(z)$ with center frequency $\omega_n = 0.5\pi$ and bandwith $\omega_b = 0.1\pi$. \label{fig:1c_lp_nczp}]{%
     \includegraphics[width=0.32\linewidth]{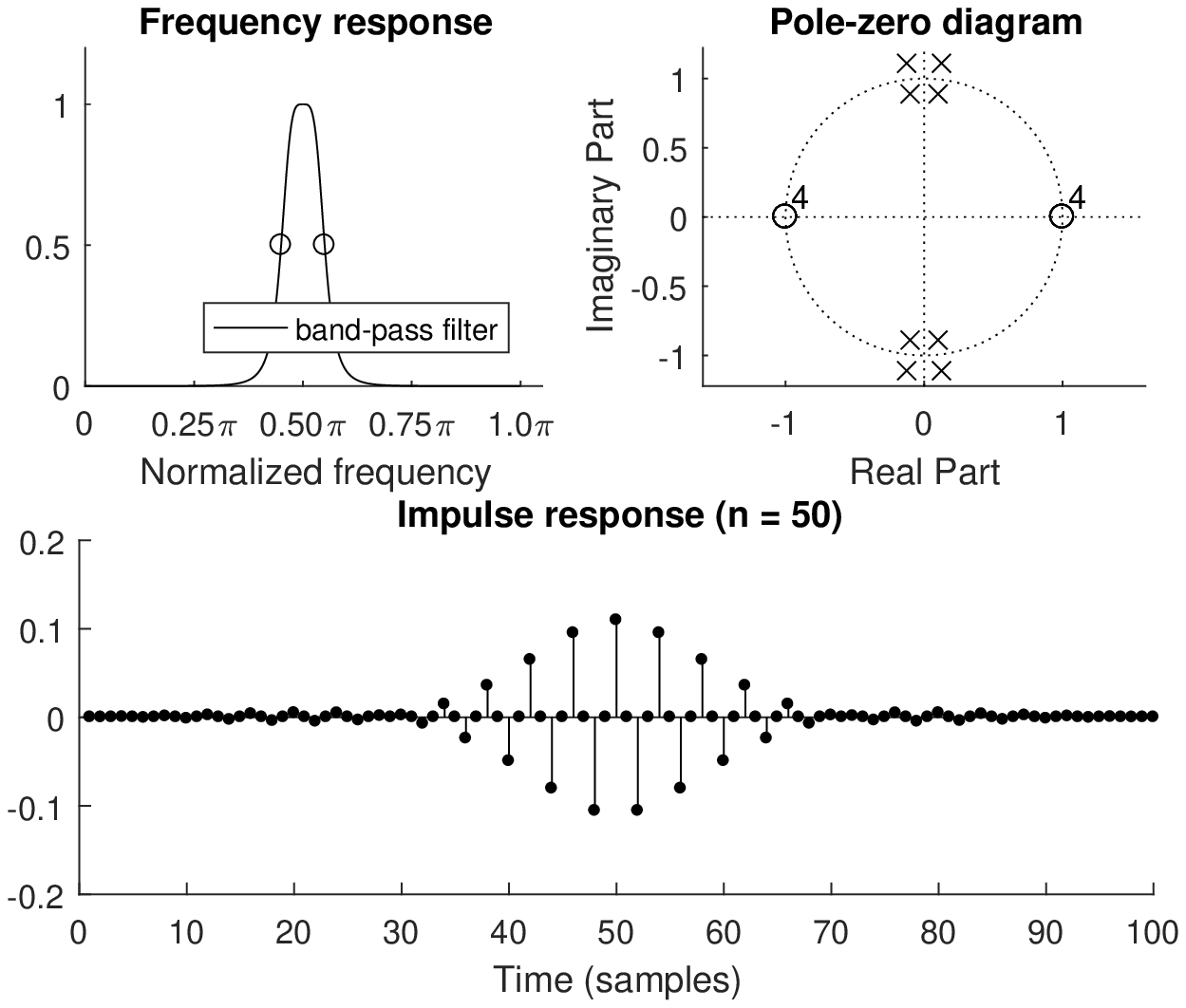}}
     \hfill
\caption{The composite filter $G(z)$ is designed using a prototype generalized digital low-pass Butterworth filter $H(z)$ with cut-off frequency at $\omega_c = 0.1\pi$ and filter order $M = 2$. The half-power point or $3 \mathrm{dB}$ points are denoted with circle in the frequency response plots. Poles are denoted with crosses whereas zeros are denoted with circles.\label{fig:non_causal_zero_phase}}
\end{figure*}
%%%%%%%%%%%%%%%%%%%%%%%%%%%%%%%%%%%%%%%%%%%%%%%%%%%%%%%%%%%%%%%%%%%%%%%%%%%%%%%%%%%%%%%%%%%%%%%%%%%%%%%%%%%%%%%%%%%%%%%%%%%%%%%%%%%%%%%%%%%%%%%%%%%%%%%%%%%
\subsection{Zero-Phase Filters as Matrices \label{sec:zero_phase_mats}}
%%%%%%%%%%%%%%%%%%%%%%%%%%%%%%%%%%%%%%%%%%%%%%%%%%%%%%%%%%%%%%%%%%%%%%%%%%%%%%%%%%%%%%%%%%%%%%%%%%%%%%%%%%%%%%%%%%%%%%%%%%%%%%%%%%%%%%%%%%%%%%%%%%%%%%%%%%%
\begin{prop}\label{lemma_1}
An $M$-th order prototype low-pass filter $H(z)$ with $M_1$ zeros at $z = -1$ and cut-off frequency $\omega_0$ can be spectrally transformed to
\begin{enumerate}[label=(\alph*)]
\item a composite low-pass filter $G(z) = H(F_{\mathrm{LP}}(z))$ with $M_1$ zeros at $z = -1$ and cut-off frequency $\omega_1$ where
\begin{align}
\frac{1}{F_{\mathrm{LP}}(z)} &= \frac{z^{-1} + \xi_{\mathrm{LP}}}{1 + \xi_{\mathrm{LP}} z^{^-1}}, \mbox{ and } \xi_{\mathrm{LP}} = \frac{\sin(\frac{\omega_0 - \omega_1}{2})}{\sin(\frac{\omega_0 + \omega_1}{2})}. \notag
\end{align}
\item a composite high-pass filter $G(z) = H(F_{\mathrm{HP}}(z))$ with $M_1$ zeros at $z = 1$ and cut-off frequency $\omega_1 $ where
\begin{align}
\frac{1}{F_{\mathrm{HP}}(z)} &= - \frac{z^{-1} + \xi_{\mathrm{HP}}}{1 + \xi_{\mathrm{HP}} z^{-1}}, \mbox{ and } \xi_{\mathrm{HP}} = \frac{\cos(\frac{\omega_0 + \omega_1}{2})}{\cos(\frac{\omega_0 - \omega_1}{2})}. \notag
\end{align}
\item a composite band-pass filter $G(z) = H(F_{\mathrm{BP}}(z))$ with $M_1$ zeros at $z = -1$ and $z = 1$, and center frequency $\omega_1$ where
\begin{align}
\frac{1}{F_{\mathrm{BP}}(z)} &= - \frac{z^{-1} (z^{-1} - \xi_{\mathrm{BP}})}{1 - \xi_{\mathrm{BP}} z^{-1}}, \mbox{ and } \xi_{\mathrm{BP}} = \cos(\omega_1). \notag
\end{align}
\end{enumerate}
% \begin{align}
% \frac{1}{F_{\mathrm{LP}}(z)} &= \frac{z^{-1} + \xi_{\mathrm{LP}}}{1 + \xi_{\mathrm{LP}} z^{^-1}}, \quad \xi_{\mathrm{LP}} = \frac{\sin(\frac{\omega_0 - \omega_1}{2})}{\sin(\frac{\omega_0 + \omega_1}{2})}, \notag \\
% \frac{1}{F_{\mathrm{HP}}(z)} &= - \frac{z^{-1} + \xi_{\mathrm{HP}}}{1 + \xi_{\mathrm{HP}} z^{-1}}, \quad \xi_{\mathrm{HP}} = \frac{\cos(\frac{\omega_0 + \omega_1}{2})}{\cos(\frac{\omega_0 - \omega_1}{2})}, \mbox{ and} \notag \\
% \frac{1}{F_{\mathrm{BP}}(z)} &= - \frac{z^{-1} (z^{-1} - \xi_{\mathrm{BP}})}{1 - \xi_{\mathrm{BP}} z^{-1}}, \quad \xi_{\mathrm{BP}} = \cos(\omega_1). \notag
% \end{align}
\end{prop}
\begin{proof}
See Appendix \ref{appx_prop_1}. \renewcommand{\qedsymbol}{}
\end{proof}

Three different categories of digital IIR filters satisfy Proposition \ref{lemma_1}: generalized digital Butterworth filter \cite{selesnic1996}, Chebyshev Type-I filter, and Chebyshev Type-II filter. The most important property of these filters is that there exists only one possible way to divide the number of zeros between $z = -1$ and the passband. Generalized Butterworth filters are a class of digital filters which have maximally flat response in the passband, i.e., no ripples, and their frequency response rolls off towards zero in the passband \cite{selesnic1996}. The orders of the numerator polynomial and denominator polynomials of these filters need not be equal. As a result of Proposition \ref{lemma_1}, the composite filter obtained by transforming a generalized digital Butterworth low-pass filter preserves its flatness. Further, Chebyshev Type-I and Chebyshev Type-II digital filters with odd filter orders also satisfy Proposition \ref{lemma_1} because the numerator polynomial can be uniquely factorized as the product of two polynomials, with one of them representing the number of zeros at $z = -1$. However, a flat frequency response is observed only in either the pass-band or stop-band, not in both, unlike the case of generalized digital Butterworth filters.

Using the results of Proposition \ref{lemma_1} and (\ref{eqn:spec_tf_2}), we get an internally balanced state-space representation of the composite filter $G(z)$ denoted by $(\bm{\mathcal{A}}_{\mathrm{f}},\bm{\mathcal{B}}_{\mathrm{f}},\bm{\mathcal{C}}_{\mathrm{f}},{\mathcal{D}}_{\mathrm{f}})$. The lower-triangular Toeplitz impulse response matrix of the composite filter, denoted by $\bm{G}_{\mathrm{f}}$, is obtained using (\ref{eqn:toeplitz_matrix}). The impulse response matrix is LTI and causal, and thus introduces phase distortions into the filtered signal. To avoid phase distortions introduced by the impulse response matrix $\bm{G}_{\mathrm{f}}$, we propose a simple approach to implement zero-phase noncausal digital filters as matrices. Our approach is based on forward-backward filtering \cite{gustafsson1996}. The resulting filter is zero-phase because it removes the phase distortions introduced by the filter. Further, the filter is noncausal because the backward filter depends on the future state vectors to compute the filter output.
\begin{prop}\label{lemma_2}
If the initial conditions of the forward and backward filter are set to zero, then $\bm{G}_{\mathrm{f}}^\tpose \bm{G}_{\mathrm{f}}$ represents a zero-phase impulse response matrix of the composite filter $G(z)$ with undesirable transients at the start and end of the signal. 
\end{prop}
\begin{proof}
See Appendix \ref{appx_prop_2}. \renewcommand{\qedsymbol}{}
\end{proof}
It is important to note that the orders of the zero-phase low-pass and high-pass filters are twice that of the composite filter because the overall filter response of the zero-phase filter implemented using the forward-backward filtering approach \cite{gustafsson1996} is $G(z) G(1/z)$.

\textit{Example:} We first begin by designing a prototype digital IIR Butterworth low-pass filter $H(z)$. We use the \texttt{maxflat} function in MATLAB \cite{selesnic1996}\cite{MATLAB:2016} to design a maximally flat  or classical Butterworth filter of order $M = 2$, i.e., the orders of the numerator and denominator polynomials are equal. The cut-off frequency $\omega_0 = 0.1 \pi$ {radians}/{sample}. Based on Proposition \ref{lemma_1}, the cut-off frequency represents the half-power point, or $3 \mathrm{dB}$ point, of the composite low-pass and high-pass filter, or the band-width of the composite band-pass filter. We convert the transfer function $H(z)$ into a state-space representation ($\bm{A}_{\mathrm{f}}, \bm{B}_{\mathrm{f}}, \bm{C}_{\mathrm{f}}, D_{\mathrm{f}}$) using the \texttt{tf2ss} command in MATLAB, and apply internal balanced transformation \cite[Section II]{laub1987}. Depending on the type of frequency response (low-pass/high-pass/band-pass) and the required operating frequency $\omega_1$, the function $1/F(z)$ is selected and balanced internally (see Proposition \ref{lemma_1}). Next, using the equations in (\ref{eqn:spec_tf_2}), we obtain the state-space representation of $G(z)$, which is denoted as $(\bm{\mathcal{A}}_{\mathrm{f}},\bm{\mathcal{B}}_{\mathrm{f}},\bm{\mathcal{C}}_{\mathrm{f}},\mathcal{D}_{\mathrm{f}})$. Note that $G(z)$ is already internally balanced because $H(z)$ and $1/F(z)$ are internally balanced. Finally, we obtain the zero-phase impulse response matrix using (\ref{eqn:toeplitz_matrix}) and Proposition \ref{lemma_2}. 

In Fig. \ref{fig:non_causal_zero_phase}, we demonstrate the results of Proposition \ref{lemma_1} and \ref{lemma_2} for a prototype digital low-pass Butterworth filter $H(z)$. We make three important observations based on the impulse response plots in Fig. \ref{fig:non_causal_zero_phase}. First, we note that the impulse response is almost symmetric and noncausal. We use the word ``almost'' because the response to an impulse $\delta(n - n_0)$ is not strictly symmetric because of the finite length \cite{selesnick2014} and non-optimal initial condition of the state-space digital filter. Second, we observe, because the magnitude of the frequency response of a zero-phase filter is square of the magnitude of the original filter, the transition bands are steeper and the half-power points are now half-magnitude points. As a result, the points on the transition band whose magnitude $\in (0,1)$ are squared, thereby making the transition band steeper. Further, because the pass-band and stop-band of the generalized digital Butterworth filter are flat, the magnitude response of the zero-phase filters remains the same. Finally, the order of the zero-phase low-pass and high-pass filters are twice the order of the composite filter $G(z)$. Hence, the numbers of zeros at $z = 1$ for the high-pass and band-pass filter are twice that of the composite filter.
%%%%%%%%%%%%%%%%%%%%%%%%%%%%%%%%%%%%%%%%%%%%%%%%%%%%%%%%%%%%%%%%%%%%%%%%%%%%%%%%%%%%%%%%%%%%%%%%%%%%%%%%%%%%%%%%%%%%%%%%%%%%%%%%%%%%%%%%%%%%%%%%%%%%%%%%%%%
\vspace{-1em}
\subsection{Factorization of a Zero-Phase Impulse Response Matrix\label{sec:factorize_zero_phase_mats}}
%%%%%%%%%%%%%%%%%%%%%%%%%%%%%%%%%%%%%%%%%%%%%%%%%%%%%%%%%%%%%%%%%%%%%%%%%%%%%%%%%%%%%%%%%%%%%%%%%%%%%%%%%%%%%%%%%%%%%%%%%%%%%%%%%%%%%%%%%%%%%%%%%%%%%%%%%%%
In this subsection, we address the problem of incorporating discontinuities in the input signal as $K$-order sparse derivatives. When the input vector is multiplied with the zero-phase filter $\bm{G}_{\mathrm{f}}^\tpose \bm{G}_{\mathrm{f}}$, the resulting output is a filtered signal that depends on the frequency response of the filter. However, these filters over-smooth the discontinuities. In order to preserve the discontinuities, $\bm{G}_{\mathrm{f}}$ should be factorized as $\bm{G}_1 \bm{D}$, where $\bm{G}_1 \in \mathbb{R}^{N \times N-K}$ and $\bm{D} \in \mathbb{R}^{N-K \times N}$ is a $K$-order sparse derivative matrix. For example, if $K = 1$, then $\bm{D}$ can be written as
\begin{align}
\bm{D} = \begin{bmatrix} -1 &  1 &   & & \\
                            & -1 & 1 & & \\
                            &    & \ddots & \ddots & \\
                            &    &   & -1 & 1 
        \end{bmatrix}. \notag
\end{align}
In other words, the matrix $\bm{D}$, when multiplied with an input signal with discontinuities, outputs an $N-K \times 1$ sparse vector, whose non-zero elements represents the singularity points in the input signal.
\begin{table*}[t!]
\caption{Properties of IIR filters as Matrix Operators}\vspace{-0.5em}
\footnotesize
\centering
\label{tab:zero_phase_properties}
\def\arraystretch{1.0}
\begin{tabular}{l  p{12.75cm}}
\toprule
\toprule
Zero-Phase Filter Operator & Conditions \\ 
\midrule 
\midrule
\multirow{1}{*}{$\mathsf{LPF}_{\omega_1}( \bm{u} ) \triangleq \bm{L}^\tpose \bm{L} \bm{u}$} & $\bm{L}$ is the impulse response matrix of the composite low-pass filter $G(z)$ with cut-off frequency $\omega_1$. \\

\midrule 
\multirow{1}{*}{$\mathsf{HPF}_{\omega_1}( \bm{u} ) \triangleq \bm{H}^\tpose \bm{H} \bm{u} \approx \bm{H}^\tpose \bm{H}_1 \bm{D} \bm{u}$} & $\bm{H}$ is the impulse response matrix of the composite high-pass filter $G(z)$ with cut-off frequency $\omega_1$, $\bm{D}$ is the $K$-order sparse derivative matrix with $0 < K \le M_1$, and $\bm{H}_1$ is a factor of $\bm{H}$ obtained by solving (\ref{eqn:matrix_factor_opt_2}).\\  

\midrule 
\multirow{1}{*}{$\mathsf{HPF}_{\omega_1}( \bm{u} ) \triangleq  \{ \mathsf{I} - \mathsf{LPF}_{\omega_1} \} ( \bm{u} )$} & $\mathsf{I}$ is an Identity matrix of size $N\times N$. The degrees of the numerator and denominator polynomials of the composite filter $G(z)$ are equal and developed from the same prototype classical Butterworth low-pass filter \cite{selesnick2014,selesnick2015,selesnic2017}. \\

\midrule 
\multirow{1}{*}{$\mathsf{BPF}_{\omega_1}^{\omega_2} ( \bm{u} ) \triangleq \bm{B}^\tpose \bm{B} \bm{u} \approx \bm{B}^\tpose \bm{B}_1 \bm{D} \bm{u}$} & $\bm{B}$ is the impulse response matrix of the composite band-pass filter $G(z)$ with pass-band $(\omega_1,\omega_2)$, $\omega_1 < \omega_2$, and $\bm{B}_1$ is a factor of $\bm{B}$ obtained by solving (\ref{eqn:matrix_factor_opt_2}). \\ 

%\multirow{1}{*}{$\mathsf{HPF}_{\omega_0} \{ \mathsf{BPF}_{\omega_1}^{\omega_2} \} ( \bm{u} ) \approx \mathsf{BPF}_{\omega_1}^{\omega_2} ( \bm{u} )$} &  The product of a zero-phase band-pass filter with a zero-phase high-pass filter is approximately a zero-phase band-pass filter when pass-band of the band-pass filter lies in the pass-band of the high-pass filter, i.e., $\omega_0 \le \omega_1 < \omega_2$. \\
\bottomrule
\bottomrule
\end{tabular}
\end{table*}

Based on Proposition \ref{lemma_1}, the numerator polynomial of the composite high-pass or band-pass filters can be factorized as a product of two polynomial functions $B'_1(z) B'_2(z)$, where $B'_1(z)$ represents the number of zeros at $z = 1$. Because the maximum number of zeros at $z = 1$ is $M_1$, we can further factorize $B'_1(z)$ as $(1-z^{-1})^{M_1-K} (1-z^{-1})^{K}$, where $0 < K \le M_1$. Let $D(z) = (1-z^{-1})^{K}$, and then the transfer function of the composite filter can be rewritten as
\begin{align}
G(z) &= \frac{B''_2(z) D(z)}{A'(z)} = G_1(z) D(z), \label{eqn:poly_factorization} %\frac{(1-z^{-1})^{M_1-K} B'_2(z) D(z)}{A'(z)}
\end{align}
where $B''_2(z)$ is obtained by deconvolving $B'_1(z)B'_2(z)$ with $D(z)$. However, we cannot factorize $\bm{G}_{\mathrm{f}}$ as $\bm{G}_1 \bm{D}$, where $\bm{G}_1 \in \mathbb{R}^{N \times N-K}$ is a lower-triangular Toeplitz impulse response matrix of $G_1(z)$ and $\bm{D} \in \mathbb{R}^{N-K \times N}$, because $\bm{G}_1$ represents a partial impulse response matrix with only $N-K$ columns. Thus, the resulting product  $\bm{G}_{\mathrm{f}}^\tpose \bm{G}_1 \bm{D}$ is no longer zero-phase. To overcome this, we solve for $\bm{G}_1$ in the least squares sense by imposing zero-phase property on the product of $\bm{G}^\tpose_{\mathrm{f}} \bm{G}_1 \bm{D}$. We minimize the following objective function:
\begin{align}
\begin{split}
\arg \min_{\bm{G}_1} &\quad \norm{\bm{G}_{\mathrm{f}}^{\tpose} \bm{G}_{\mathrm{f}} - \bm{G}_{\mathrm{f}}^{\tpose} \bm{G}_1 \bm{D} }_{\mathrm{F}}^2. \label{eqn:matrix_factor_opt}
%\mathrm{subject \ to}       &\quad \mathrm{tril}(\bm{G}_1), \label{eqn:matrix_factor_opt}
\end{split}
\end{align}
The optimization problem in (\ref{eqn:matrix_factor_opt}) is convex (see Appendix \ref{appx_a}) and has a closed-form solution. However, the closed-form solution requires the computation of the inverse of the lower-triangular Toeplitz matrix $\bm{G}_{\mathrm{f}}$, which is ill-conditioned, especially when $\mathcal{D}_\mathrm{f}$ (the diagonal element of the impulse response matrix $\bm{G}_\mathrm{f}$) is close to zero. Further, as the sample size increases, the closed-form solution requires computation and storage of the inverse of a large matrix of size $N^2 \times N^2$. To avoid these computationally expensive tasks, we propose an proximal gradient decent algorithm to solve (\ref{eqn:matrix_factor_opt}). We provide the details of our algorithm in Appendix \ref{appx_a}. Our proposed method is inspired by the fast iterative shrinkage-thresholding algorithm (FISTA) \cite{beck2009}. Note that the optimization problem in (\ref{eqn:matrix_factor_opt}) does not impose any constraint on $\bm{G}_1$. 

Ideally, $\bm{G}_1$ takes the form of a lower-triangular Toeplitz matrix structure, which can be imposed as a set of linear constraints in the optimization problem (\ref{eqn:matrix_factor_opt}). However, it will lead only to an overdetermined set of equations with tight constraints. Therefore, we relax the constraint so that the lower-triangular matrix structure of $\bm{G}_1$ is preserved and not the Toeplitz structure. Thus, the optimization problem in (\ref{eqn:matrix_factor_opt}) can be formulated as
\begin{align}
\begin{split}
\arg \min_{\bm{G}_1} &\quad \norm{\bm{G}_{\mathrm{f}}^{\tpose} \bm{G}_{\mathrm{f}} - \bm{G}_{\mathrm{f}}^{\tpose} \bm{G}_1 \bm{D} }_{\mathrm{F}}^2, \\
\mathrm{subject \ to}       &\quad \mathrm{tril}(\bm{G}_1), \label{eqn:matrix_factor_opt_2}
\end{split}
\end{align}
where $\mathrm{tril}$ applies a lower-triangular matrix constraint on $\bm{G}_1$. The lower-triangular matrix constraint in (\ref{eqn:matrix_factor_opt_2}) can also be formulated as a set of equality constraints. Thus, the optimization problem in (\ref{eqn:matrix_factor_opt_2}) is a quadratic program with linear equality constraints. Because no matrix inversion step is required in solving (\ref{eqn:matrix_factor_opt_2}), we can design higher order filters as long as we can compute Gramian preserving transformations.
\begin{table}[t!]
\caption{Performance Metrics of Zero-Phase Filters} \vspace{-0.5em}
\footnotesize
\centering
\label{tab:performance_metrics}
\setlength\tabcolsep{2.0pt}
\def\arraystretch{.90}
\begin{tabular}{c  c  c  c  c}
\toprule
\toprule
Length & Sparsity & Error & Filt. Norm & Filt. Norm \cite{selesnic2017} \\
$N$ & $K$ & $\norm{\bm{G}_{\mathrm{f}}^{\tpose} \bm{G}_{\mathrm{f}} - \bm{G}_{\mathrm{f}}^{\tpose} \bm{G}_1 \bm{D} }_{\mathrm{F}}^2$ & $\norm{\bm{G}_{\mathrm{f}}^{\tpose} \bm{G}_1 \bm{h}}$ & $\norm{\bm{A}^{-1} \bm{B}_1 \bm{h}}$ \\
\midrule
\midrule
\multirow{2}{*}{$N = 100$} & $K = 1$ & 0.0497 & 0.6384 & 0.6388 \\
						   & $K = 2$ & 0.2044 & 0.6515 & 0.6512 \\
\midrule
\multirow{2}{*}{$N = 500$} & $K = 1$ & 0.0389 & 0.6388 & 0.6388 \\
						   & $K = 2$ & 0.1992 & 0.6512 & 0.6512 \\
\midrule
\multirow{2}{*}{$N = 1000$} & $K = 1$ & 0.0389 & 0.6388 & 0.6388 \\
						    & $K = 2$ & 0.1992 & 0.6512 & 0.6512 \\
\bottomrule
\bottomrule
\end{tabular}
\end{table}
In Table \ref{tab:zero_phase_properties}, we summarize the properties of the zero-phase filters as matrix operators developed in this section. We denote low-pass filtering, high-pass filtering, and band-pass filtering, by $\mathsf{LPF}$, $\mathsf{HPF}$, and $\mathsf{BPF}$, respectively. The subscripts and superscripts indicate the half-power cut-off frequency points.

\textit{Example (Cont.):} In Table \ref{tab:performance_metrics}, we use the proposed filter designs in Section \ref{sec:filt_as_mats} and \ref{sec:factorize_zero_phase_mats}, and compute various performance metrics that allow us to compare our proposed approach of designing zero-phase filters as matrices with an existing method based on recursive filters, proposed in \cite{selesnic2017}. We design zero-phase high-pass filters with a cut-off frequency $0.2 \pi$ {radians}/{second} and degrees of the numerator and denominator polynomials $M = 4$, as shown in Fig. \ref{fig:1b_lp_nczp}, for different values of sample size $N$ and sparse derivative order $K$. We compute the value of the cost function obtained by solving the matrix factorization optimization problem in (\ref{eqn:matrix_factor_opt_2}). As can be seen from Table \ref{tab:performance_metrics}, the Frobenius norm in (\ref{eqn:matrix_factor_opt_2}) increases with $K$ due to the overdetermined nature of the matrix factorization problem in (\ref{eqn:matrix_factor_opt_2}). In most practical applications that perform signal smoothing or denoising, $K \in (0,2]$. In addition, we also compute the filter norms of the zero-phase high-pass filters in columns four and five of Table \ref{tab:performance_metrics}, where $\bm{h}$ denotes an impulse vector and the impulse is located at the center to avoid transients (see Fig. \ref{fig:1b_lp_nczp}). As the value of $N$ increases, we notice that the filter norms obtained by the proposed and existing method of designing zero-phase filters converge.
%%%%%%%%%%%%%%%%%%%%%%%%%%%%%%%%%%%%%%%%%%%%%%%%%%%%%%%%%%%%%%%%%%%%%%%%%%%%%%%%%%%%%%%%%%%%%%%%%%%%%%%%%%%%%%%%%%%%%%%%%%%%%%%%%%%%%%%%%%%%%%%%%%%%%%%%%%%
\subsection{Preprocessing Step\label{sec:preproc_step}}
%%%%%%%%%%%%%%%%%%%%%%%%%%%%%%%%%%%%%%%%%%%%%%%%%%%%%%%%%%%%%%%%%%%%%%%%%%%%%%%%%%%%%%%%%%%%%%%%%%%%%%%%%%%%%%%%%%%%%%%%%%%%%%%%%%%%%%%%%%%%%%%%%%%%%%%%%%%
The proposed zero-phase filters introduce undesirable transients at the start and end of the signal when the initial state vectors of the recursive filter are initialized to zero (see Proposition \ref{lemma_2}). To remove the effect of undesirable transients, we introduce a preprocessing step. Our approach is inspired by the preprocessing method introduced in \cite[Section 4.4]{selesnick2015}. In the preprocessing step, we pad the input signal of sample size $N$ with $P$ samples of preprocessed data at the start and end of the input signal. The size of $P$ depends on the sampling rate of the input signal. In our work, we choose $P$ as one-fifth of the sampling rate. Further, $P$ samples of padding data at the start and end of the input signal are obtained by using a polynomial fit of the first $P$ and last $P$ samples of the input signal, respectively. Then, the approximate polynomial used for extrapolating the input signal at the start and end, with $P$ samples. The degree of the polynomial approximation depends on the nature of the input signal. The $P$ padded samples at the start and end of the signal are removed after filtering. We avoid a simpler approach, such as zero-padding of the input signal, to escape the abrupt transients that are introduced by zero-padding.
%%%%%%%%%%%%%%%%%%%%%%%%%%%%%%%%%%%%%%%%%%%%%%%%%%%%%%%%%%%%%%%%%%%%%%%%%%%%%%%%%%%%%%%%%%%%%%%%%%%%%%%%%%%%%%%%%%%%%%%%%%%%%%%%%%%%%%%%%%%%%%%%%%%%%%%%%%%
%%%%%%%%%%%%%%%%%%%%%%%%%%%%%%%%%%%%%%%%%%%%%%%%%%%%%%%%%%%%%%%%%%%%%%%%%%%%%%%%%%%%%%%%%%%%%%%%%%%%%%%%%%%%%%%%%%%%%%%%%%%%%%%%%%%%%%%%%%%%%%%%%%%%%%%%%%%
%%%%%%%%%%%%%%%%%%%%%%%%%%%%%%%%%%%%%%%%%%%%%%%%%%%%%%%%%%%%%%%%%%%%%%%%%%%%%%%%%%%%%%%%%%%%%%%%%%%%%%%%%%%%%%%%%%%%%%%%%%%%%%%%%%%%%%%%%%%%%%%%%%%%%%%%%%%
\section{Signal Denoising and Pattern Recognition\label{sec:sig_sdpr}}
{\dred
In this section, we develop various signal models for signal denoising and pattern recognition. We apply our proposed filter designs to an existing signal model \cite{selesnick2014,selesnick2015,selesnic2017} and demonstrate the robustness of our filter designs using real and simulated data. We also propose two new signal models using our proposed zero-phase narrow band-pass filter to simultaneously denoise and detect patterns of interest. We illustrate the capabilities of the proposed frameworks using sleep-electroencephalography data to detect K-complexes and sleep spindles. All algorithms are evaluated on a Windows computer (2.7 GHz Intel Core i7) running MATLAB 2016b, unless otherwise stated explicitly.
}
%%%%%%%%%%%%%%%%%%%%%%%%%%%%%%%%%%%%%%%%%%%%%%%%%%%%%%%%%%%%%%%%%%%%%%%%%%%%%%%%%%%%%%%%%%%%%%%%%%%%%%%%%%%%%%%%%%%%%%%%%%%%%%%%%%%%%%%%%%%%%%%%%%%%%%%%%%%
\vspace{-0.5em}
\subsection{Sparsity-Assisted Signal Denoising \label{subsec:sig_denoise}}
%%%%%%%%%%%%%%%%%%%%%%%%%%%%%%%%%%%%%%%%%%%%%%%%%%%%%%%%%%%%%%%%%%%%%%%%%%%%%%%%%%%%%%%%%%%%%%%%%%%%%%%%%%%%%%%%%%%%%%%%%%%%%%%%%%%%%%%%%%%%%%%%%%%%%%%%%%%
{\dred
In this subsection, we validate the sparsity-assisted signal smoothing signal model \cite{selesnick2014,selesnick2015,selesnic2017} using our proposed filter designs. The details of implementing the SASD and its performance are presented in Appendix \ref{sec:sig_denoising}. In particular, we use illustrative examples employing simulated and real data to demonstrate the robustness of our proposed filter designs.
}
%%%%%%%%%%%%%%%%%%%%%%%%%%%%%%%%%%%%%%%%%%%%%%%%%%%%%%%%%%%%%%%%%%%%%%%%%%%%%%%%%%%%%%%%%%%%%%%%%%%%%%%%%%%%%%%%%%%%%%%%%%%%%%%%%%%%%%%%%%%%%%%%%%%%%%%%%%%
\vspace{-0.5em}
\subsection{Sparsity-Assisted Pattern Recognition \label{subsec:pattern_reco}}
%%%%%%%%%%%%%%%%%%%%%%%%%%%%%%%%%%%%%%%%%%%%%%%%%%%%%%%%%%%%%%%%%%%%%%%%%%%%%%%%%%%%%%%%%%%%%%%%%%%%%%%%%%%%%%%%%%%%%%%%%%%%%%%%%%%%%%%%%%%%%%%%%%%%%%%%%%%
Let $\bm{y}$ denote the noisy measured signal, which is written as the sum of three components. The first component is a low-frequency signal $\bm{x}_1$ with cut-off frequency $\omega_0$; the second is a band-limited signal $\bm{x}_2$ in the frequency band $[\omega_1,\omega_2]$, where $\omega_0 \le \omega_1 < \omega_2$ describes a pattern that is wavelet-shaped; and the third component is residue, which is not necessarily additive white Gaussian. Our goal is to detect patterns of interest in the input signal, i.e., wavelet-shaped components of the band-limited signal $\bm{x}_2$. We begin by modeling the noisy measured signal as
\begin{align}
\bm{y} = \bm{x}_1 + \bm{x}_2 + \bm{w}, \label{eqn:band_signal_model}
\end{align}
where $\bm{w}$ is the residual signal. Let $\hat{\bm{x}}_1$ and $\hat{\bm{x}}_2$ denote approximate estimates of $\bm{x}_1$ and $\bm{x}_2$, respectively. Given an estimate of $\bm{x}_2$, we can estimate $\bm{x}_1$ as
\begin{align}
\hat{\bm{x}}_1 &:= \mathsf{LPF}_{\omega_0}(\bm{y} - \hat{\bm{x}}_2), \label{eqn:band_lpf_apprx}
\end{align}
where $\mathsf{LPF}_{\omega_0}(\cdot)$ is the specified zero-phase low-pass impulse response matrix operator. If estimates of $\hat{\bm{x}}_2$ is known, then we can write the estimate of $\hat{\bm{x}}$ as
\begin{align}
\hat{\bm{x}} &= \hat{\bm{x}}_1 + \hat{\bm{x}}_2 \notag \\
         &= \mathsf{LPF}_{\omega_0}(\bm{y} - \hat{\bm{x}}_2) + \hat{\bm{x}}_2 \notag \\
         %&= \mathsf{LPF}_{\omega_0} (\bm{y}) + \{ \mathsf{I} - \mathsf{LPF}_{\omega_0} \} (\hat{\bm{x}}_2) \notag \\
         &= \mathsf{LPF}_{\omega_0} (\bm{y}) + \mathsf{HPF}_{\omega_0} (\hat{\bm{x}}_2) \notag \\
         &\approx \mathsf{LPF}_{\omega_0} (\bm{y}) + \mathsf{BPF}_{\omega_1}^{\omega_2} (\hat{\bm{x}}_2), \label{eqn:appx_band_signal_model}
\end{align}
where in (\ref{eqn:appx_band_signal_model}) we used the prior knowledge of the signal component of interest and limited the continuous high-pass region into a band-pass region where the signal of interest resides. To model the wavelet-shaped signal of interest in $\bm{x}_2$, we use windowed discrete wavelet transform ($\mathrm{WDWT}$). The $\mathrm{WDWT}$ coefficients, denoted by $\bm{k} \in \mathbb{R}^{W \times V}$, depend on the window length, windows overlapping factor, and the number of levels of the wavelet decomposition. In our work, we define $\boldsymbol\varPsi: \mathbb{R}^{W \times V} \to \mathbb{R}^{N}$ (the synthesis equation of $\mathrm{WDWT}$) as
\begin{align}
\boldsymbol\varPsi \bm{k} \triangleq \mathrm{WDWT}^{-1}(\bm{k}), \label{eqn:k_complex}
\end{align}
whereas $\boldsymbol\varPsi^\tpose: \mathbb{R}^{N} \to \mathbb{R}^{W \times V}$ (the analysis equation of $\mathrm{WDWT}$) is defined as
\begin{align}
\boldsymbol\varPsi^\tpose \bm{y} \triangleq \mathrm{WDWT}(\bm{y}). \label{eqn:inv_k_complex}
\end{align}
In addition, the $\mathrm{WDWT}$ satisfies a generalized version of Parseval's identity\cite{akansu1993,mallat2008wavelet}, i.e., $\norm{\boldsymbol\varPsi \bm{k}} = \norm{\bm{y}}$. Using (\ref{eqn:k_complex}) in (\ref{eqn:appx_band_signal_model}), we get
\begin{align}
\hat{\bm{x}} &\approx \mathsf{LPF}_{\omega_0} (\bm{y}) + \mathsf{BPF}_{\omega_1}^{\omega_2} (\boldsymbol\varPsi \bm{k}), \notag \\
       &= \bm{L}^\tpose \bm{L} \bm{y} + \bm{B}^\tpose \bm{B} \boldsymbol\varPsi \bm{k}, \notag 
\end{align}
where $\bm{L}^\tpose \bm{L}$ and $\bm{B}^\tpose \bm{B}$ are the zero-phase filters representing $\mathsf{LPF}_{\omega_0}$ and $\mathsf{BPF}_{\omega_1}^{\omega_2}$, respectively. In order to detect the signal patterns of interest, we construct a suitable cost function, expressed as
\begin{align}
\arg \min_{\bm{k}} & \Bigl\{ \frac{1}{2} \norm{ \bm{y} - \bm{L}^\tpose \bm{L} \bm{y} - \bm{B}^\tpose \bm{B} \boldsymbol\varPsi \bm{k} }_2^2   \Bigr\}. \label{eqn:cs_l1_min_cf_1}
\end{align}
Because the orders of the numerators and denominators polynomials of the composite low-pass filter are of equal, we can further simplify (\ref{eqn:l1_min_cf_1}) using the identity $\bm{I} - \bm{L}^\tpose \bm{L} = \bm{H}^\tpose \bm{H}$. {\dred In addition, we impose sparsity on the wavelet coefficients and the first-order difference of the reconstructed signal $\boldsymbol\varPsi \bm{k}$. Imposing sparsity on the wavelet coefficients allows the coefficients representing the signal pattern of interest dominate and the remaining coefficients are set to zero, whereas imposing sparsity on the first order difference of the reconstructed signal allows the separation of two or more K-complexes that appear close to each other. Therefore, (\ref{eqn:cs_l1_min_cf_1}) can be rewritten as
\begin{align}
\hspace{-1em}\arg \min_{\bm{k}} & \Bigl\{ \frac{1}{2} \norm{ \bm{H}^\tpose \bm{H} \bm{y} - \bm{B}^\tpose \bm{B} \boldsymbol\varPsi \bm{k} }_2^2 +  \notag \\
& \qquad \qquad \qquad \lambda_0 \norm{\bm{k}}_1 + \lambda_1 \norm{\bm{D} \boldsymbol\varPsi \bm{k}}_1 \Bigr\}. \label{eqn:cs_l1_min_cf_2}
\end{align}
}The optimization problem in (\ref{eqn:cs_l1_min_cf_2}) is convex. In our work, we solve (\ref{eqn:cs_l1_min_cf_2}) using the alternating direction method of multipliers (ADMM) \cite[Chapter 3]{boyd2011}. We call our proposed algorithm as sparsity-assisted pattern recognition (SAPR). The details of the SAPR algorithm are listed in Appendix \ref{appx_b}. On solving (\ref{eqn:cs_l1_min_cf_2}), we get estimate of $\bm{k}$. The regions with patterns of interest in the input signal can be obtained by using simple energy-based thresholding methods on the band-pass filtered estimate $\bm{B}^\tpose \bm{B} \boldsymbol\varPsi \bm{k}$.
%%%%%%%%%%%%%%%%%%%%%%%%%%%%%%%%%%%%%%%%%%%%%%%%%%%%%%%%%%%%%%%%%%%%%%%%%%%%%%%%%%%%%%%%%%%%%%%%%%%%%%%%%%%%%%%%%%%%%%%%%%%%%%%%%%%%%%%%%%%%%%%%%%%%%%%%%%%
\subsubsection{Example}
%%%%%%%%%%%%%%%%%%%%%%%%%%%%%%%%%%%%%%%%%%%%%%%%%%%%%%%%%%%%%%%%%%%%%%%%%%%%%%%%%%%%%%%%%%%%%%%%%%%%%%%%%%%%%%%%%%%%%%%%%%%%%%%%%%%%%%%%%%%%%%%%%%%%%%%%%%%
\begin{figure}[t!]
\centering
\includegraphics[width=0.48\textwidth, height=0.47\textwidth]{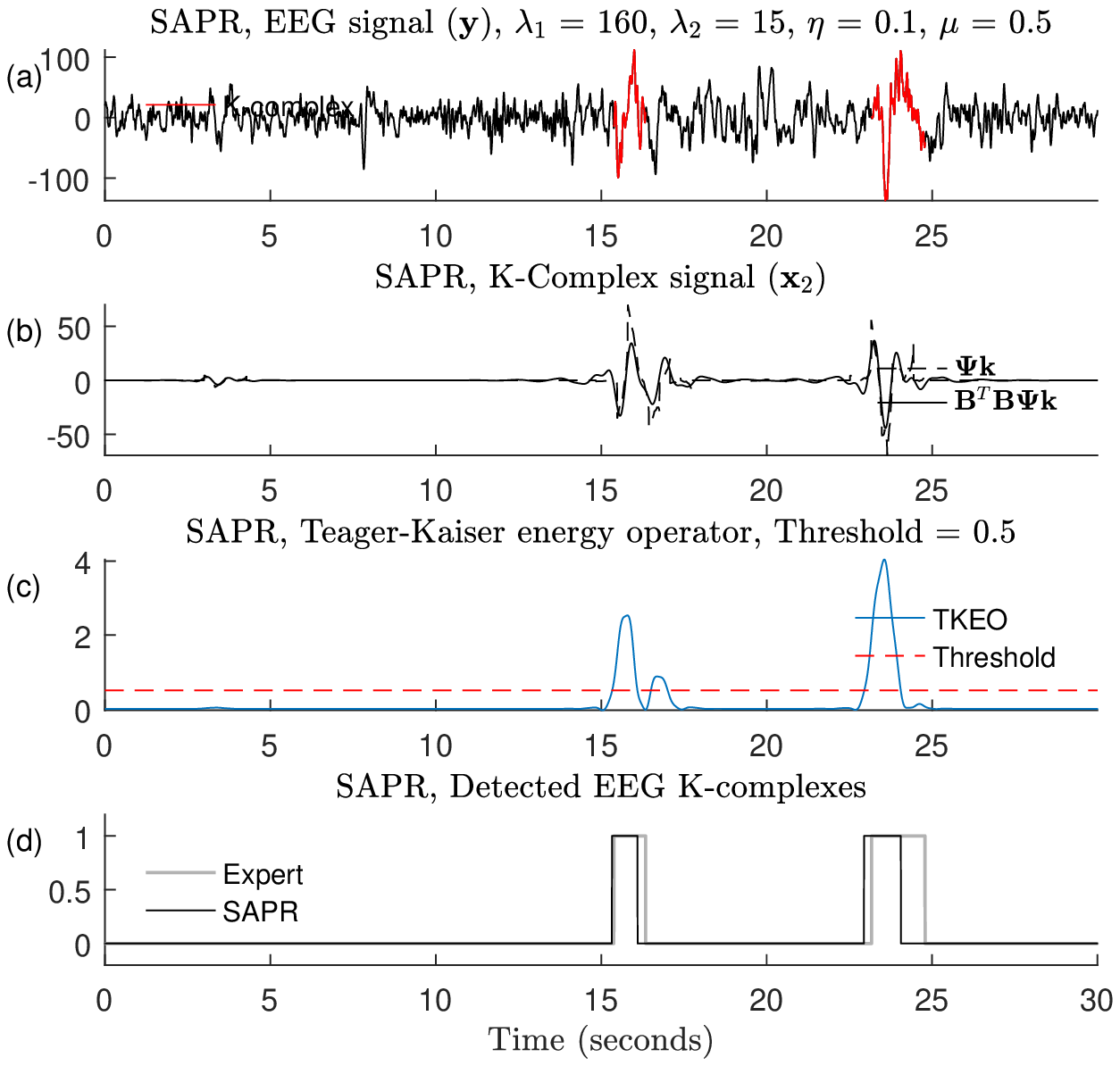}
\caption{Pattern recognition. (a) 30 second epoch of sleep-EEG data obtained from \texttt{excerpt4.edf}. The epoch consists of two K-complexes as identified by two experts. (b) Reconstructed K-complex signal $\boldsymbol\varPsi \bm{k}$ and its corresponding band-pass filtered component $\bm{B} \bm{B}^\tpose \boldsymbol\varPsi \bm{k}$ obtained using the SAPR method. (c) Signal obtained by applying the Teager-Kaiser energy operator on $\bm{B} \bm{B}^\tpose \boldsymbol\varPsi \bm{k}$. (d) Expert and algorithm annotated K-complex regions \label{fig:kcomplex}}
\end{figure}
We provide an example of detecting specific pattern of interest in sleep-EEG data, using the method described in Section \ref{subsec:pattern_reco}.  Based on the sleep scoring guidelines set by the American Academy of Sleep Medicine (AASM), human sleep can be broadly divided into two stages: rapid eye movement (REM) and non-rapid eye movement (NREM) \cite{berry2012aasm}. The NREM stage is further divided into three stages: N1, N2, and N3 \cite{berry2012aasm}. K-complexes and sleep spindles constitute physiological markers of the NREM stage of sleep. The AASM guidelines defines K-complex as ``a well-delineated, negative, sharp wave immediately followed by a positive component standing out from the background EEG, with total duration $\ge$ 0.5 seconds, usually maximal in amplitude when recorded using frontal derivations'' \cite{berry2012aasm}. These K-complexes appear within the frequency range of 0.5-2.0 Hz \cite{weigenand2014}. Current ``gold standard'' sleep staging and K-complex detection are by visual scoring by trained experts; however, this method is cumbersome and is subject to error with inter-scorer agreement of 82\% considered acceptable \cite{hopfe2009}. Our goal is to automatically detect K-complexes in sleep-EEG data using the SAPR algorithm in (\ref{eqn:cs_l1_min_cf_2}), for accurate and rapid EEG processing.

{\dred As SAPR is applicable only in batch-processing mode, we consider a fixed window of $30$ seconds. The size of the window is determined based on an established scoring criteria \cite{berry2012aasm} used by experts when scoring the K-complexes manually. In Fig. \ref{fig:kcomplex}(a), we plot $30$ seconds of sleep-EEG data obtained from the C3-A1 channels of the DREAMS database \cite{devuyst2010}. Each dataset contains $30$ minutes of sleep-EEG data sampled at $f_s = 200$ Hz, i.e., $60$ epochs of sleep-EEG data with each epoch of length $30$ seconds or $N=6000$. The ``true'' regions of K-complexes in the 30 second sleep-EEG epoch, as annotated by experts, are shown as red curves in the first plot of Fig. \ref{fig:kcomplex}(a). To design the $\mathrm{WDWT}$ in (\ref{eqn:k_complex}) and (\ref{eqn:inv_k_complex}), we select a Daubechies wavelet (\texttt{db2} or \texttt{D4}) as the mother wavelet because it closely resembles with the shape of a K-complex signal. We use a window length of the next highest power of the sampling rate $f_s$, expressed as a power of 2, i.e., $W = 2^8 = 256$. A window of length $W$, where $W$ is a power of 2, gives $\log_2 W$ levels of wavelet coefficients. In addition, we use 75\% overlap between the windows to generate an over-complete dictionary $\bm{k}$. We begin by designing a narrow band-pass filter as matrix which spread across the frequency range of the signal of interest. To detect K-complexes, we choose $\omega_1 = 0.006 \pi$ rads/s (equivalent to $0.6$ Hz), $\omega_2 = 0.02 \pi$ rads/s (equivalent to $2$ Hz), orders of the filter $M = 4$, and design a zero-phase narrow band-pass filter, denoted by $\bm{B}^\tpose \bm{B}$. We set the high-pass filter cutoff frequency $\omega_0 = 0.006\pi$ rads/s (equivalent to $0.6$ Hz), and design a zero-phase high-pass filter denoted by $\bm{H}^\tpose \bm{H}$.

\begin{table*}[t!]
\centering \footnotesize
\setlength\tabcolsep{4.9pt} % default value: 6pt
\def\arraystretch{0.65}
\caption{Performance evaluation of DETOKS\cite{parekh2015} and SAPR for K-complex detection. \label{tab:kcomplex}}
\begin{tabular}{c c cccccccccc}
\toprule
\toprule
\multirow{2}{*}{Dataset \cite{devuyst2010}} & {Cohen's $\kappa$} & \multicolumn{2}{c}{{F1 Score}} & \multicolumn{2}{c}{{Cohen's $\kappa$}} & \multicolumn{2}{c}{{Events Detected}} & \multicolumn{2}{c}{{False Detections}} & \multicolumn{2}{c}{{Computation Time (sec)}}\\
\cmidrule(lr){2-2} \cmidrule(lr){3-4} \cmidrule(lr){5-6} \cmidrule(lr){7-8} \cmidrule(lr){9-10} \cmidrule(lr){11-12}
        & Between Experts & DETOKS & SAPR  & DETOKS & SAPR & DETOKS & SAPR & DETOKS & SAPR & DETOKS & SAPR\\
\midrule
\midrule
{excerpt1}       & 0.206 & 0.378  & 0.375  & 0.354 & 0.396  & 31/45  & 27/45   & 57  & 46  & 0.456 $\pm$ 0.07 & 4.516 $\pm$ 0.44 \\
{excerpt2}       & 0.200 & 0.627  & 0.616  & 0.615 & 0.627  & 36/45  & 38/45   & 21  & 16  & 0.475 $\pm$ 0.04 & 4.957 $\pm$ 0.65 \\
{excerpt3}       & 0.233 & 0.486  & 0.492  & 0.481 & 0.496  & 10/12  & 9/12    & 11  & 4   & 0.471 $\pm$ 0.10 & 4.695 $\pm$ 0.70 \\
{excerpt4}       & 0.098 & 0.484  & 0.512  & 0.447 & 0.543  & 56/81  & 65/81   & 58  & 53  & 0.423 $\pm$ 0.11 & 4.115 $\pm$ 0.65 \\
{excerpt5}       & 0.304 & 0.437  & 0.519  & 0.415 & 0.536  & 37/45  & 37/45   & 48  & 40  & 0.463 $\pm$ 0.17 & 4.532 $\pm$ 0.52 \\
\midrule
{Average}        & 0.208 & 0.482  & 0.502  & 0.462 & 0.519  & 170/228& 176/228 & 195 & 159 & 0.457 $\pm$ 0.09 & 4.563 $\pm$ 0.59 \\
\bottomrule
\bottomrule
\end{tabular}
\vspace{1ex}

\raggedright
\footnotesize{DETOKS: Detection of K-complexes and sleep spindles; SAPR: Sparsity-assisted pattern recognition.}
\end{table*}
In Appendix \ref{sec:param_kcomplex}, we develop a methodology to determine the regularization parameters $\lambda_0$ and $\lambda_1$, and the rate of convergence parameters $\mu$ and $\eta$ of the SAPR algorithm. We set $\lambda_0 = 160$ and $\lambda_1 = 15$, $\mu = 0.5$, and $\eta = 0.1$ in our work. Note that the parameters $\mu$ and $\eta$ only affect the rate of convergence of the SAPR algorithm and not the final value of the cost function. On solving (\ref{eqn:cs_l1_min_cf_2}), we get an estimate of $\bm{B}^\tpose \bm{B} \boldsymbol\varPsi\bm{k}$ which contain information about the wavelet-like pattern of interest as shown in Fig. \ref{fig:kcomplex}(b). To detect K-complexes, we apply the Teager-Kaiser energy operator ($\mathrm{TKEO}$) \cite{kaiser1993} to estimate the instantaneous energy present in $\bm{B}^\tpose \bm{B} \boldsymbol\varPsi\bm{k}$ as shown in Fig. \ref{fig:kcomplex}(c). We select regions of $\mathrm{TKEO}(\cdot)$ where the instantaneous energy is greater than a fixed threshold value $0.5$. We allow the minimum and maximum duration of a detected K-complex to be $0.5$ and $2.25$ seconds, respectively. The lower threshold of the duration of the K-complex is determined based on the definition of the K-complex whereas the upper threshold of $2.25$ seconds is used to minimize the number of false detections caused by slow wave activity which belongs to the same frequency band as the K-complex signal but occur in multiples (see Fig. \ref{fig:kcomplex_sapr} in Appendix \ref{sec:param_kcomplex} for an illustrative example). As can be seen in Fig. \ref{fig:kcomplex}(c), the SAPR algorithm can separate K-complex like patterns that appear close to each other because of the additional sparsity inducing term $\norm{\bm{D} \boldsymbol\varPsi \bm{k}}_1$ (see Fig. \ref{fig:kcomplex_detoks} in Appendix \ref{sec:param_kcomplex} for the output of the DETOKS algorithm). To further minimize the number of false detections introduced by the slow wave activity, the SAPR algorithm only selects the first peak if two or more K-complex like patterns appear within $1.5$ seconds duration. If Fig. \ref{fig:kcomplex}(c), the SAPR algorithm rejects the second peak detected near the $15$ second interval of the sleep-EEG data because it appears within $1.5$ second interval of the first peak. Finally, in Fig. \ref{fig:kcomplex}(d), we plot the K-complex regions annotated by the experts and the regions detected using SAPR algorithm. As can be seen, the SAPR algorithm detects the expert annotated K-complex regions accurately.

In Table \ref{tab:kcomplex}, we evaluate the performance of the SAPR algorithm for the K-complex EEG dataset \cite{devuyst2010} using various performance measures. These measures include the F1-score (harmonic mean of accuracy and recall computed across all sample points), Cohen's $\kappa$ (agreement between the experts and algorithm detected K-complex intervals across all sample point), number of K-complex events detected i.e., the number of overlapping intervals between events detected by the experts and algorithm, and number of false detections, i.e., the number of non-overlapping regions between the events detected by experts the and algorithm. Note that an event consists of a set of sample points. The K-complex EEG repository consists of ten datasets each of $30$ minutes duration \cite{devuyst2010}. Only five of the ten datasets were scored independently by two experts. In our work, we considered only those datasets that were scored by two or more experts because the average value of the Cohen's $\kappa$ coefficient for the inter-rater manual scoring is $0.208$, which is low. We compare the performance of our proposed method with the DETOKS algorithm \cite{parekh2015}. In DETOKS \cite{parekh2015}, to detect K-complexes, the TKEO is applied to the low-frequency signal ($< 2$ Hz) which is obtained by removing the transient and oscillatory signal components. We use the same threshold and regularization parameters as mentioned in \cite{parekh2015} because we are evaluating the SAPR and DETOKS for the same database \cite{devuyst2010}. As can be seen in Table \ref{tab:kcomplex}, the SAPR algorithm outperforms the DETKOS algorithm in all measures except for the computation time because the DETOKS algorithm employs recursive sparse banded matrices as zero-phase filters \cite{selesnick2014,selesnick2015,selesnic2017}.}
%%%%%%%%%%%%%%%%%%%%%%%%%%%%%%%%%%%%%%%%%%%%%%%%%%%%%%%%%%%%%%%%%%%%%%%%%%%%%%%%%%%%%%%%%%%%%%%%%%%%%%%%%%%%%%%%%%%%%%%%%%%%%%%%%%%%%%%%%%%%%%%%%%%%%%%%%%%
\vspace{-1.5em}
\subsection{\dred Sparsity-Assisted Signal Denoising and Pattern  \\ Recognition \label{sec:sig_sdpr}}
%%%%%%%%%%%%%%%%%%%%%%%%%%%%%%%%%%%%%%%%%%%%%%%%%%%%%%%%%%%%%%%%%%%%%%%%%%%%%%%%%%%%%%%%%%%%%%%%%%%%%%%%%%%%%%%%%%%%%%%%%%%%%%%%%%%%%%%%%%%%%%%%%%%%%%%%%%%
Let $\bm{y}$ denote the noisy measured signal, which is written as the sum of four components. The first component is a low-frequency signal $\bm{x}_1$ with cut-off frequency $\omega_0$; the second is a band-limited signal $\bm{x}_2$ in the frequency band $[\omega_1,\omega_2]$, where $\omega_0 \le \omega_1 < \omega_2$ describes an oscillatory pattern; $\bm{x}_3$ is the sparse signal with sparse first-order derivative; and the fourth component is residue, which is not necessarily additive white Gaussian.
\setlength{\abovedisplayskip}{2.5pt} \setlength{\abovedisplayshortskip}{2.5pt}
\begin{align}
\bm{y} = \bm{x}_1 + \bm{x}_2 + \bm{x}_3 + \bm{w}, \label{eqn:spin_signal_model}
\end{align}
where $\bm{w}$ is the residual signal. Let $\hat{\bm{x}}_1$, $\hat{\bm{x}}_2$, and $\hat{\bm{x}}_3$ denote approximate estimates of $\bm{x}_1$, $\bm{x}_2$, and $\bm{x}_3$, respectively. Given an estimate of $\bm{x}_2$ and $\bm{x}_3$, we can estimate $\bm{x}_1$ as
\begin{align}
\hat{\bm{x}}_1 &:= \mathsf{LPF}_{\omega_0}(\bm{y} - \hat{\bm{x}}_2 - \hat{\bm{x}}_3), \label{eqn:band_lpf_apprx}
\end{align}
where $\mathsf{LPF}_{\omega_0}(\cdot)$ is the specified zero-phase low-pass impulse response matrix operator. If estimates of $\hat{\bm{x}}_2$ and $\hat{\bm{x}}_3$ are known, then we can write the estimate of $\hat{\bm{x}}$ as
\setlength{\abovedisplayskip}{2.5pt} \setlength{\abovedisplayshortskip}{2.5pt}
\begin{align}
\hat{\bm{x}} &= \hat{\bm{x}}_1 + \hat{\bm{x}}_2 + \hat{\bm{x}}_3 \notag \\
         &= \mathsf{LPF}_{\omega_0}(\bm{y} - \hat{\bm{x}}_2 - \hat{\bm{x}}_3) + \hat{\bm{x}}_2 + \hat{\bm{x}}_3 \notag \\
         &= \mathsf{LPF}_{\omega_0} (\bm{y}) + \{ \mathsf{I} - \mathsf{LPF}_{\omega_0} \} (\hat{\bm{x}}_2 + \hat{\bm{x}}_3) \notag \\
 %        &= \mathsf{LPF}_{\omega_0} (\bm{y}) + \mathsf{HPF}_{\omega_0} (\hat{\bm{x}}_2 + \hat{\bm{x}}_3) \notag \\
         &= \mathsf{LPF}_{\omega_0} (\bm{y}) + \mathsf{HPF}_{\omega_0} (\hat{\bm{x}}_2) + \mathsf{HPF}_{\omega_0} (\hat{\bm{x}}_3) \notag \\
         &\approx \mathsf{LPF}_{\omega_0} (\bm{y}) + \mathsf{BPF}_{\omega_1}^{\omega_2} (\hat{\bm{x}}_2) + \mathsf{HPF}_{\omega_0} (\hat{\bm{x}}_3), \label{eqn:appx_spin_signal_model}
\end{align}
where in (\ref{eqn:appx_spin_signal_model}) we used the prior knowledge of the oscillatory signal of interest. To model the oscillatory behavior of $\bm{x}_2$, we use short-time Fourier transform ($\mathrm{STFT}$). The $\mathrm{STFT}$ coefficients, denoted by $\bm{c} \in \mathbb{C}^{W \times V}$, depend on the window length, $\mathrm{STFT}$ window overlapping factor, and the length of discrete Fourier transform ($\mathrm{DFT}$). We define $\boldsymbol\varPhi: \mathbb{C}^{W \times V} \to \mathbb{R}^{N}$ (the synthesis equation of $\mathrm{STFT}$) as
\begin{align}
\boldsymbol\varPhi \bm{c} \triangleq \mathrm{STFT}^{-1}(\bm{c}), \label{eqn:sleep_spindle}
\end{align}
whereas $\boldsymbol\varPhi^\hermi: \mathbb{R}^{N} \to \mathbb{C}^{W \times V}$ (the analysis equation of $\mathrm{STFT}$) is defined as
\begin{align}
\boldsymbol\varPhi^\hermi \bm{y} \triangleq \mathrm{STFT}(\bm{y}). \label{eqn:inv_sleep_spindle}
\end{align}
Note that for a sine window, the $\mathrm{STFT}$ satisfies a generalized version of Parseval's identity, i.e., $\norm{\boldsymbol\varPhi \bm{c}} = \norm{\bm{y}}$. Using (\ref{eqn:sleep_spindle}) in (\ref{eqn:appx_spin_signal_model}), we get
\begin{align}
\hat{\bm{x}} &\approx \mathsf{LPF}_{\omega_0} (\bm{y}) + \mathsf{BPF}_{\omega_1}^{\omega_2} (\boldsymbol\varPhi \bm{c}) + \mathsf{HPF}_{\omega_0} (\hat{\bm{x}}_3), \notag \\
       &= \bm{L}^\tpose \bm{L} \bm{y} + \bm{B}^\tpose \bm{B} \boldsymbol\varPhi \bm{c} + \bm{H}^\tpose \bm{H} \bm{x}_3, \notag 
\end{align}
where $\bm{L}^\tpose \bm{L}$, $\bm{H}^\tpose \bm{H}$, and $\bm{B}^\tpose \bm{B}$, are the zero-phase filters representing $\mathsf{LPF}_{\omega_0}$, $\mathsf{HPF}_{\omega_0}$, and $\mathsf{BPF}_{\omega_1}^{\omega_2}$, respectively. In order to detect the signal patterns of interest, we construct a suitable cost function, expressed as
\begin{align}
\hspace{-0.7em}\arg \min_{\bm{c}, \bm{x}_3} &\quad \biggl\{ \frac{1}{2} \norm{ \bm{y} - \bm{L}^\tpose \bm{L} \bm{y} - \bm{B}^\tpose \bm{B} \boldsymbol\varPhi \bm{c} - \bm{H}^\tpose \bm{H} \bm{x}_3 }_2^2 \biggr\} \label{eqn:spin_l1_min_cf_1}.
\end{align}
Because the orders of the numerators and denominators polynomials of the composite low-pass filter are of equal, we can further simplify (\ref{eqn:spin_l1_min_cf_1}) using the identity $\bm{I} - \bm{L}^\tpose \bm{L} = \bm{H}^\tpose \bm{H}$. In addition, we can also impose sparsity on the Fourier coefficients $\bm{c}$, the signal $\bm{x}_3$ and its derivative. Therefore, (\ref{eqn:spin_l1_min_cf_1}) can be rewritten as
\setlength{\abovedisplayskip}{2.5pt} \setlength{\abovedisplayshortskip}{2.5pt}
\begin{align}
\arg \min_{\bm{c}, \bm{x}_3} &\quad \biggl\{ \frac{1}{2} \norm{ \bm{H}^\tpose \bm{H} \left( \bm{y} - \bm{x}_3 \right) - \bm{B}^\tpose \bm{B} \boldsymbol\varPhi \bm{c} }_2^2 + \notag \\ 
                             & \qquad \qquad \lambda_0 \norm{\bm{c}}_1 + \lambda_1 \norm{ \bm{D} \bm{x}_3}_1 + \lambda_2 \norm{\bm{x}_3}_1 \biggr\} \label{eqn:spin_l1_min_cf_2}.
\end{align}
The optimization problem in (\ref{eqn:spin_l1_min_cf_2}) is convex. In our work, we solve (\ref{eqn:spin_l1_min_cf_2}) using the alternating direction method of multipliers (ADMM) \cite[Chapter 4]{boyd2011}. We call our proposed algorithm as sparsity-assisted signal denoising and pattern recognition (SASDPR). The details of deriving the iterative procedure to solve the cost function in (\ref{eqn:spin_l1_min_cf_2}) of the SASDPR algorithm are listed in Appendix \ref{appx_spindle}. \\

\emph{Remark}: The signal model in (\ref{eqn:appx_spin_signal_model}) is different from the signal model in DETKOS \cite{parekh2015}. The SASDPR method employs zero-phase narrow band-pass filters as matrices to detect the oscillatory pattern of interest whereas in DETKOS \cite{parekh2015} uses a zero-phase Butterworth band-pass filter (not designed as a matrix) after detecting the oscillatory pattern. While the DETOKS \cite{parekh2015} is limited to detecting oscillatory patterns, the signal model in (\ref{eqn:appx_spin_signal_model}) can be further extended by developing multiple narrow non-overlapping zero-phase band-pass filters and corresponding multiresolution over-complete dictionaries that represent the patterns of interest.
%%%%%%%%%%%%%%%%%%%%%%%%%%%%%%%%%%%%%%%%%%%%%%%%%%%%%%%%%%%%%%%%%%%%%%%%%%%%%%%%%%%%%%%%%%%%%%%%%%%%%%%%%%%%%%%%%%%%%%%%%%%%%%%%%%%%%%%%%%%%%%%%%%%%%%%%%%%
\begin{figure}[t!]
\centering
\hspace{-0.25em}\includegraphics[width=0.49\textwidth, height=0.46\textwidth]{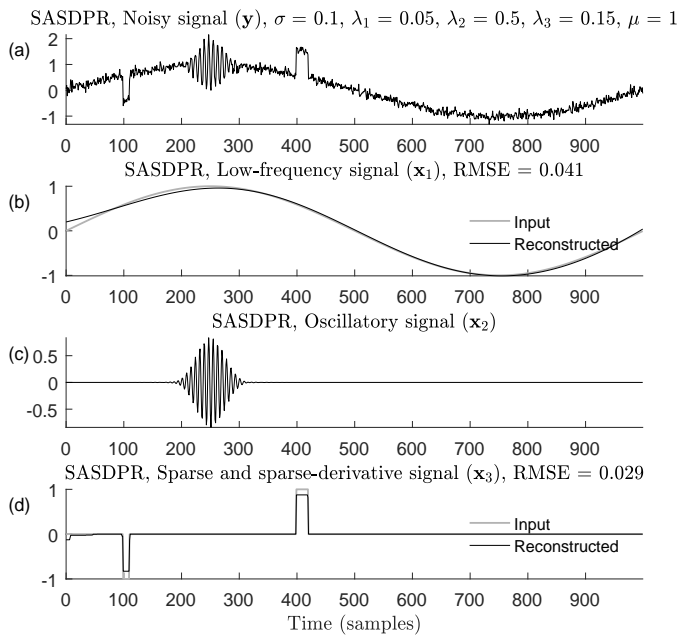}
\caption{Signal denoising and pattern recognition. (a) Input signal $\bm{y}$ is written as sum of low-frequency signal ($\bm{x}_1$), oscillatory signal ($\bm{x}_2$), sparse and sparse-derivative signal ($\bm{x}_3$), and additive white Gaussian noise with $\sigma = 0.1$. (b) Reconstructed low-frequency signal. (c) Reconstructed sparse and sparse-derivative discontinuous signal. (d) Extracted oscillatory pattern using narrow zero-phase bandpass filters. \label{fig:sdpr_toy_example}}
\end{figure}
%%%%%%%%%%%%%%%%%%%%%%%%%%%%%%%%%%%%%%%%%%%%%%%%%%%%%%%%%%%%%%%%%%%%%%%%%%%%%%%%%%%%%%%%%%%%%%%%%%%%%%%%%%%%%%%%%%%%%%%%%%%%%%%%%%%%%%%%%%%%%%%%%%%%%%%%%%%
\begin{table*}[t!]
\centering \footnotesize
\setlength\tabcolsep{8.9pt} % default value: 6pt
\def\arraystretch{0.9}
\caption{Performance Evaluation of DETOKS\cite{parekh2015} and SASDPR. \label{tab:sasdpr}}
\begin{tabular}{cccccccccc}
\toprule
\toprule
\multirow{2}{*}{Degree of Filter} & \multirow{2}{*}{RMSE}  & \multicolumn{2}{c}{$f_s = 50$} & \multicolumn{2}{c}{$f_s = 100$} & \multicolumn{2}{c}{$f_s = 150$} & \multicolumn{2}{c}{$f_s = 200$}\\
\cmidrule(lr){3-4} \cmidrule(lr){5-6} \cmidrule(lr){7-8} \cmidrule(lr){9-10}
                         & & DETOKS & SASDPR   & DETOKS & SASDPR   & DETOKS & SASDPR  & DETOKS & SASDPR \\
\midrule
\midrule
% \multirow{2}{*}{$d = 1$} & $\mathrm{rmse}(\bm{x}_1-\hat{\bm{x}}_1)$ &  0.378  & 0.352   & 0.291  & 0.278    & 0.950  & 0.950 & 0.950  & 0.950 \\
%                          & $\mathrm{rmse}(\bm{x}_3-\hat{\bm{x}}_3)$ &  0.378  & 0.352   & 0.291  & 0.278    & 0.950  & 0.950 & 0.950  & 0.950 \\
\multirow{2}{*}{$M = 2$} & $\mathrm{rmse}(\bm{x}_1)$ &  0.076  & 0.073   & 0.060  & 0.063    & 0.064  & 0.067 & 0.331  & 0.066 \\
                         & $\mathrm{rmse}(\bm{x}_3)$ &  0.053  & 0.046   & 0.025  & 0.029    & 0.019  & 0.023 & 0.279  & 0.022 \\
\multirow{2}{*}{$M = 3$} & $\mathrm{rmse}(\bm{x}_1)$ &  INF    & 0.049   & INF    & 0.046    & INF    & 0.043 & 0.028  & 0.046 \\
                         & $\mathrm{rmse}(\bm{x}_3)$ &  INF    & 0.042   & INF    & 0.028    & INF    & 0.021 & 0.065  & 0.019 \\
\multirow{2}{*}{$M = 4$} & $\mathrm{rmse}(\bm{x}_1)$ &  NA     & 0.043   & NA     & 0.041    & NA     & 0.037 & NA     & 0.040 \\
                         & $\mathrm{rmse}(\bm{x}_3)$ &  NA     & 0.042   & NA     & 0.029    & NA     & 0.021 & NA     & 0.018 \\   
\midrule
\midrule
Time (secs)        &                                          &  0.051  & 0.162   & 0.079  & 0.530    & 0.107  & 1.138 & 0.180  & 2.351 \\
\bottomrule
\bottomrule
\end{tabular}
\vspace{1ex}

\raggedright
\footnotesize{DETOKS: Detection of K-complexes and sleep spindles; SASDPR: Sparsity-assisted signal denoising and pattern recognition; INF: Very large value; NA: Not applicable.}
\end{table*}
%%%%%%%%%%%%%%%%%%%%%%%%%%%%%%%%%%%%%%%%%%%%%%%%%%%%%%%%%%%%%%%%%%%%%%%%%%%%%%%%%%%%%%%%%%%%%%%%%%%%%%%%%%%%%%%%%%%%%%%%%%%%%%%%%%%%%%%%%%%%%%%%%%%%%%%%%%%
\subsubsection{Example}
%%%%%%%%%%%%%%%%%%%%%%%%%%%%%%%%%%%%%%%%%%%%%%%%%%%%%%%%%%%%%%%%%%%%%%%%%%%%%%%%%%%%%%%%%%%%%%%%%%%%%%%%%%%%%%%%%%%%%%%%%%%%%%%%%%%%%%%%%%%%%%%%%%%%%%%%%%%
We illustrate an example to demonstrate the performance of the SASDPR algorithm in (\ref{eqn:spin_l1_min_cf_2}). In Fig. \ref{fig:sdpr_toy_example}(a), the noisy input signal $\bm{y}$, sampled at $f_s = 100$ Hz, consists of a low-frequency signal of $0.1$ Hz, an oscillatory signal belonging to $11-15$ Hz frequency range, two discontinues, before and after the oscillatory signal, and additive white Gaussian noise. Our goal in this example is to use the signal model in (\ref{eqn:spin_signal_model}), and reconstruct $\hat{\bm{x}}_1$ and $\hat{\bm{x}}_3$, and also detect the oscillatory pattern $\hat{\bm{x}}_2$. We begin by designing a narrow band-pass filter as matrix which spread across the oscillating signal's frequency range. To detect the oscillatory pattern, we choose $\omega_1 = 0.18 \pi$ rads/s (equivalent to $11$ Hz), $\omega_2 = 0.34 \pi$ rads/s (equivalent to $15$ Hz), filter order of $M = 4$, and design a zero-phase narrow band-pass filter, denoted by $\bm{B}^\tpose \bm{B}$. Next, we set the high-pass filter cutoff frequency $\omega_0 = 0.004\pi$ rads/s (equivalent to $0.1$ Hz), and design a zero-phase high-pass filter denoted by $\bm{H}^\tpose \bm{H}$. Finally, to design $\mathrm{STFT}$ and inverse $\mathrm{STFT}$ in (\ref{eqn:sleep_spindle}) and (\ref{eqn:inv_sleep_spindle}), respectively, we use a window length of the next highest power of the sampling rate $f_s$, expressed as a power of 2, i.e., $W = 2^7 = 128$, with 75\% overlap between the windows.

To determine the optimal value of the regularization parameters $\lambda_0$, $\lambda_1$, and $\lambda_2$, we perform a grid search of different combinations of $\lambda_0 \in \{0.01,0.03,\ldots,0.09\}$, $\lambda_1 \in \{0.1,0.2,\ldots,0.5\}$, and $\lambda_2 \in \{0.1,0.2,\ldots,0.5\}$. We select $\lambda_0$, $\lambda_1$, and $\lambda_2$, such that the root-mean-square error of $\bm{x}_1$ and $\bm{x}_3$ are minimized, and the oscillatory pattern $\hat{\bm{x}}_2$ is detected. In Fig. \ref{fig:sdpr_toy_example}(b)-(d), we plot the reconstructed signals $\hat{\bm{x}}_1$, $\hat{\bm{x}}_2$, and $\hat{\bm{x}}_3$ for $\lambda_0 = 0.05$, $\lambda_1 = 0.5$, and $\lambda_2 = 0.15$. Note that the parameter $\mu$ in the derivation of the SASDPR iterative method only affects the rate of convergence of the algorithm and not affect the final value of the cost function. In our simulation, we set $\mu = 1.0$. As can be seen in Fig. \ref{fig:sdpr_toy_example}, using the SASDPR signal model we can reconstruct ${\bm{x}}_1$ and ${\bm{x}}_3$, and also detect the oscillatory pattern $\bm{x}_2$. The root-mean-square error of the reconstructed signals $\bm{x}_1$ and $\bm{x}_3$ are $0.041$ and $0.029$, respectively, as shown in Fig \ref{fig:sdpr_toy_example}.

We compare the performance of the SASDPR method with the DETOKS algorithm \cite{parekh2015}. The DETOKS algorithm employs recursive sparse banded matrices as zero-phase high-pass filters \cite{selesnick2014,selesnick2015,selesnic2017}. The algorithm decomposes the input signal into three components: a) low-frequency signal, b) oscillatory signal, and c) sum of sparse and sparse-derivative signal. To find the optimal parameters of the DETOKS algorithm for the input signal in Fig. \ref{fig:sdpr_toy_example}(a), we again perform a grid search and choose $\lambda_0$, $\lambda_1$, and $\lambda_2$ such that the root-mean-square errors of $\bm{x}_1$ and $\bm{x}_33$ are minimized, and the oscillatory pattern $\hat{\bm{x}}_2$ is detected. We begin by designing zero-phase high-pass filter of order $M=2$. For $\lambda_0 = 0.05$, $\lambda_1 = 0.5$, and $\lambda_2 = 0.15$, the root-mean-square error of the reconstructed signals $\bm{x}_1$ and $\bm{x}_3$ using the DETOKS algorithm are $0.060$ and $0.025$, respectively.

In Table \ref{tab:sasdpr}, we evaluate the performance of the SASDPR and DETOKS algorithms, across various sampling rates and filter orders. Because the sampling rate of the simulated signal changes, the position of the discontinuities and oscillatory signal in Fig \ref{fig:sdpr_toy_example}(a) also change. We keep the regularization parameters fixed for the SASDPR and DETOKS algorithms, and evaluate the performance of the two algorithms by varying the sampling rates $f_s$ and orders of the filter $M$. As can be seen in Table \ref{tab:sasdpr}, the DETOKS algorithm performs well for filter orders $M \le 2$ across all sampling rates. However, when $M > 2$ and $f_s \le 150$ Hz, the zero-phase filters in the DETOKS are no longer stable due to which the reconstruction error is very large. The large reconstruction error is mainly because the sparse-banded matrix designs used as zero-phase high pass filters in DETOKS framework are unstable. In particular, the condition number of $\bm{G} = \mu \bm{A} \bm{A}^\tpose + 2 \bm{B} \bm{B}^\tpose$ in \cite[Eq. (31a)]{parekh2015} is very large and the matrix is no longer invertible. On the other hand, SASDPR algorithm demonstrates a consistent performance across different sampling rates and orders of the filter because the matrix $\bm{F}$ in (\ref{eqn:matrix_inverted}) is always positive definite. While the filter designs in the SASS framework are limited to filter orders of $M \le 3$, the same filter designs when applied in the DETOKS framework are limited to filter orders of $M \le 2$. However, filter designs obtained using our proposed method in Section \ref{sec:filt_as_mats} can achieve filter orders of $M > 3$, and also demonstrate a consistent performance across different sampling rates. 
%%%%%%%%%%%%%%%%%%%%%%%%%%%%%%%%%%%%%%%%%%%%%%%%%%%%%%%%%%%%%%%%%%%%%%%%%%%%%%%%%%%%%%%%%%%%%%%%%%%%%%%%%%%%%%%%%%%%%%%%%%%%%%%%%%%%%%%%%%%%%%%%%%%%%%%%%%%
\subsubsection{Example}
%%%%%%%%%%%%%%%%%%%%%%%%%%%%%%%%%%%%%%%%%%%%%%%%%%%%%%%%%%%%%%%%%%%%%%%%%%%%%%%%%%%%%%%%%%%%%%%%%%%%%%%%%%%%%%%%%%%%%%%%%%%%%%%%%%%%%%%%%%%%%%%%%%%%%%%%%%%
We provide an example of detecting sleep spindles in sleep-EEG data using the SASDPR method described in Section \ref{sec:sig_sdpr}. Sleep spindles are bursts of oscillatory neural activity that are generated by interplay of the thalamic reticular nucleus and other thalamic nuclei during the N2 stage of sleep. These bursts are of at least $0.5$ seconds in duration and observed in the sigma frequency range ($11$-$15$ Hz). Our goal is to automatically detect sleep-spindles in sleep-EEG data using the SASDPR algorithm in (\ref{eqn:spin_l1_min_cf_2}), for accurate and rapid EEG processing.

As SASDPR is applicable only in batch-processing mode, we consider a fixed window of $30$ seconds. The size of the window is again determined based on an established scoring criteria \cite{berry2012aasm} used by experts when scoring the sleep spindle manually. In Fig. \ref{fig:spindle}(a), we plot 30 seconds of sleep-EEG data obtained from the C3-A1 channels of the DREAMS database \cite{devuyst2011}. Each dataset contains $30$ minutes of sleep-EEG data sampled at $f_s = 200$ Hz, i.e., $60$ epochs of sleep-EEG data with each epoch of length $N=6000$. The ``true'' regions of sleep spindles in the 30 second sleep-EEG epoch, as annotated by experts, are shown as red curves in the first plot of Fig. \ref{fig:spindle}(a). We begin by designing a narrow band-pass filter as matrix which spread across the frequency range of the signal of interest. To detect sleep spindles, we choose $\omega_1 = 0.11 \pi$ rads/s (equivalent to $11$ Hz), $\omega_2 = 0.15 \pi$ rads/s (equivalent to $15$ Hz), filter order of $M = 4$, and design a zero-phase narrow band-pass filter, denoted by $\bm{B}^\tpose \bm{B}$. We set the high-pass filter cutoff frequency $\omega_0 = 0.02\pi$ rads/s (equivalent to $2$ Hz), and design a zero-phase high-pass filter denoted by $\bm{H}^\tpose \bm{H}$. The choice of the high pass filter cutoff frequency is selected so that the low-frequency signal extracted using (\ref{eqn:spin_l1_min_cf_2}) represents the deep sleep or slow wave activity of the N3 stage of sleep. Next, to design $\mathrm{STFT}$ and inverse $\mathrm{STFT}$ in (\ref{eqn:sleep_spindle}) and (\ref{eqn:inv_sleep_spindle}), respectively, we use a window length of the next highest power of the sampling rate $f_s$, expressed as a power of 2, i.e., $W = 2^8 = 256$. In addition, we use 75\% overlap between the windows to generate an over-complete dictionary of the short time Fourier transform coefficients.
\begin{table*}[t!]
\centering \footnotesize
\setlength\tabcolsep{3.4pt} % default value: 6pt
\def\arraystretch{0.65}
\caption{Performance evaluation of DETOKS\cite{parekh2015} and SASDPR for sleep spindle detection. \label{tab:spindle}}
\begin{tabular}{c c cccccccccc}
\toprule
\toprule
\multirow{2}{*}{Dataset \cite{devuyst2011}} & {Cohen's $\kappa$} & \multicolumn{2}{c}{{F1 Score}} & \multicolumn{2}{c}{{Cohen's $\kappa$}} & \multicolumn{2}{c}{{Events Detected}} & \multicolumn{2}{c}{{False Detections}} & \multicolumn{2}{c}{{Computation Time (sec)}}\\
\cmidrule(lr){2-2} \cmidrule(lr){3-4} \cmidrule(lr){5-6} \cmidrule(lr){7-8} \cmidrule(lr){9-10} \cmidrule(lr){11-12}
        & Between Experts & DETOKS & SASDPR  & DETOKS & SASDPR & DETOKS & SASDPR & DETOKS & SASDPR & DETOKS & SASDPR\\
\midrule
\midrule
{excerpt2}       & 0.515 & 0.670  & 0.660  & 0.655 & 0.644  & 60/77  & 59/77   & 22  & 23 & 0.424 $\pm$ 0.13 & 5.654 $\pm$ 0.57 \\
{excerpt4}       & 0.112 & 0.341  & 0.374  & 0.323 & 0.359  & 22/62  & 22/62   & 21  & 11 & 0.448 $\pm$ 0.12 & 5.800 $\pm$ 0.57 \\
{excerpt5}       & 0.410 & 0.523  & 0.569  & 0.506 & 0.551  & 43/103 & 50/103  & 4   & 6  & 0.400 $\pm$ 0.10 & 5.532 $\pm$ 0.59 \\
{excerpt6}       & 0.396 & 0.649  & 0.641  & 0.630 & 0.622  & 66/117 & 63/117  & 6   & 6  & 0.360 $\pm$ 0.05 & 5.310 $\pm$ 0.29 \\
\midrule
{Average}        & 0.358 & 0.545  & 0.561  & 0.528 & 0.544  & 191/359& 194/359 & 53  & 46 & 0.408 $\pm$ 0.10 & 5.574 $\pm$ 0.50 \\
\bottomrule
\bottomrule
\end{tabular}
\vspace{1ex}

\raggedright
\footnotesize{DETOKS: Detection of K-complexes and sleep spindles; SASDPR: Sparsity-assisted signal denoising and pattern recognition.}
\end{table*}
\begin{figure}[t!]
\centering
\includegraphics[width=0.47\textwidth, height=0.46\textwidth]{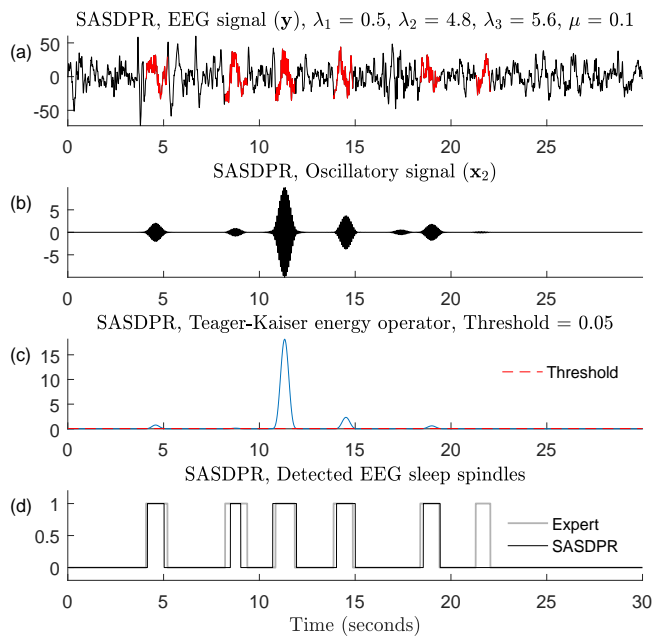}
\caption{Spindle detection. (a) 30 second epoch of sleep-EEG data obtained from \texttt{excerpt5.edf}. The epoch consists of six spindles as identified by two experts. (b) Oscillatory signal component $\bm{x}_2$ detected using the SASDPR algorithm. (c) Signal obtained by applying Teager-Kaiser energy operator on the extracted oscillatory signal component. (d) Expert and algorithm annotated sleep spindle regions. \label{fig:spindle}}
\end{figure}

In Appendix \ref{sec:param_spindle}, we develop a methodology to determine the regularization parameters $\lambda_0$, $\lambda_1$, and $\lambda_2$, and the rate of convergence parameter $\mu$. In our simulation, we set $\lambda_0 = 0.6$, $\lambda_1 = 4.8$, $\lambda_2 = 5.6$, and $\mu = 0.1$. Note that the parameter $\mu$ only affects the rate of convergence of the SASDPR algorithm and not the final value of the cost function. On solving (\ref{eqn:spin_l1_min_cf_2}), we get an estimate of $\bm{B}^\tpose \bm{B} \boldsymbol\varPhi\bm{c}$ which consists of the information about the oscillatory parameter of interest as shown in Fig. \ref{fig:spindle}(b). To detect sleep spindles, we apply the Teager-Kaiser energy operator ($\mathrm{TKEO}$) \cite{kaiser1993} and estimate the instantaneous energy present in $\bm{B}^\tpose \bm{B} \boldsymbol\varPhi\bm{c}$ as shown in Fig. \ref{fig:spindle}(c). We select regions of $\mathrm{TKEO}(\cdot)$ where the instantaneous energy is greater than a fixed threshold value $0.05$. We allow the minimum and maximum duration of a detected sleep spindle to be $0.5$ and $3.0$ seconds, respectively. In Fig. \ref{fig:spindle}(d), we plot the spindle regions annotated by the experts and SASDPR algorithm. 

In Table \ref{tab:spindle}, we evaluate the performance of the SASDPR algorithm for the sleep spindle EEG database \cite{devuyst2011} using the same performance metrics defined in the K-complex detection algorithm via the SAPR (see Table \ref{tab:kcomplex}). The sleep spindle EEG repository consists of eight datasets each of $30$ minutes duration. Only six of the eight datasets were scored independently by two experts. Of those six, only four datasets satisfy the minimum sampling rate criteria of $200$ Hz set by the AASM \cite{berry2012aasm}. We compare the performance of our proposed method with the DETOKS algorithm \cite{parekh2015}. As reported in \cite{parekh2015}, the regularization parameter, $\lambda_2$, which controls the sparsity of the STFT coefficients, belongs to a range of values from $[7.5,8.5]$. For the performance measures presented in Table \ref{tab:spindle}, we fix $\lambda_2 = 7.0$ for the DETOKS algorithm and keep the remaining parameters unchanged. See Fig. \ref{fig:spindle_detoks} in Appendix \ref{sec:param_spindle} for the output of the DETOKS algorithm. As can be seen, the SASDPR algorithm detects more number of sleep spindle events and also demonstrates a higher value F1-score than the DETOKS algorithm across all sample points. The performance of the DETOKS and SASDPR algorithms are almost similar because both algorithms are identifying oscillatory patterns using an over-complete dictionary consisting of the coefficients of the STFT. The main difference between the DETOKS and SASDPR is that the zero-phase narrow band-pass filters are incorporated within the cost function. Although, the DETOKS algorithm is computationally inexpensive, its performance varies on changing the sampling rate or the orders of the zero-phase filters (see Table \ref{tab:sasdpr}).
%%%%%%%%%%%%%%%%%%%%%%%%%%%%%%%%%%%%%%%%%%%%%%%%%%%%%%%%%%%%%%%%%%%%%%%%%%%%%%%%%%%%%%%%%%%%%%%%%%%%%%%%%%%%%%%%%%%%%%%%%%%%%%%%%%%%%%%%%%%%%%%%%%%%%%%%%%%
%%%%%%%%%%%%%%%%%%%%%%%%%%%%%%%%%%%%%%%%%%%%%%%%%%%%%%%%%%%%%%%%%%%%%%%%%%%%%%%%%%%%%%%%%%%%%%%%%%%%%%%%%%%%%%%%%%%%%%%%%%%%%%%%%%%%%%%%%%%%%%%%%%%%%%%%%%%
%%%%%%%%%%%%%%%%%%%%%%%%%%%%%%%%%%%%%%%%%%%%%%%%%%%%%%%%%%%%%%%%%%%%%%%%%%%%%%%%%%%%%%%%%%%%%%%%%%%%%%%%%%%%%%%%%%%%%%%%%%%%%%%%%%%%%%%%%%%%%%%%%%%%%%%%%%%
\section{Conclusion\label{sec:conclusion}}
In this paper, we proposed a novel approach to designing higher-order zero-phase filters as matrices using spectral transformation of the state-space representation of digital filters. We also proposed a proximal gradient-based technique to factorize a special class of zero-phase high-pass and band-pass digital filters which contain at least one zero at $z = 1$. The factorization procedure developed preserves the zero-phase property of the filters and also enables the incorporation of discontinuities into the signal model. {\dred Using the proposed filter designs, we validated and developed various signal models for denoising and pattern recognition applications. In SASD, we addressed the problem of signal denoising by simultaneously combining LTI filtering and sparsity-based techniques, and demonstrated consistent results in reconstructing the original signal when the orders of the filter are varied. In SAPR and SASDPR, we developed a general framework that combined orthogonal multiresolution representations, LTI filtering, and sparsity-based techniques to denoise and detect patterns of interest, simultaneously. Our proposed SAPR method reduced the number of false detections in identifying K-complexes in sleep-EEG data relative to the existing method. Further, using simulated data, we showed the robustness of the SASDPR method across fixed regularization parameters and varying sampling rate.}

{\dred The SAPR and SASDPR model can be extended to simultaneously detect multiple patterns of interest, which reside in non-overlapping frequency bands. In the future, we plan to develop signal models that combine stable zero-phase non-overlapping band-pass filters and corresponding over-complete dictionaries, to detect multiple patterns of interest, simultaneously. We also plan to develop stable and sparse zero-phase filters so that the computational cost is minimized while improving the performance of the various signal models presented in this work.}

\newpage
\clearpage
\setcounter{page}{1}
\twocolumn[
%  Title and authors
    \begin{center}
      {\huge Supplemental Material: Sparsity-Assisted Signal Denoising \\ \vspace{0.25em} and Pattern Recognition}\\
       \vspace{2ex}
      {G.V.~Prateek, Yo-El~Ju, and Arye~Nehorai}
    \end{center}
  ]
\appendices
\numberwithin{equation}{section}
\renewcommand{\thefigure}{\arabic{figure}}
\setcounter{figure}{0}
\renewcommand{\thetable}{\arabic{table}}
\setcounter{table}{0}
%%%%%%%%%%%%%%%%%%%%%%%%%%%%%%%%%%%%%%%%%%%%%%%%%%%%%%%%%%%%%%%%%%%%%%%%%%%%%%%%%%%%%%%%%%%%%%%%%%%%%%%%%%%%%%%%%%%%%%%%%%%%%%%%%%%%%%%%%%%%%%%%%%%%%%%%%%%
%%%%%%%%%%%%%%%%%%%%%%%%%%%%%%%%%%%%%%%%%%%%%%%%%%%%%%%%%%%%%%%%%%%%%%%%%%%%%%%%%%%%%%%%%%%%%%%%%%%%%%%%%%%%%%%%%%%%%%%%%%%%%%%%%%%%%%%%%%%%%%%%%%%%%%%%%%%
%%%%%%%%%%%%%%%%%%%%%%%%%%%%%%%%%%%%%%%%%%%%%%%%%%%%%%%%%%%%%%%%%%%%%%%%%%%%%%%%%%%%%%%%%%%%%%%%%%%%%%%%%%%%%%%%%%%%%%%%%%%%%%%%%%%%%%%%%%%%%%%%%%%%%%%%%%%
%%%%%%%%%%%%%%%%%%%%%%%%%%%%%%%%%%%%%%%%%%%%%%%%%%%%%%%%%%%%%%%%%%%%%%%%%%%%%%%%%%%%%%%%%%%%%%%%%%%%%%%%%%%%%%%%%%%%%%%%%%%%%%%%%%%%%%%%%%%%%%%%%%%%%%%%%%%
\appendices
%%%%%%%%%%%%%%%%%%%%%%%%%%%%%%%%%%%%%%%%%%%%%%%%%%%%%%%%%%%%%%%%%%%%%%%%%%%%%%%%%%%%%%%%%%%%%%%%%%%%%%%%%%%%%%%%%%%%%%%%%%%%%%%%%%%%%%%%%%%%%%%%%%%%%%%%%%%
\section{Spectral Transformation\label{appx_prop_1}}
%%%%%%%%%%%%%%%%%%%%%%%%%%%%%%%%%%%%%%%%%%%%%%%%%%%%%%%%%%%%%%%%%%%%%%%%%%%%%%%%%%%%%%%%%%%%%%%%%%%%%%%%%%%%%%%%%%%%%%%%%%%%%%%%%%%%%%%%%%%%%%%%%%%%%%%%%%%
We prove only part (b) of Proposition \ref{lemma_1}. The proofs for part (a) and (c) follow a similar approach. The numerator polynomial $B(z)$ of the transfer function $H(z)$ can be \emph{uniquely} factorized as $B(z) = B_1(z) B_2(z)$, where $B_1(z)$ is a polynomial of order $M_1$ representing the zeros at $z = -1$. The composite filter $G(z) = H(F(z))$ is obtained by replacing $z^{-1}$ terms in $H(z)$ with $1/F(z)$, where $F(z)$ depends on the frequency response type of the composite filter. If the composite filter $G(z) = H(F_{\mathrm{HP}}(z))$ is high-pass, then the transfer function can be written as
\begin{align}
G(z) = H(F_{\mathrm{HP}}(z)) = \frac{B_1(F_{\mathrm{HP}}(z)) B_2(F_{\mathrm{HP}}(z))}{A(F_{\mathrm{HP}}(z))}. \notag
\end{align}
In particular, replacing $z^{-1}$ with $1/F_{\mathrm{HP}}(z)$ in $B_1(z)$, we get
\begin{align}
B_1(F_{\mathrm{HP}}(z)) = \frac{(1 - \xi_{\mathrm{HP}})^{M_1} (1 - z^{-1})^{M_1}}{(1 + \xi_{\mathrm{HP}} z^{-1})^{M_1}} \notag.
\end{align}
Therefore, the composite high-pass filter $G(z)$ contains $M_1$ zeros at $z = 1$ and can be written as
\begin{align}
G(z) = \frac{B'_1(z) B'_2(z)}{A'(z)}, \label{eqn:factorized_tf}
\end{align}
where $B_1'(z) = (1 - z^{-1})^{M_1}$, $B_2'(z) = (1-\xi_{\mathrm{HP}})^{M_1} (1 + \xi_{\mathrm{HP}} z^{-1})^{-M_1} B_2(F_{\mathrm{HP}}(z))$, and $A'(z) = A(F_{\mathrm{HP}}(z))$. The value of $\xi_{\mathrm{HP}}$ depends on the cut-off frequency of the prototype low-pass filter and cut-off frequency of the composite high-pass filter.
%%%%%%%%%%%%%%%%%%%%%%%%%%%%%%%%%%%%%%%%%%%%%%%%%%%%%%%%%%%%%%%%%%%%%%%%%%%%%%%%%%%%%%%%%%%%%%%%%%%%%%%%%%%%%%%%%%%%%%%%%%%%%%%%%%%%%%%%%%%%%%%%%%%%%%%%%%%
\vspace{-1em}
\section{Forward-Backward Filtering\label{appx_prop_2}}
%%%%%%%%%%%%%%%%%%%%%%%%%%%%%%%%%%%%%%%%%%%%%%%%%%%%%%%%%%%%%%%%%%%%%%%%%%%%%%%%%%%%%%%%%%%%%%%%%%%%%%%%%%%%%%%%%%%%%%%%%%%%%%%%%%%%%%%%%%%%%%%%%%%%%%%%%%%
Let $(\bm{\mathcal{A}}_{\mathrm{f}},\bm{\mathcal{B}}_{\mathrm{f}},\bm{\mathcal{C}}_{\mathrm{f}},\mathcal{D}_{\mathrm{f}})$ denote the state-space representation of the forward filter. Then, based on (\ref{eqn:state_eqn}), we get
\begin{equation}
\begin{aligned}
\bm{s}(k+1) &= \bm{\mathcal{A}}_{\mathrm{f}} \bm{s}(k) + \bm{\mathcal{B}}_{\mathrm{f}} u(k) \\ \label{eqn:for_state_eqn}
y(k)    &= \bm{\mathcal{C}}_{\mathrm{f}} \bm{s}(k) + \mathcal{D}_{\mathrm{f}} u(k).
\end{aligned}
\end{equation}
Let $(\bm{\mathcal{A}}_{\mathrm{b}},\bm{\mathcal{B}}_{\mathrm{b}},\bm{\mathcal{C}}_{\mathrm{b}},\mathcal{D}_{\mathrm{b}})$ denote the state-space representation of the backward filter. Next, we express the forward filter in terms of a backward filter and find the relationship between $(\bm{\mathcal{A}}_{\mathrm{f}},\bm{\mathcal{B}}_{\mathrm{f}},\bm{\mathcal{C}}_{\mathrm{f}},\mathcal{D}_{\mathrm{f}})$ and $(\bm{\mathcal{A}}_{\mathrm{b}},\bm{\mathcal{B}}_{\mathrm{b}},\bm{\mathcal{C}}_{\mathrm{b}},\mathcal{D}_{\mathrm{b}})$. From the state equation of (\ref{eqn:for_state_eqn}), we get
\begin{align}
- \bm{s}(k) &= -\bm{\mathcal{A}}^{-1}_{\mathrm{f}} \bm{s}(k+1) + \bm{\mathcal{A}}^{-1}_{\mathrm{f}} \bm{\mathcal{B}}_{\mathrm{f}} u(k). \notag
\end{align}
Let $\bm{z}(k) \triangleq - \bm{s}(k) - \bm{\mathcal{A}}^{-1}_{\mathrm{f}} \bm{\mathcal{B}}_{\mathrm{f}} u(k)$ be the new state vector. Then,
\begin{align}
\bm{z}(k) &= - \bm{\mathcal{A}}^{-1}_{\mathrm{f}} \bm{s}(k+1) \notag \\
          &= - \bm{\mathcal{A}}^{-1}_{\mathrm{f}} \left[ - \bm{z}(k+1) - \bm{\mathcal{A}}^{-1}_{\mathrm{f}} \bm{\mathcal{B}}_{\mathrm{f}} u(k+1) \right] \label{eqn:change_state} \\
          &= \bm{\mathcal{A}}^{-1}_{\mathrm{f}} \bm{z}(k+1) + \bm{\mathcal{A}}^{-2}_{\mathrm{f}} \bm{\mathcal{B}}_{\mathrm{f}} u(k+1), \label{eqn:tf_back_filt_part_a}
\end{align}
where in (\ref{eqn:change_state}) we use the definition of $\bm{z}(k)$. From (\ref{eqn:tf_back_filt_part_a}) and the definition of $\bm{z}(k)$, we get $\bm{s}(k)$. Using $\bm{s}(k)$, we compute the output as
\begin{align}
y(k)    &= \bm{\mathcal{C}}_{\mathrm{b}} \bm{s}(k) + \mathcal{D}_{\mathrm{b}} u(k). \label{eqn:tf_back_filt_part_b}
\end{align}
Therefore, an equivalent representation of $(\bm{\mathcal{A}}_{\mathrm{b}},\bm{\mathcal{B}}_{\mathrm{b}},\bm{\mathcal{C}}_{\mathrm{b}},\mathcal{D}_{\mathrm{b}})$ in terms of $(\bm{\mathcal{A}}_{\mathrm{f}},\bm{\mathcal{B}}_{\mathrm{f}},\bm{\mathcal{C}}_{\mathrm{f}},\mathcal{D}_{\mathrm{f}})$ is expressed as $(\bm{\mathcal{A}}^{-1}_{\mathrm{f}}, \bm{\mathcal{A}}^{-2}_{\mathrm{f}} \bm{\mathcal{B}}_{\mathrm{f}}, \bm{\mathcal{C}}_{\mathrm{f}}, \mathcal{D}_{\mathrm{f}})$. Because the backward filter is noncausal, impulse response coefficients are expressed as
\begin{align}
g(k) = \begin{cases} \mathcal{D}_{\mathrm{b}} &\quad k = 0, \\ 
                    \bm{\mathcal{C}}_{\mathrm{b}} \bm{\mathcal{A}}_{\mathrm{b}}^{-k-1} \bm{\mathcal{B}}_{\mathrm{b}} &\quad k = 1,2,\ldots
       \end{cases}.
\end{align}
\begin{figure}[t!]
\centering
\psfrag{aaa}{\footnotesize $\bm{u}$}
\psfrag{ccc}{\footnotesize $\bm{p}$}
\psfrag{bbb}{\footnotesize $G(z)$}
\psfrag{ddd}{\footnotesize $\mathrm{Rev}$}
\psfrag{eee}{\footnotesize $ $}
\psfrag{fff}{\footnotesize $ $}
\psfrag{ggg}{\footnotesize $\bm{y}$}
\psfrag{hhh}{\footnotesize $\mathrm{Rev}$}
\psfrag{forward}{\footnotesize forward}
\psfrag{backward}{\footnotesize backward}
\includegraphics[width=0.45\textwidth, height=0.09\textwidth]{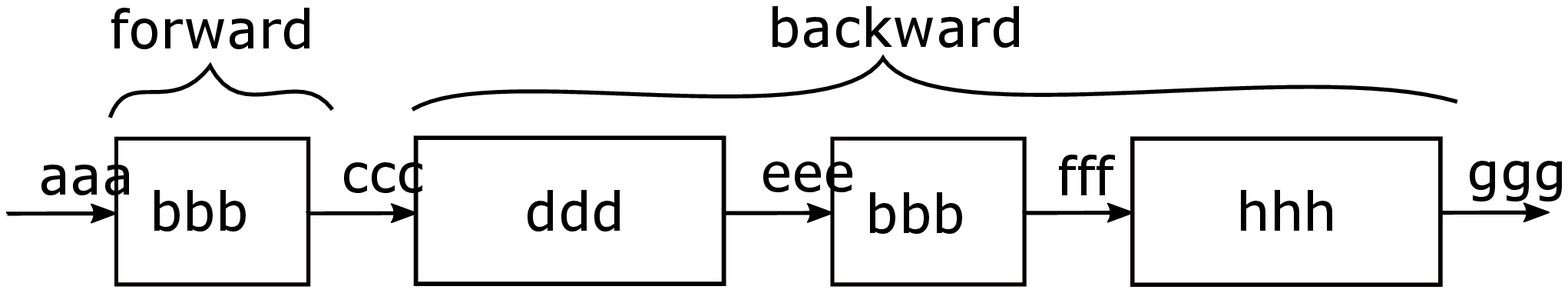}
\caption{A block diagram illustration of the forward-backward filtering approach. The forward filter computes the filter outputs using (\ref{eqn:forward_filter}) for the given transfer function whereas the backward filter computes the output using (\ref{eqn:backward_filt}). \label{fig:zp_filter}}
\end{figure}

In addition, the relationship between input and output vectors for a backward filter is expressed as
\begin{align}
\bm{y} = \bm{G}_{\mathrm{b}} \bm{u} + \bm{O}^{\mathrm{C}}_{\mathrm{b}} \bm{s}(N-1), \label{eqn:backward_filt}
\end{align}
where $\bm{s}(N-1)$ is the initial condition of the state vector in the recursive implementation of the backward filter, 
\begin{align}
\bm{G}_{\mathrm{b}} = \begin{bmatrix}
                      \mathcal{D}_{\mathrm{b}} & \bm{\mathcal{C}}_{\mathrm{b}} \bm{\mathcal{A}}_{\mathrm{b}}^{-2} \bm{\mathcal{B}}_{\mathrm{b}} & \ldots & \bm{\mathcal{C}}_{\mathrm{b}} \bm{\mathcal{A}}_{\mathrm{b}}^{-N} \bm{\mathcal{B}}_{\mathrm{b}} \\
                      0 & \mathcal{D}_{\mathrm{b}} & \ldots & \bm{\mathcal{C}}_{\mathrm{b}} \bm{\mathcal{A}}_{\mathrm{b}}^{-N+1} \bm{\mathcal{B}}_{\mathrm{b}} \\
                      \vdots & \ddots  & \ddots & \vdots \\
                      0      & \ldots &   0     & \mathcal{D}_{\mathrm{b}}, \label{eqn:back_toeplitz_mat}
                      \end{bmatrix}
\end{align}
where $\bm{G}_{\mathrm{b}}$ is an upper-triangular Toeplitz matrix representing the impulse response matrix, and $\bm{O}^{\mathrm{C}}_{\mathrm{b}}$ is the observability matrix of the backward filter. We observe that on substituting $(\bm{\mathcal{A}}_{\mathrm{b}},\bm{\mathcal{B}}_{\mathrm{b}},\bm{\mathcal{C}}_{\mathrm{b}},\mathcal{D}_{\mathrm{b}})$ as $(\bm{\mathcal{A}}^{-1}_{\mathrm{f}}, \bm{\mathcal{A}}^{-2}_{\mathrm{f}} \bm{\mathcal{B}}_{\mathrm{f}}, \bm{\mathcal{C}}_{\mathrm{f}}, \mathcal{D}_{\mathrm{f}})$ in (\ref{eqn:back_toeplitz_mat}), we get
\begin{align}
\bm{G}_{\mathrm{b}} = \bm{G}^\tpose_{\mathrm{f}}, \label{eqn:for_back_rel}
\end{align}
where $\bm{G}_{\mathrm{f}}$ is the impulse response matrix in (\ref{eqn:toeplitz_matrix}) for the forward filter represented by $(\bm{\mathcal{A}}_{\mathrm{f}},\bm{\mathcal{B}}_{\mathrm{f}},\bm{\mathcal{C}}_{\mathrm{f}},\mathcal{D}_{\mathrm{f}})$. In Fig. \ref{fig:zp_filter}, we provide a block diagram illustration of the forward-backward filtering approach \cite{gustafsson1996}. If $\bm{p}$ denotes the output of the forward filter, then
\begin{align}
\bm{p} = \bm{G}_{\mathrm{f}} \bm{u} + \bm{O}_{\mathrm{f}} \bm{s}(0). 
\end{align}
The output of the forward filter is the input to the backward filter, and thus, we get
\begin{align}
\bm{y} &= \bm{G}_{\mathrm{b}} \bm{p} + \bm{O}^{\mathrm{C}}_{\mathrm{b}} \bm{s}(N-1), \notag \\
	   &= \bm{G}_{\mathrm{b}} (\bm{G}_{\mathrm{f}} \bm{u} + \bm{O}_{\mathrm{f}} \bm{s}(0)) + \bm{O}^{\mathrm{C}}_{\mathrm{b}} \bm{s}(N-1), \notag \\
	   &= \bm{G}^\tpose_{\mathrm{f}} \bm{G}_{\mathrm{f}} \bm{u} + \bm{G}_{\mathrm{b}} \bm{O}_{\mathrm{f}} \bm{s}(0) + \bm{O}^{\mathrm{C}}_{\mathrm{b}} \bm{s}(N-1), \label{eqn:fb_filt_full_exp}
\end{align}
where (\ref{eqn:fb_filt_full_exp}) is obtained from (\ref{eqn:for_back_rel}). Since the initial states of the forward and backward filter are set to zero, the impulse response matrix of the composite filter with zero-phase is expressed as $\bm{G}^\tpose_{\mathrm{f}} \bm{G}_{\mathrm{f}}$.
%%%%%%%%%%%%%%%%%%%%%%%%%%%%%%%%%%%%%%%%%%%%%%%%%%%%%%%%%%%%%%%%%%%%%%%%%%%%%%%%%%%%%%%%%%%%%%%%%%%%%%%%%%%%%%%%%%%%%%%%%%%%%%%%%%%%%%%%%%%%%%%%%%%%%%%%%%%
\section{Matrix Factorization of Zero-Phase Filters\label{appx_a}}
%%%%%%%%%%%%%%%%%%%%%%%%%%%%%%%%%%%%%%%%%%%%%%%%%%%%%%%%%%%%%%%%%%%%%%%%%%%%%%%%%%%%%%%%%%%%%%%%%%%%%%%%%%%%%%%%%%%%%%%%%%%%%%%%%%%%%%%%%%%%%%%%%%%%%%%%%%%
Expanding the Frobenius norm  as the trace of the inner product, and selecting terms that only depend on $\bm{G}_1$, the optimization problem in (\ref{eqn:matrix_factor_opt_2}) can be written as
\begin{align}
\arg \min_{\bm{G}_1} \quad \mathrm{tr}(\bm{D}^\tpose \bm{G}_1^\tpose \bm{G}_{\mathrm{f}} \bm{G}_{\mathrm{f}}^\tpose \bm{G}_1 \bm{D} ) - 2 \mathrm{tr} (\bm{G}_{\mathrm{f}}^{\tpose} \bm{G}_{\mathrm{f}} \bm{G}_{\mathrm{f}}^\tpose \bm{G}_1 \bm{D}). \notag
\end{align}
The trace of the product of two matrices can be written as the dot product of two vectors using the vectorization operator, i.e., $\mathrm{tr}(\bm{A}^\tpose \bm{B}) = \mathrm{vec}(\bm{A}^\tpose) \mathrm{vec}(\bm{B})$. Thus, objective function can be rewritten in vectorized form as
\begin{align}
\begin{split}
\arg \min_{\bm{G}_1} &\quad \mathrm{vec}(\bm{G}_{\mathrm{f}}^\tpose \bm{G}_1 \bm{D})^\tpose \mathrm{vec}(\bm{G}_{\mathrm{f}}^\tpose \bm{G}_1 \bm{D}) \\
                     &\qquad - 2 \mathrm{vec}(\bm{G}_{\mathrm{f}} \bm{G}_{\mathrm{f}}^\tpose \bm{G}_{\mathrm{f}} \bm{D}^\tpose)^\tpose \mathrm{vec}(\bm{G}_1). \notag
\end{split}
\end{align}
Using the mixed-product property of the Kronecker product, the objective function can be further simplified as
\begin{align}
\begin{split}
\arg \min_{\mathrm{vec}(\bm{G}_1)} &\quad \mathrm{vec}(\bm{G}_1)^\tpose (\bm{D} \bm{D}^\tpose \otimes \bm{G}_{\mathrm{f}} \bm{G}_{\mathrm{f}}^\tpose) \mathrm{vec}(\bm{G}_1) \\
                     &\qquad - 2 \mathrm{vec}(\bm{G}_{\mathrm{f}} \bm{G}_{\mathrm{f}}^\tpose \bm{G}_{\mathrm{f}} \bm{D}^\tpose)^\tpose \mathrm{vec}(\bm{G}_1). \label{eqn:cost_func_vectorized}
\end{split}
\end{align}
The objective function in (\ref{eqn:cost_func_vectorized}) is convex because $\bm{D} \bm{D}^\tpose \otimes \bm{G}_{\mathrm{f}} \bm{G}_{\mathrm{f}}^\tpose$ is always positive semidefinite. Further, if $\bm{\mathcal{T}}$ is a set of vectorized lower-triangular matrices, then $\bm{\mathcal{T}} = \{ \bm{G}_1 \in \mathbb{R}^{N \times (N-K)} | \mathrm{vec}(\bm{G}_1) = \mathrm{vec}(\mathrm{tril}(\bm{G}_1)) \}$ is a convex set, because any linear combination of lower-triangular matrices is always a lower-triangular matrix. Therefore, the optimization problem in (\ref{eqn:matrix_factor_opt}) is a quadratic program with linear equality constraints, and thus convex. There are many efficient ways to solve the optimization problem in (\ref{eqn:matrix_factor_opt}), however, for sufficiently large $N$, solving (\ref{eqn:matrix_factor_opt}) would require computing and storing a large Kronecker product $\bm{D} \bm{D}^\tpose \otimes \bm{G}_{\mathrm{f}} \bm{G}_{\mathrm{f}}^\tpose$ of size $N^2 \times N^2$, and also find its pseudo-inverse. Further, when $\mathcal{D}_\mathrm{f}$ is close to zero, inverting $\bm{G}_{\mathrm{f}}$ is not stable because it is an ill conditioned matrix. To avoid the computational and storage burdens of solving (\ref{eqn:matrix_factor_opt}), we develop an accelerated projected gradient descent approach. In particular, we use the fast iterative shrinkage-threshold algorithm (FISTA) to solve the optimization problem in (\ref{eqn:matrix_factor_opt}). First, we rewrite the cost function as
\begin{align}
\arg \min_{\mathrm{vec}(\bm{G}_1)} \psi(\mathrm{vec}(\bm{G}_1)) + \phi(\mathrm{vec}(\bm{G}_1)),
\end{align}
where $\psi(\cdot)$ is given (\ref{eqn:cost_func_vectorized}) and $\phi(\cdot)$ imposes a lower-triangular matrix constraint on $\bm{G}_1$. The iterative shrinkage operator denoted by $p_L(\cdot)$ is given as
\begin{align}
\begin{split}
p_L(\mathrm{vec}(\bm{G}_1)) &= \arg \min_{\bm{b}} \biggl\{ \phi(\bm{b}) + \biggr. \\
&\hspace{-3em}+ \biggl. \frac{L}{2} \norm{\bm{b} - \left( \mathrm{vec}(\bm{G}_1) - \frac{1}{L} \nabla \psi(\mathrm{vec}(\bm{G}_1)) \right)}_2^2 \biggr\}, \label{eqn:big_opt}
\end{split}
\end{align}
where $\bm{b} \in \mathbb{R}^{N(N-K) \times 1}$ is the vectorized version of the lower-triangular matrix to be found and $L$ is the Lipschitz constant. The lower-triangular constraint can be directly incorporated into the optimization problem (\ref{eqn:big_opt}) by imposing the constraint $\bm{b} \in \bm{\mathcal{T}}$, which gives us the following objective function:
\begin{align}
\arg \min_{\bm{b} \in \bm{\mathcal{T}}} \frac{L}{2} \norm{\bm{b} - \left( \mathrm{vec}(\bm{G}_1) - \frac{1}{L} \nabla \psi(\mathrm{vec}(\bm{G}_1)) \right)}_2^2, \label{eqn:opt_simplified}
\end{align}
where the gradient of the cost function in vectorized form $\nabla \psi (\bm{G}_1) = \mathrm{vec}(\bm{G}_{\mathrm{f}} \bm{G}_{\mathrm{f}}^\tpose \bm{G}_1 \bm{D} \bm{D}^\tpose) - \mathrm{vec}(\bm{G}_{\mathrm{f}} \bm{G}_{\mathrm{f}}^\tpose \bm{G}_{\mathrm{f}} \bm{D}^\tpose)$. The constraint is applied by reshaping the $N(N-K) \times 1$ vector $\mathrm{vec}(\bm{G}_1)$ into an $N \times N-K$ matrix $\bm{G}_1$ and applying lower triangular matrix constraint using $\texttt{tril}$ command in MATLAB. In the initialization step, we initialize $\bm{G}_1^{(1)}$ as an $N \times N-K$ impulse response matrix obtained from the transfer function $G_1(z)$ in (\ref{eqn:factorized_tf}), where $\bm{G}_1$ is a lower triangular Toeplitz matrix. The complete algorithm to factorize zero-phase filters as matrices is listed as \emph{Algorithm \ref{alg:pgd_algo}}.
%%%%%%%%%%%%%%%%%%%%%%%%%%%%%%%%%%%%%%%%%%%%%%%%%%%%%%%%%%%%%%%%%%%%%%%%%%%%%%%%%%%%%%%%%%%%%%%%%%%%%%%%%%%%%%%%%%%%%%%%%%%%%%%%%%%%%%%%%%%%%%%%%%%%%%%%%%%
\begin{algorithm}[t!]
\small
\setstretch{1.05}
\caption{\small Accelerated Projected Gradient Descent (\ref{eqn:matrix_factor_opt_2})}\label{alg:pgd_algo}
\begin{algorithmic}
\Procedure{APGD}{$\bm{G},\bm{D},\bm{G}_1^{(1)},K,\epsilon,k_{\mathrm{max}}$}
\State \textbf{initialize}
\State $\mathrm{TOL} \gets \epsilon$, $\mathrm{MAX} \gets k_\mathrm{max}$, $k \gets 1$
\State $t_1 \gets 1$, $\bm{b}_1 \gets \mathrm{vec}(\bm{G}_1^{(1)})$, $\bm{f}_0 \gets \bm{b}_1$
\State $\bm{B} \gets \bm{G} \bm{G}^\tpose$, $\bm{Q} \gets \bm{D} \bm{D}^\tpose$, $\bm{R} \gets \bm{G} \bm{G}^\tpose \bm{G} \bm{D}^\tpose$
\State $c_1 \gets \norm{\bm{G}^{\tpose} \bm{G} - \bm{G}^{\tpose} \bm{G}^{(k)}_1 \bm{D} }_{\mathrm{F}}^2$
\Repeat
	\State $k = k + 1$ \Comment{(increment counter)}
	\State $\bm{f}_{k-1} \gets \bm{b}_{k-1} - ({1}/{L}) [ \mathrm{vec}(\bm{B} \bm{G}_1^{(k-1)} \bm{Q}) - \mathrm{vec}(\bm{R}) ]$
	\State $\bm{G}_1^{(k)} \gets \mathrm{tril}(\mathrm{reshape}(\bm{f}_{k-1},N,N-K))$ \Comment{(project)}
	%\State $\bm{b}_k \gets \mathrm{vec}(\bm{G}_1^{(k)})$
	\State $t_k \gets (1 + \sqrt{1 + 4 t_k^2})/2$
	\State $\bm{b}_k \gets \bm{f}_{k-1} + (\bm{f}_{k-1} - \bm{f}_{k-2})({t_k - 1}/{ t_{k-1}})$
	\State $c_k = \norm{\bm{G}^{\tpose} \bm{G} - \bm{G}^{\tpose} \bm{G}^{(k)}_1 \bm{D} }_{\mathrm{F}}^2$ \Comment{(compute cost)}
\Until{$(k = \mathrm{MAX}) \mbox{ \textbf{or} } |c_k - c_{k-1}| < \mathrm{TOL}$}
\State \Return{$\bm{G}_1^{(k)}$}
\EndProcedure
\end{algorithmic}
\end{algorithm}
%%%%%%%%%%%%%%%%%%%%%%%%%%%%%%%%%%%%%%%%%%%%%%%%%%%%%%%%%%%%%%%%%%%%%%%%%%%%%%%%%%%%%%%%%%%%%%%%%%%%%%%%%%%%%%%%%%%%%%%%%%%%%%%%%%%%%%%%%%%%%%%%%%%%%%%%%%%
%%%%%%%%%%%%%%%%%%%%%%%%%%%%%%%%%%%%%%%%%%%%%%%%%%%%%%%%%%%%%%%%%%%%%%%%%%%%%%%%%%%%%%%%%%%%%%%%%%%%%%%%%%%%%%%%%%%%%%%%%%%%%%%%%%%%%%%%%%%%%%%%%%%%%%%%%%%
\section{Sparsity-Assisted Signal Denoising \label{sec:sig_denoising}}
%%%%%%%%%%%%%%%%%%%%%%%%%%%%%%%%%%%%%%%%%%%%%%%%%%%%%%%%%%%%%%%%%%%%%%%%%%%%%%%%%%%%%%%%%%%%%%%%%%%%%%%%%%%%%%%%%%%%%%%%%%%%%%%%%%%%%%%%%%%%%%%%%%%%%%%%%%%
Let $\bm{y}$ denote the noisy measured signal, which can be modeled as an additive mixture of a low-frequency signal $\bm{x}_1$ and a sparse-derivative signal $\bm{x}_2$:
\begin{align}
\bm{y} = \bm{x}_1 + \bm{x}_2 + \bm{w}, \label{eqn:sass_signal_model}
\end{align}
where $\bm{w}$ is assumed be stationary white Gaussian noise. Let $\hat{\bm{x}}_1$ and $\hat{\bm{x}}_2$ denote the approximate estimates of $\bm{x}_1$ and $\bm{x}_2$, respectively. Given an estimate of $\bm{x}_2$, we can estimate $\bm{x}_1$ as
\begin{align}
\hat{\bm{x}}_1 &:= \mathsf{LPF}_{\omega_1}(\bm{y} - \hat{\bm{x}}_2), \label{eqn:sass_lpf_apprx}
\end{align}
where $\mathsf{LPF}_{\omega_1}(\cdot)$ is the specified zero-phase low-pass impulse response matrix operator. If an estimate of $\hat{\bm{x}}_2$ is known, then we can write the estimate of $\hat{\bm{x}}$ as
\begin{align}
\hat{\bm{x}} &= \hat{\bm{x}}_1 + \hat{\bm{x}}_2 \notag \\
         &= \mathsf{LPF}_{\omega_1}(\bm{y} - \hat{\bm{x}}_2) + \hat{\bm{x}}_2 \notag \\
         &= \mathsf{LPF}_{\omega_1} (\bm{y}) - \mathsf{LPF}_{\omega_1} (\hat{\bm{x}}_2) + \hat{\bm{x}}_2 \notag \\
         &= \mathsf{LPF}_{\omega_1} (\bm{y}) + \{ \mathsf{I} - \mathsf{LPF}_{\omega_1} \} (\hat{\bm{x}}_2) \notag \\
         &= \mathsf{LPF}_{\omega_1} (\bm{y}) + \mathsf{HPF}_{\omega_1} (\hat{\bm{x}}_2). \label{eqn:appx_sig_model}
\end{align}
In (\ref{eqn:appx_sig_model}), we assumed that the orders of the denominator and numerator polynomials of the composite filter $G(z)$ are equal and used the identity $\mathsf{HPF}_{\omega_1}( \bm{u} ) \triangleq  \{ \mathsf{I} - \mathsf{LPF}_{\omega_1} \} ( \bm{u} )$. Using the definitions of $\mathsf{LPF}_{\omega_1}$ and $\mathsf{HPF}_{\omega_1}$ in Table \ref{tab:zero_phase_properties}, we get
\begin{align}
\hat{\bm{x}} &= \bm{L}^\tpose \bm{L} \bm{y} + \bm{H}^\tpose \bm{H} \bm{x}_2. \label{eqn:x_hat_estimate}
\end{align}
Now, we can incorporate the sparse derivative nature of $\bm{x}_2$ in (\ref{eqn:x_hat_estimate}) by factorizing the zero-phase impulse response matrix $\bm{H}^\tpose \bm{H}$ as $\bm{H}^\tpose \bm{H}_1 \bm{D}$. Therefore, we can approximate (\ref{eqn:x_hat_estimate}) as
\begin{align}
\hat{\bm{x}} \approx \bm{L}^\tpose \bm{L} \bm{y} + \bm{H}^\tpose \bm{H}_1 \bm{D} \bm{x}_2. \label{eqn:x_hat_appx}
\end{align}
As $\bm{x}_2$ is unknown, we cannot directly estimate $\hat{\bm{x}}$ from (\ref{eqn:x_hat_appx}). Let $\bm{D} \bm{x}_2 = \bm{v}$, where $\bm{v}$ is sparse, i.e., $\bm{x}_2$ is the sparse derivative signal. In order to estimate $\bm{x}$, we minimize the following cost function:
\begin{align}
\arg \min_{\bm{v}} \frac{1}{2} \norm{ \bm{y} - \bm{L}^\tpose \bm{L} \bm{y} - \bm{H}^\tpose \bm{H}_1 \bm{v} }_2^2 + \lambda \norm{\bm{v}}_1, \label{eqn:l1_min_cf_1}
\end{align}
where $\lambda$ is the regularization parameter. The cost function in (\ref{eqn:l1_min_cf_1}) is convex. Because the orders of the numerator and denominator polynomials of the composite low-pass filter are equal, we can further simplify (\ref{eqn:l1_min_cf_1}) using the identity $\bm{I} - \bm{L}^\tpose \bm{L} = \bm{H}^\tpose \bm{H}$ (see Table \ref{tab:zero_phase_properties}). Therefore, (\ref{eqn:l1_min_cf_1}) can be rewritten as
\begin{align}
\arg \min_{\bm{v}}  \frac{1}{2} \norm{ \bm{H}^\tpose \bm{H} \bm{y} - \bm{H}^\tpose \bm{H}_1 \bm{v} }^2_2 + \lambda \norm{\bm{v}}_1. \label{eqn:l1_min_cf_2}
\end{align}
The optimization problem in (\ref{eqn:l1_min_cf_2}) is a standard $\ell_1$ norm sparse least squares problem, which can be solved using iterative optimization techniques \cite{figueiredo2003,daubechies2004,combettes2008,beck2009,afonso2010,goldstein2009}. In our work, we solve (\ref{eqn:l1_min_cf_2}) using the fast iterative shrinkage/threshold algorithm (FISTA) \cite{beck2009}. We use FISTA because most proximal algorithms work under extremely general conditions, including cases where the functions are non-smooth, and it offers an improved convergence rate $\mathcal{O}(1/k^2)$ while ISTA demonstrates convergence rate of $\mathcal{O}(1/k)$. We skip the details of the algorithm and direct the readers to \cite{beck2009} for more details. On solving (\ref{eqn:l1_min_cf_2}), we get $\bm{v}$, using which we can estimate $\bm{x}_2 = \bm{S} \bm{v}$ when $K = 1$, where $\bm{S} \in \mathbb{R}^{N \times N-1}$ is the integration matrix \cite{selesnick2014,selesnick2015} and given as 
\begin{align}
\bm{S} = \begin{bmatrix} 0 &   &   &  & \\
             1 & 0 &   &  & \\ 
             1 & 1 & 0 &  & \\
             \vdots & & \ddots & \\
             1 & 1 & \ldots & 1 & 0 \\
             1 & 1 & \ldots & 1 & 1 \end{bmatrix}, \label{eqn:reconst_matrix}
\end{align}
and $\bm{D}\bm{S} = \bm{I}$. Based on \cite[Proposition 1.3]{bach2012}, the vector $\bm{v}$ is a solution of the optimization problem in (\ref{eqn:l1_min_cf_2}) if and only if $\forall j = 1,\ldots,N-K$,
\begin{align}
\begin{cases} (1/\lambda)\left|\bm{H}_1^\tpose \bm{H} (\bm{H}^\tpose \bm{H} \bm{y} - \bm{H}^\tpose \bm{H} \bm{v})\right| \le 1 &\mbox{if } \bm{v}_j = 0 \\ (1/\lambda)\left[\bm{H}_1^\tpose \bm{H} (\bm{H}^\tpose \bm{H} \bm{y} - \bm{H}^\tpose \bm{H} \bm{v})\right] = \mathrm{sign}(\bm{v}_j) &\mbox{if } \bm{v}_j \ne 0 \end{cases}, \label{eqn:opt_condition}
\end{align}
where $\bm{v}_j$ is the $j$-th entry of $\bm{v}$ and $\mathrm{sign}$ is the sign function. We demonstrate the effectiveness of SASD and verify the optimality condition (\ref{eqn:opt_condition}) with the help of an illustrative example.
%%%%%%%%%%%%%%%%%%%%%%%%%%%%%%%%%%%%%%%%%%%%%%%%%%%%%%%%%%%%%%%%%%%%%%%%%%%%%%%%%%%%%%%%%%%%%%%%%%%%%%%%%%%%%%%%%%%%%%%%%%%%%%%%%%%%%%%%%%%%%%%%%%%%%%%%%%%
\begin{table*}[t!]
\centering
\caption{Comparative $\mathrm{RMSE}$ of Proposed and Existing Methods \label{tab:sasd}} \vspace{-0.5em}
\footnotesize
\label{tab:compare_performance}
\setlength\tabcolsep{3pt}
\def\arraystretch{0.70}
\begin{tabular}{c c cccc cccc}
\toprule
\toprule
\multirow{2}{*}{Degree of Filter} & \multirow{2}{*}{Noise ($\sigma$)}  & \multicolumn{4}{c}{Root-mean-square error} & \multicolumn{4}{c}{Convergence time (sec)}  \\
\cmidrule(lr){3-6} \cmidrule(lr){7-10} 
  &   & $\mathsf{LPF}_{\omega_0}(\cdot)$   & $\mathrm{TVD}$\cite{figueiredo2006}      & {SASS}\cite{selesnic2017} & {SASD} & $\mathsf{LPF}_{\omega_0}(\cdot)$   & $\mathrm{TVD}$\cite{figueiredo2006}      & {SASS}\cite{selesnic2017} & {SASD} \\
\midrule
\midrule
\multirow{3}{*}{$M = 1$}  & 0.1   &  0.188 $\pm$ 0.001  &  0.061 $\pm$ 0.003  &  0.089 $\pm$ 0.008 & \textbf{0.035 $\pm$ 0.005}  &  0.029 $\pm$ 0.002 &  0.002  &  {0.008 $\pm$ 0.003} & 0.120 $\pm$ 0.008  \\
                          & 0.3   &  0.203 $\pm$ 0.009  &  0.131 $\pm$ 0.009  &  0.151 $\pm$ 0.013 & \textbf{0.100 $\pm$ 0.015}  &  0.031 $\pm$ 0.002 &  0.002  &  {0.009 $\pm$ 0.003} & 0.120 $\pm$ 0.008  \\
                          & 0.5   &  0.227 $\pm$ 0.022  &  0.191 $\pm$ 0.016  &  0.186 $\pm$ 0.022 & \textbf{0.158 $\pm$ 0.024}  &  0.031 $\pm$ 0.004 &  0.003  &  {0.011 $\pm$ 0.004} & 0.121 $\pm$ 0.013  \\
                          \midrule
\multirow{3}{*}{$M = 2$}  & 0.1   &  0.188 $\pm$ 0.001  &  0.061 $\pm$ 0.003  &  0.049 $\pm$ 0.004 & \textbf{0.035 $\pm$ 0.005}  &  0.029 $\pm$ 0.002 &  0.002  &  {0.009 $\pm$ 0.003} & 0.124 $\pm$ 0.013  \\
                          & 0.3   &  0.203 $\pm$ 0.009  &  0.131 $\pm$ 0.009  &  \textbf{0.098 $\pm$ 0.014} & 0.100 $\pm$ 0.015  &  0.031 $\pm$ 0.002  &  0.002  &  {0.012 $\pm$ 0.005} & 0.122 $\pm$ 0.011  \\
                          & 0.5   &  0.227 $\pm$ 0.022  &  0.191 $\pm$ 0.016  &  \textbf{0.146 $\pm$ 0.024} & 0.158 $\pm$ 0.024  &  0.031 $\pm$ 0.004  &  0.003  &  {0.014 $\pm$ 0.006} & 0.123 $\pm$ 0.010  \\
                          \midrule

\multirow{3}{*}{$M = 3$}  & 0.1   &  0.188 $\pm$ 0.001  &  0.061 $\pm$ 0.003  &  0.035 $\pm$ 0.005 & \textbf{0.035 $\pm$ 0.003}  &  0.029 $\pm$ 0.002  &  0.002  &  {0.010 $\pm$ 0.003} & 0.138 $\pm$ 0.118  \\
                          & 0.3   &  0.203 $\pm$ 0.009  &  0.131 $\pm$ 0.009  &  \textbf{0.094 $\pm$ 0.015} & 0.100 $\pm$ 0.015  &  0.031 $\pm$ 0.002  &  0.002  &  {0.016 $\pm$ 0.008} & 0.140 $\pm$ 0.090  \\
                          & 0.5   &  0.227 $\pm$ 0.022  &  0.191 $\pm$ 0.016  &  \textbf{0.149 $\pm$ 0.024} & 0.158 $\pm$ 0.024  &  0.031 $\pm$ 0.004  &  0.003  &  {0.021 $\pm$ 0.016} & 0.145 $\pm$ 0.096  \\
                          \midrule

\multirow{3}{*}{$M = 4$}  & 0.1   &  0.188 $\pm$ 0.001  &  0.061 $\pm$ 0.003  &  NA & \textbf{0.035 $\pm$ 0.005}  &  0.029 $\pm$ 0.002  &  0.002  &  NA & 0.141 $\pm$ 0.038  \\
                          & 0.3   &  0.203 $\pm$ 0.009  &  0.131 $\pm$ 0.009  &  NA & \textbf{0.100 $\pm$ 0.015}  &  0.031 $\pm$ 0.002  &  0.002  &  NA & 0.140 $\pm$ 0.033  \\
                          & 0.5   &  0.227 $\pm$ 0.022  &  0.191 $\pm$ 0.016  &  NA & \textbf{0.158 $\pm$ 0.024}  &  0.031 $\pm$ 0.004  &  0.003  &  NA & 0.141 $\pm$ 0.032  \\
\bottomrule
\bottomrule
\end{tabular}
\end{table*}
%%%%%%%%%%%%%%%%%%%%%%%%%%%%%%%%%%%%%%%%%%%%%%%%%%%%%%%%%%%%%%%%%%%%%%%%%%%%%%%%%%%%%%%%%%%%%%%%%%%%%%%%%%%%%%%%%%%%%%%%%%%%%%%%%%%%%%%%%%%%%%%%%%%%%%%%%%%
\begin{figure}[t!]
\centering
\includegraphics[width=0.47\textwidth, height=0.47\textwidth]{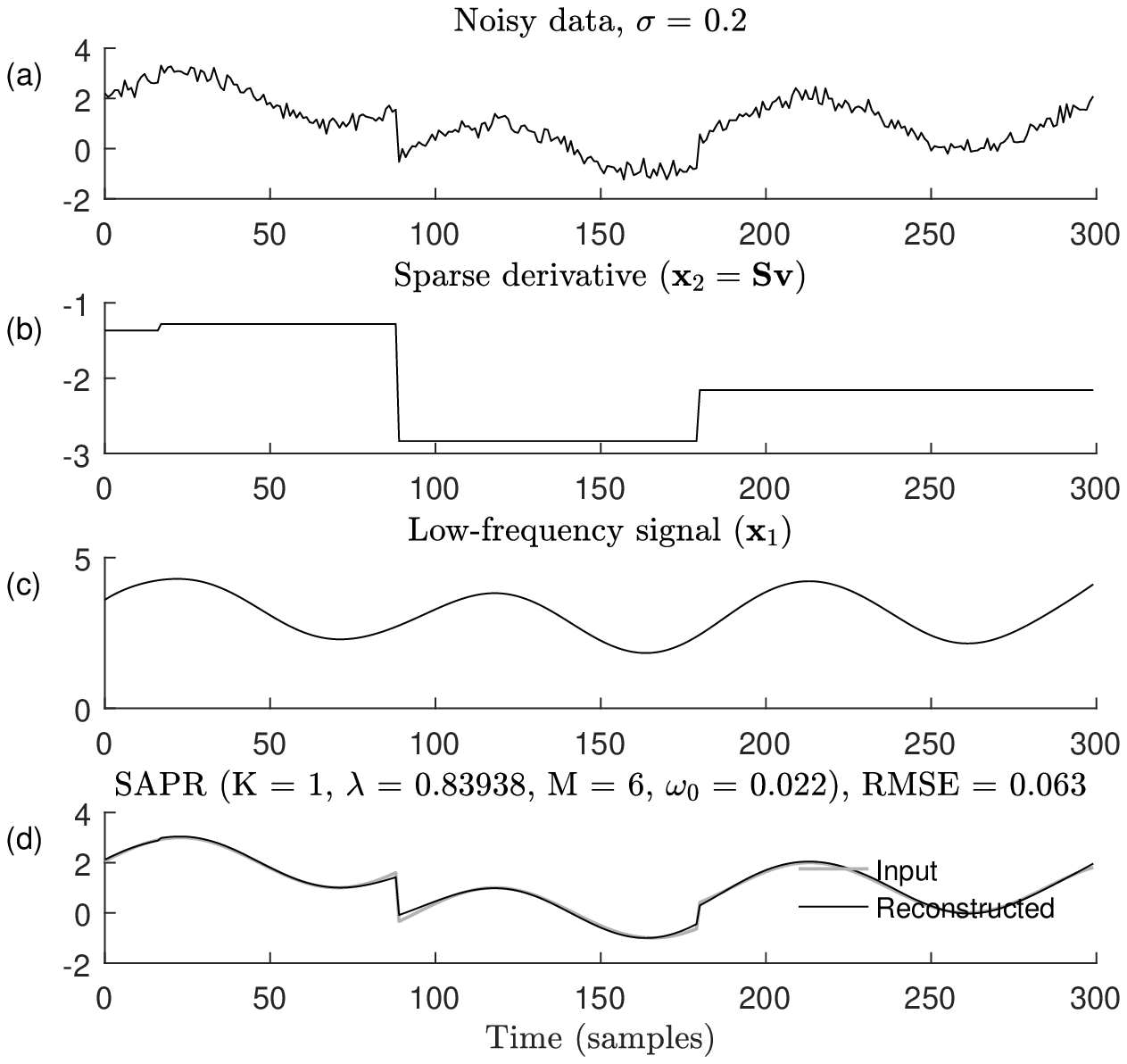}
\vspace{-1em}
\caption{Signal denoising. (a) Input signal with additive white Gaussian noise of $\sigma = 0.2$. (b) Reconstructed low-frequency signal with $\omega_0 = 0.044$ rad/s. (c) Reconstructed discontinuous signal modeled as a sparse-derivative signal. (d) Denoised input signal which is the sum of the reconstructed low-frequency and sparse-derivative signal. \label{fig:sig_denoise}}
\end{figure}
%%%%%%%%%%%%%%%%%%%%%%%%%%%%%%%%%%%%%%%%%%%%%%%%%%%%%%%%%%%%%%%%%%%%%%%%%%%%%%%%%%%%%%%%%%%%%%%%%%%%%%%%%%%%%%%%%%%%%%%%%%%%%%%%%%%%%%%%%%%%%%%%%%%%%%%%%%%
\subsubsection{Example}
%%%%%%%%%%%%%%%%%%%%%%%%%%%%%%%%%%%%%%%%%%%%%%%%%%%%%%%%%%%%%%%%%%%%%%%%%%%%%%%%%%%%%%%%%%%%%%%%%%%%%%%%%%%%%%%%%%%%%%%%%%%%%%%%%%%%%%%%%%%%%%%%%%%%%%%%%%%
We illustrate an example to demonstrate the performance of the SASD method proposed in Section \ref{sec:sig_denoising}. In Fig. \ref{fig:sig_denoise}(a), we plot a noisy measurement signal segment, where the original signal contains a low-frequency sinusoid with two discontinuities that appear at sample indices 90 and 180, respectively \cite{selesnick2012_poly}. To denoise the measured signal with minimum root mean-square error ($\mathrm{RMSE}$), we first need to determine the cut-off frequency of the low-pass filter $\omega_0$ and design the zero-phase high-pass. An estimate of the cut-off frequency is obtained from the Fourier spectrum plot of the measurement data. The cut-off frequency is set to $\omega_1 = 0.044 \pi$ rad/s, and we use that value to design an sixth-order zero-phase high-pass Butterworth filter ($M = 3$), denoted by $\bm{H}^\tpose \bm{H}$. In the preprocessing step, we extrapolate the input sequence at the start and end by $P = 20$ samples, obtained using first-order polynomial approximations of the first $P$ and last $P$ samples, respectively. In Fig. \ref{fig:sig_denoise}(b)-(d), we plot the low-frequency signal $\bm{x}_1$, and the $K$-order sparse derivative signal $\bm{x}_2$ for $K = 1$, and regularization parameter $\lambda = 1.0$. An initial estimate of $\lambda$ is determined using the `three-sigma' rule \cite{selesnick2014}. In Fig. \ref{fig:sig_denoise}(b) and (c), we plot the reconstructed low-frequency signal obtained by solving (\ref{eqn:sass_lpf_apprx}) and sparse-derivative signal given as $\bm{x}_2 = \bm{S} \bm{v}$, where $\bm{S}$ is given in (\ref{eqn:reconst_matrix}), respectively. The noise free low-frequency signal preserves the smoothness of the input signal whereas the reconstructed signal $\bm{x}_2$ preserves the discontinuities. Finally, in Fig. \ref{fig:sig_denoise}(d), we plot the denoised input signal. The denoised signal $\bm{x}$ is obtained from (\ref{eqn:x_hat_estimate}), with $\mathrm{RMSE} = 0.063$.

{\dred To demonstrate the robustness of the proposed filter designs in the SASD signal model, we perform a Monte Carlo simulation across different orders of filter and noise levels. In Table \ref{tab:sasd}, we present the average values of the root-mean square error and convergence/computation time of the proposed and existing methods for $100$ realizations of each setting of the order of the filter and noise level. As can be seen, our proposed filter designs when applied in the sparsity-assisted signal smoothing (SASS) signal model demonstrates consistent performance across different orders of filter when the noise level is held constant. In contrast, the SASS method, although computationally efficient, demonstrates variable performance when the noise levels are held constant and the order of the filters are changed. Furthermore, the low-pass filtering cannot preserve the discontinues, and thus, the root-mean-square error obtained using low-pass filtering increases on increasing the noise levels. Similarly, the total variation denoising method introduces staircase-like artifacts and cannot preserve the smoothness of the signal, and therefore, demonstrates higher root-mean-square error than the proposed method.}
%%%%%%%%%%%%%%%%%%%%%%%%%%%%%%%%%%%%%%%%%%%%%%%%%%%%%%%%%%%%%%%%%%%%%%%%%%%%%%%%%%%%%%%%%%%%%%%%%%%%%%%%%%%%%%%%%%%%%%%%%%%%%%%%%%%%%%%%%%%%%%%%%%%%%%%%%%%
\subsubsection{\dred Example}
%%%%%%%%%%%%%%%%%%%%%%%%%%%%%%%%%%%%%%%%%%%%%%%%%%%%%%%%%%%%%%%%%%%%%%%%%%%%%%%%%%%%%%%%%%%%%%%%%%%%%%%%%%%%%%%%%%%%%%%%%%%%%%%%%%%%%%%%%%%%%%%%%%%%%%%%%%%
We illustrate the application of SASD algorithm for denoising real electrocardiogram (ECG) signal. The purpose of this example is to validate the proposed filter designs using real data and also demonstrate the matrix factorization method proposed in Section \ref{sec:factorize_zero_phase_mats} when the order of sparsity $K = 2$. The noisy signal $\bm{y}$, as shown in light gray color in Fig. \ref{fig:ecg_denoise}, consists of two PQRST segments of the electrical activity of the heart. Our goal is to denoise the ECG signal such that the PQRST segments are preserved. We begin by applying a low-pass filter with normalized cut-off frequency $\omega_0 = 0.02$ rads/s. The low-pass filter can denoise P and T segmetns; however, as shown in Fig. \ref{fig:ecg_denoise}(a), it cannot preserve the peak-to-peak voltage of the QRS complex. In Fig. \ref{fig:ecg_denoise}(b), we apply the TVD with a regularization parameter $\lambda = 20$ which was determined using a grid-search technique. The TVD method suppresses the noise and also preserves the peak-to-peak voltage in the QRS complex; however, it introduces staircase-like artifacts because it treats smooth segments of the signal as piecewise constants \cite{figueiredo2006}. On the contrary, both SASS and SASD, preserve smoothness and discontinuities because the cost function to be minimized accounts for filtering, and ramp-like signals when $K = 2$ (second-order sparse derivative), i.e., QRS complex.

\begin{figure}[t!]
\centering
\includegraphics[width=0.47\textwidth, height=0.47\textwidth]{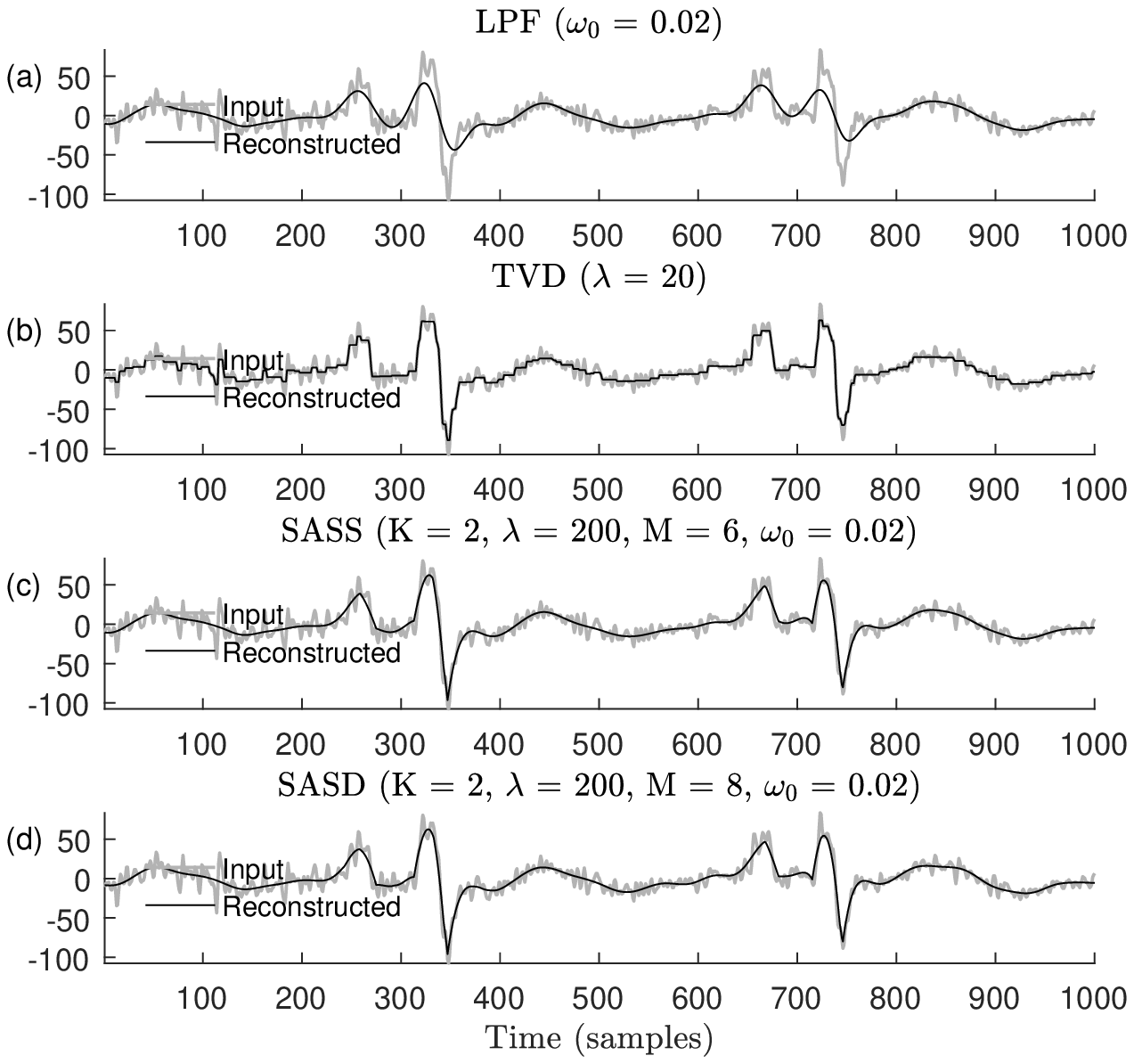}
\caption{ECG denoising. (a) Low-pass filtering with normalize cutoff frequency $\omega_0 = 0.02$ rads/s. (b) Total-variation denoising with $\lambda = 20$. (c) SASS using $\omega_0 = 0.02$ rads/s and $\lambda = 200$. (d) SASD using $\omega_0 = 0.02$ rads/s and $\lambda = 200$. \label{fig:ecg_denoise}}
\end{figure}
The output of the SASS algorithm in Fig. \ref{fig:ecg_denoise}(c) is obtained using $\lambda = 200$ and a sixth-order zero-phase high-pass Butterworth filter with normalized cutoff frequency $\omega_0 = 0.02$. The algorithm takes approximately $0.5$ seconds to converge. We keep the simulation parameters unchanged, i.e., $\lambda = 200$ and $\omega_0 = 0.02$, and plot the output of the SASD algorithm in Fig. \ref{fig:ecg_denoise}(d). Note that the cost function of SASD in (\ref{eqn:l1_min_cf_2}) does not use sparse banded matrices as filters and requires precomputed matrix $\bm{H}_1$ obtained using the matrix factorization method discussed in Section \ref{sec:factorize_zero_phase_mats}. The SASD algorithm takes approximately $5.0$ seconds to converge given that the factorized matrix, $\bm{H}_1$, obtained by solving (\ref{eqn:matrix_factor_opt_2}), is precomputed. Clearly, in terms of computational cost, the SASS algorithm outperforms the SASD because the zero-phase filters employed in the SASS algorithm are banded and sparse, which makes them computationally efficient. However, it is important to note that the proposed filter designs and matrix factorization method in Sections \ref{sec:zero_phase_mats} and \ref{sec:factorize_zero_phase_mats}, respectively, can be employed for signal denoising. 
\vfill\null 
%%%%%%%%%%%%%%%%%%%%%%%%%%%%%%%%%%%%%%%%%%%%%%%%%%%%%%%%%%%%%%%%%%%%%%%%%%%%%%%%%%%%%%%%%%%%%%%%%%%%%%%%%%%%%%%%%%%%%%%%%%%%%%%%%%%%%%%%%%%%%%%%%%%%%%%%%%%
\section{\dred Derivation: Sparsity-Assisted Pattern Recognition \label{appx_b}}
%%%%%%%%%%%%%%%%%%%%%%%%%%%%%%%%%%%%%%%%%%%%%%%%%%%%%%%%%%%%%%%%%%%%%%%%%%%%%%%%%%%%%%%%%%%%%%%%%%%%%%%%%%%%%%%%%%%%%%%%%%%%%%%%%%%%%%%%%%%%%%%%%%%%%%%%%%%
We apply the ADMM \cite[Chapter 3]{boyd2011} to solve the cost function in (\ref{eqn:cs_l1_min_cf_2}). We begin by decoupling the cost function in (\ref{eqn:cs_l1_min_cf_2}) using the `variable splitting' method. The cost can be written as
\begin{align}
\hspace{-1em}\arg \min_{\bm{k},\bm{u}_1}  &\biggl\{ \frac{1}{2} \norm{ \bm{H}^\tpose \bm{H} \bm{y} - \bm{B}^\tpose \bm{B} \boldsymbol\varPsi \bm{k} }_2^2 + \notag \\
& \qquad \qquad \qquad  \lambda_0 \norm{\bm{k}}_1 + \lambda_1 \norm{\bm{D} \boldsymbol\varPsi \bm{k}}_1 \biggr\} \notag \\
\mbox{such that}    &\quad \bm{u}_1 = \bm{k}. \label{eqn:var_split_kcom_model}
\end{align}
Applying ADMM to (\ref{eqn:var_split_kcom_model}) gives
\begin{subequations}
\begin{align}
\bm{u}_1 &\gets \arg \min_{\bm{u}_1} \biggl\{ \frac{1}{2} \norm{ \bm{H}^\tpose \bm{H} \bm{y} - \bm{B}^\tpose \bm{B} \boldsymbol\varPsi \bm{u}_1 }_2^2 + \biggr. \notag \\
         &\qquad \qquad \qquad \biggl. \frac{\mu}{2} \norm{\bm{u}_1 - \bm{k} - \bm{d}_1}_2^2  \biggr\} \label{eqn:kc_arg_min_u1} \\
\bm{k} &\gets \arg \min_{\bm{k}} \biggl\{ \lambda_0 \norm{\bm{k}}_1 + \lambda_1 \norm{\bm{D} \boldsymbol\varPsi \bm{k}}_1 + \notag  \\ 
    &\qquad \qquad \qquad \frac{\mu}{2} \norm{\bm{u}_1 - \bm{k} - \bm{d}_1}_2^2 \biggr\} \label{eqn:kc_arg_min_k} \\
\bm{d}_1 &\gets \bm{d}_1 - (\bm{u}_1 - \bm{k}) \label{eqn:update_u1}
\end{align}
\end{subequations}
To solve (\ref{eqn:kc_arg_min_u1}), we make the following substitutions: $\bm{M} =  \bm{B}^\tpose \bm{B} \boldsymbol\varPsi$, which simplifies (\ref{eqn:kc_arg_min_u1}) to
\begin{align}
\bm{u}_1 \gets \arg \min_{\bm{u}_1} \Bigl\{ \frac{1}{2} \norm{\bm{H}^\tpose \bm{H} \bm{y} - \bm{M} \bm{u}_1}^2_2 + \frac{\mu}{2} \norm{\bm{u}_1 - \bm{d}_1 - \bm{k}}^2_2 \Bigr\}. \notag
\end{align}
The above equation is a standard least squares problem whose solution is given as
\begin{align}
\bm{u}_1 \gets (\bm{M}^\tpose \bm{M} + \mu \bm{I})^{-1} (\bm{M}^\tpose \bm{H}^\tpose \bm{H} \bm{y} + \mu(\bm{k} + \bm{d}_1) ). \label{eqn:sol_arg_min_u}
\end{align}
Using the matrix inversion lemma \cite{harville1997matrix} to expand $(\bm{M}^\tpose \bm{M} + \mu \bm{I})^{-1}$, we can further simplify (\ref{eqn:sol_arg_min_u}) as
\begin{align}
\hspace{-0.5em}(\bm{M}^\tpose \bm{M} + \mu \bm{I})^{-1} = \frac{1}{\mu} \left[ \bm{I} - \bm{M}^\tpose (\mu \bm{I} + \bm{M} \bm{M}^\tpose)^{-1} \bm{M} \right] \label{eqn:matrix_inv_lemma}
\end{align}
Furthermore, the outer product of $\bm{M} \bm{M}^T = (\bm{B}^\tpose \bm{B})^2 $ where we used the generalized version of Parseval's identity for the $\mathrm{WDWT}$ operator, i.e., $\boldsymbol\Psi \boldsymbol\Psi^\tpose = \bm{I}$. Let $\bm{F}$ be defined as
\begin{align}
\bm{F} = \left[ \mu \bm{I} + (\bm{B}^\tpose \bm{B})^2 \right]^{-1}. \label{eqn:kc_matrix_inv}
\end{align}
Using (\ref{eqn:matrix_inv_lemma}), we can simplify (\ref{eqn:sol_arg_min_u}) as a two-step solution, which is implemented as
\begin{subequations}
\begin{align}
\bm{g}_1 &\gets \frac{1}{\mu} (\boldsymbol\varPsi^\tpose \bm{B}^\tpose \bm{B} \bm{H}^\tpose \bm{H} \bm{y} ) + (\bm{k} + \bm{d}_1)  \label{eqn:kc_sol_arg_min_g1}\\
\bm{u}_1 &\gets \bm{g}_1 - \boldsymbol\varPsi^\tpose \bm{B}^\tpose \bm{B} \bm{F} (\bm{B}^\tpose \bm{B} \boldsymbol\varPsi \bm{g}_1 ) \label{eqn:kc_sol_arg_min_u1}
\end{align}
\end{subequations}
To solve (\ref{eqn:kc_arg_min_k}), we again apply variable splitting and rewrite (\ref{eqn:kc_arg_min_k}) as
\begin{align}
&\bm{k} \gets \arg \min_{\bm{k}} \biggl\{ \frac{\mu}{2} \norm{\bm{u}_1 - \bm{d}_1 - \bm{k}}_2^2 + \notag \\ 
&\qquad \qquad \qquad \lambda_0 \norm{\bm{k}}_1 + \lambda_1 \norm{\bm{D} \boldsymbol\varPsi \bm{v}}_1 \biggr\} \notag \\
& \mbox{such that}  \quad \bm{v} = \bm{k}. \label{eqn:var_split_kc_arg_min_k}
\end{align}
Using the scaled augmented Lagrangian, we can minimize (\ref{eqn:var_split_kc_arg_min_k}) by developing the following iterative procedure
\begin{subequations}
\begin{align}
\bm{k} &\gets \arg \min_{\bm{k}} \biggl\{ \lambda_0 \norm{\bm{k}}_1 + \frac{\mu}{2} \norm{\bm{u}_1 - \bm{d}_1 - \bm{k}}_2^2 +  \notag \\
     &\qquad \qquad \qquad \frac{\eta}{2} \norm{\bm{v} - \bm{d}_2 - \bm{k}} \biggr\} \label{eqn:kc_arg_min_k1} \\
\bm{v} &\gets \arg \min_{\bm{v}} \biggl\{ \lambda_1 \norm{\bm{D} \boldsymbol\varPsi \bm{v}}_1 + \frac{\eta}{2} \norm{\bm{v} - \bm{d}_2 - \bm{k}} \biggr\} \label{eqn:kc_arg_min_v} \\
\bm{d}_2 &\gets \bm{d}_2 - (\bm{v} - \bm{k}) \label{eqn:update_d2}
\end{align}
\end{subequations}
To solve (\ref{eqn:kc_arg_min_k1}), we remove those terms that are not dependent on $\bm{k}$ from (\ref{eqn:kc_arg_min_k1}) and make the following substitution:
\begin{align}
\bm{p} = \left( \mu(\bm{u}_1 - \bm{d}_1) + \eta(\bm{v} - \bm{d}_2) \right)/(\mu + \eta) \label{eqn:kc_p}.
\end{align}
Therefore, (\ref{eqn:kc_arg_min_k1}) can be simplified as
\begin{align}
\bm{k} &\gets \arg \min_{\bm{k}} \biggl\{ \frac{\lambda_0}{\mu + \eta} \norm{\bm{k}}_1 + \frac{1}{2} \norm{\bm{p} - \bm{k}}^2_2 \biggr\}. \label{eqn:kc_sim_cf_k}
\end{align}
The solution of (\ref{eqn:kc_sim_cf_k}) is the solution to the least absolute shrinkage and selection operator (LASSO) problem \cite{donoho1995} and expressed as
\begin{align}
\bm{k} &\gets \mathrm{soft}(\bm{p},{\lambda_0}/({\mu + \eta})). \label{eqn:kc_k_cf_sol}
\end{align}
To solve (\ref{eqn:kc_arg_min_v}), we simplify the cost function as follows:
\begin{align}
\bm{v} &\gets \arg \min_{\bm{v}} \biggl\{ \lambda_1 \norm{\bm{D} \boldsymbol\varPsi \bm{v}}_1 + \frac{\eta}{2} \norm{\bm{v} - \bm{m}} \biggr\}  \\
     &= \arg \min_{\bm{v}} \biggl\{ \lambda_1 \Omega(\boldsymbol\varPsi \bm{v}) + \frac{\eta}{2} \norm{\bm{v} - \bm{m}} \biggr\} \\
     &= \arg \min_{\bm{v}} \biggl\{ \lambda_1 h(\bm{v}) + \frac{\eta}{2} \norm{\bm{v} - \bm{m}} \biggr\} \\
     &= \mathrm{prox}_{h}(\bm{m}) \\
     &= \bm{m} + \boldsymbol\varPsi^\tpose \left( \mathrm{prox}_{\Omega} (\boldsymbol\varPsi \bm{m}) - \boldsymbol\varPsi \bm{m} \right) \label{eqn:semi_ortho_tf}
\end{align}
where $\bm{m} = \bm{k} + \bm{d}_2$, $\Omega(\bm{x}) = \frac{\lambda_1}{\eta} \norm{\bm{D} \bm{x}}_1$, and $h(\bm{v}) = \Omega(\boldsymbol\varPsi \bm{v})$.
In (\ref{eqn:semi_ortho_tf}), we used the semi-orthonormal linear transform of the proximal operator \cite{combettes2011}. The proximal operator in (\ref{eqn:semi_ortho_tf}) can be further simplified as
\begin{align}
\mathrm{prox}_{\Omega} (\boldsymbol\varPsi \bm{m}) &= \arg \min_{\bm{x}} \biggl\{ \Omega(\bm{x}) + \frac{1}{2} \norm{\bm{x} - \boldsymbol\varPsi \bm{m}}_2^2 \biggr\} \notag \\
                           &= \arg \min_{\bm{x}} \biggl\{ \frac{\lambda_1}{\eta} \norm{\bm{D} \bm{x}}_1 + \frac{1}{2} \norm{\bm{x} - \boldsymbol\varPsi \bm{m}}_2^2 \biggr\} \notag \\
                           &= \mathrm{tvd}(\boldsymbol\varPsi \bm{m}, \lambda_1/\eta).
\end{align}
where $\mathrm{tvd}(\cdot,\cdot)$ represents the solution to the total-variation denoising problem \cite{rinaldo2009,tibshirani2009}. Therefore, the solution to the optimization problem in (\ref{eqn:kc_arg_min_v}) can be written in two-steps as
\begin{subequations}
\begin{align}
\bm{m} &\gets \bm{d}_2 + \bm{k} \label{eqn:sol_v_1} \\
\bm{v} &\gets \bm{m} + \boldsymbol\varPsi^\tpose \left( \mathrm{tvd}(\boldsymbol\varPsi \bm{m}, \lambda_1/\eta) - \boldsymbol\varPsi \bm{m} \right) \label{eqn:sol_v_2}.
\end{align}
\end{subequations}
The details of implementing the SAPR algorithm is listed as Algorithm \ref{alg:sapr_kcomp}.
%%%%%%%%%%%%%%%%%%%%%%%%%%%%%%%%%%%%%%%%%%%%%%%%%%%%%%%%%%%%%%%%%%%%%%%%%%%%%%%%%%%%%%%%%%%%%%%%%%%%%%%%%%%%%%%%%%%%%%%%%%%%%%%%%%%%%%%%%%%%%%%%%%%%%%%%%%%
\begin{algorithm}[t!]
\small
\setstretch{1.05}
\caption{Sparsity-Assisted Pattern Recognition (\ref{eqn:cs_l1_min_cf_2})}\label{alg:sapr_kcomp}
\begin{algorithmic}
\Procedure{SAPR}{$\bm{y}$, $\bm{H}$, $\bm{B}$, $\lambda_0$, $\lambda_1$, $\mu$, $\eta$}
\State \textbf{initialize}
\State $\bm{F} = \left[ \mu \bm{I} + (\bm{B}^\tpose \bm{B})^2 \right]^{-1}$ \Comment{From (\ref{eqn:kc_matrix_inv})},
\State $\bm{k} \gets \boldsymbol\varPsi^\tpose \bm{B}^\tpose \bm{B} \bm{y}$, $\bm{v} \gets \bm{k}$
\State $\bm{d}_1 \gets \bm{0}$, $\bm{d}_2 \gets \bm{0}$
\State $\bm{b}_1 \gets (1/\mu) \boldsymbol\varPsi^\tpose \bm{B}^\tpose \bm{B} \bm{H}^\tpose \bm{H} \bm{y}$
\Repeat
  \State $\bm{g}_1 \gets \bm{b}_1 + \bm{k} + \bm{d}_1$ \Comment{From (\ref{eqn:kc_sol_arg_min_g1})}
  \State $\bm{u}_1 \gets \bm{g}_1 - \boldsymbol\varPsi^\tpose \bm{B}^\tpose \bm{B} \bm{F} (\bm{B}^\tpose \bm{B} \boldsymbol\varPsi \bm{g}_1 )$ \Comment{From (\ref{eqn:kc_sol_arg_min_u1})}
  \State $\bm{p} \gets { \mu(\bm{u}_1 - \bm{d}_1) + \eta(\bm{v} - \bm{d}_2) }/{(\mu + \eta)}$ \Comment{From (\ref{eqn:kc_p})}
  \State $\bm{k} \gets \mathrm{soft}\left(\bm{p},{\lambda_0}/{(\mu + \eta)}\right)$ \Comment{From (\ref{eqn:kc_k_cf_sol})}
  \State $\bm{m} \gets \bm{d}_2 + \bm{k}$ \Comment{From (\ref{eqn:sol_v_1})}
  \State $\bm{v} \gets \bm{m} + \boldsymbol\varPsi^\tpose \left( \mathrm{tvd} \left(\boldsymbol\varPsi \bm{m}, {\lambda_1}/{\eta} \right) - \boldsymbol\varPsi \bm{m} \right)$ \Comment{From (\ref{eqn:sol_v_2})}
  \State $\bm{d}_1 \gets \bm{d}_1 - (\bm{u}_1 - \bm{k})$ \Comment{From (\ref{eqn:update_u1})}
  \State $\bm{d}_2 \gets \bm{d}_2 - (\bm{v} - \bm{k})$ \Comment{From (\ref{eqn:update_d2})}
\Until{convergence}
\State \Return $\bm{B}^\tpose \bm{B} \boldsymbol\varPsi\bm{k}$
\EndProcedure
\end{algorithmic}
\end{algorithm}
%%%%%%%%%%%%%%%%%%%%%%%%%%%%%%%%%%%%%%%%%%%%%%%%%%%%%%%%%%%%%%%%%%%%%%%%%%%%%%%%%%%%%%%%%%%%%%%%%%%%%%%%%%%%%%%%%%%%%%%%%%%%%%%%%%%%%%%%%%%%%%%%%%%%%%%%%%%
%%%%%%%%%%%%%%%%%%%%%%%%%%%%%%%%%%%%%%%%%%%%%%%%%%%%%%%%%%%%%%%%%%%%%%%%%%%%%%%%%%%%%%%%%%%%%%%%%%%%%%%%%%%%%%%%%%%%%%%%%%%%%%%%%%%%%%%%%%%%%%%%%%%%%%%%%%%
\section{\dred K-Complex Detection\label{sec:param_kcomplex}}
%%%%%%%%%%%%%%%%%%%%%%%%%%%%%%%%%%%%%%%%%%%%%%%%%%%%%%%%%%%%%%%%%%%%%%%%%%%%%%%%%%%%%%%%%%%%%%%%%%%%%%%%%%%%%%%%%%%%%%%%%%%%%%%%%%%%%%%%%%%%%%%%%%%%%%%%%%%
\subsection{Parameter Selection}
%%%%%%%%%%%%%%%%%%%%%%%%%%%%%%%%%%%%%%%%%%%%%%%%%%%%%%%%%%%%%%%%%%%%%%%%%%%%%%%%%%%%%%%%%%%%%%%%%%%%%%%%%%%%%%%%%%%%%%%%%%%%%%%%%%%%%%%%%%%%%%%%%%%%%%%%%%%
In this Appendix, we develop a methodology to tune the regularization parameters to solve the optimization problem in (\ref{eqn:cs_l1_min_cf_2}). We begin by the discretizing the regularization parameters so that $\lambda_0 \in [100,160]$ and $\lambda_1 \in [10,70]$, in step sizes of $5$ and $5$, respectively. Because performing a grid search on the entire K-complex EEG database in \cite{devuyst2010} is computationally expensive, we choose a small set of epochs to determine the feasible operating region of the regularization parameters. We select one epoch from each dataset, such that the minimum sampling rate of the dataset is $200$ Hz, the dataset is annotated by at least two experts, and the selected $30$ second epoch contains maximum number of K-complexes. Note that the length of each epoch is determined based on an established scoring criteria \cite{berry2012aasm} used by experts when scoring K-complexes manually.

\begin{figure}[t!]
\centering
\includegraphics[width=0.49\textwidth, height=0.32\textwidth]{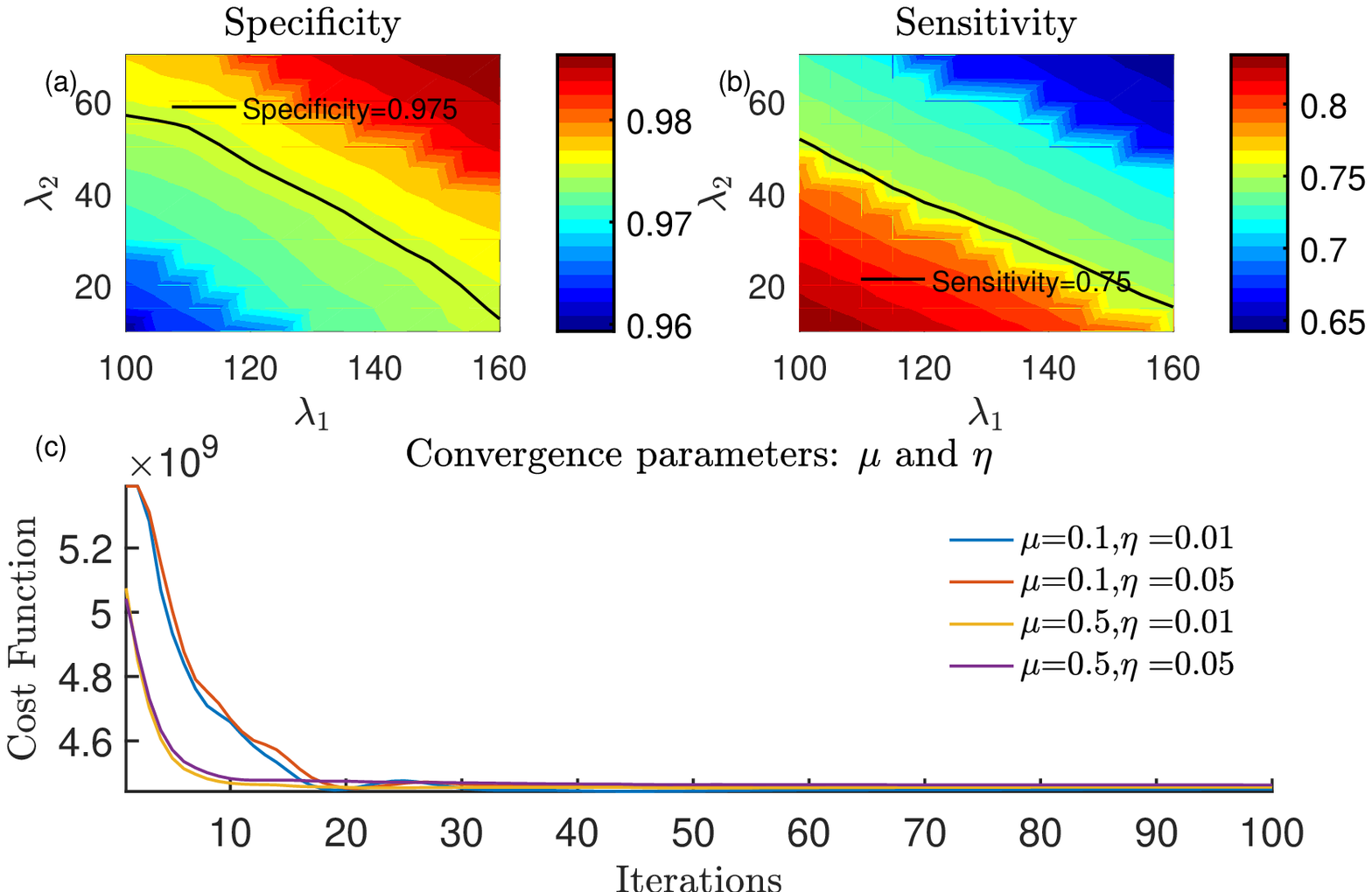}
\caption{Regularization parameters for SAPR algorithm. (a) Specificity plot across all values of $\lambda_1$ and $\lambda_2$. (b) Cost function across different values of $\mu$ and $\eta$.\label{fig:rparam_kcomp}}
\end{figure}
Next, to determine a feasible operating region of the regularization parameters, we perform a grid search using the selected epochs (total of five epochs where each epoch is of $30$ seconds length, i.e., $N = 6000$ samples). To evaluate the performance of the SAPR method, we compute the specificity and sensitivity across all sample points of the selected epochs for all regularization parameters. In Fig. \ref{fig:rparam_kcomp}(a) and (b), we plot the average value of the sensitivity and specificity across all $\lambda_1$ and $\lambda_2$. As can be seen, regions consisting of high values of specificity demonstrate low values of sensitivity and vice versa. To find a good balance between specificity and sensitivity, we use contour plot (solid black line in Fig. \ref{fig:rparam_kcomp}(a) and (b)) to indicate the regions where the specificity and sensitivity are $0.975$ and $0.75$, respectively. The contour represents the feasible operating region of the regularization parameters. Thereafter, to determine the parameter $\mu$ and $\eta$ which affect the rate of convergence of the SAPR algorihtm, we compute the average value of the cost function in (\ref{eqn:l1_min_cf_2}) for the selected epochs for different combinations of $\mu$ and $\eta$. Note that the parameters $\mu$ and $\eta$ does not affect the final value of the cost function. As can be seen in Fig. \ref{fig:rparam_kcomp}(c), the algorithm converges fastest when $\mu = 0.5$ and $\eta = 0.01$.

\begin{figure}[t!]
\centering
\includegraphics[width=0.49\textwidth, height=0.15\textwidth]{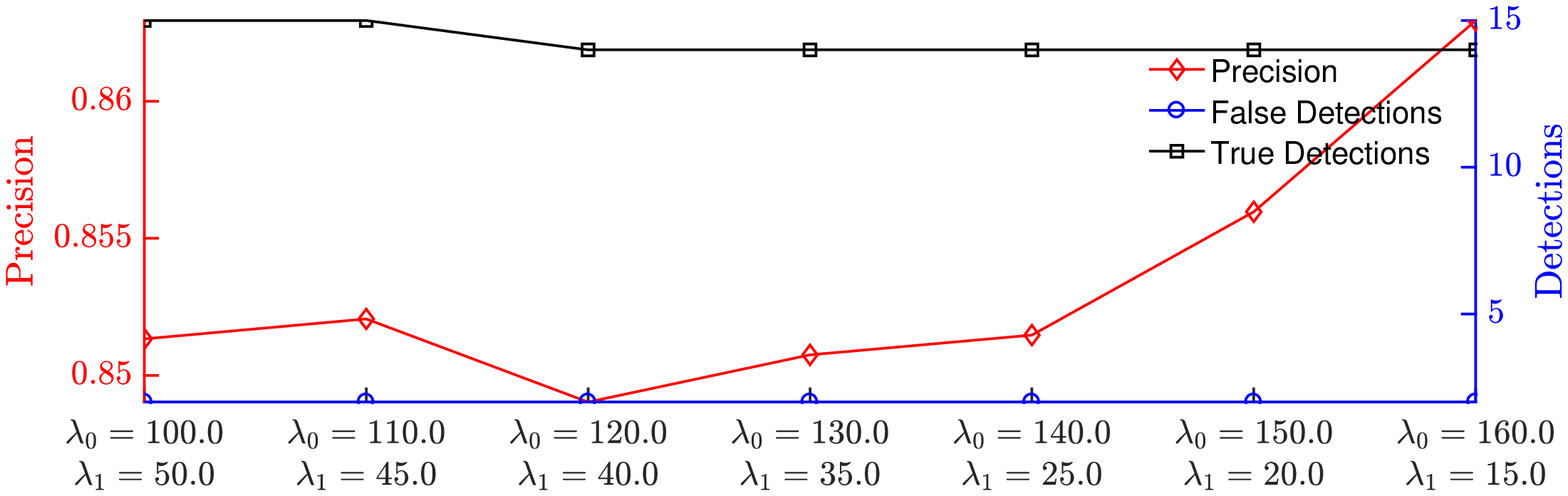}
\caption{Precision, false detections, and true detections across different values of $\lambda_1$ and $\lambda_2$. \label{fig:rparam_kcomp_1}}
\end{figure}
Using Fig. \ref{fig:rparam_kcomp_1}, we reduced the search space of the regularization parameters from a two dimensional space to a feasible operating region represented as the space between the two solid lines in Fig. \ref{fig:rparam_kcomp}(a) and (b). Finally, to find an optimal value of the regularization parameters $\lambda_1$ and $\lambda_2$ in the feasible operating region, and compute various performance measures across different values of $\lambda_1$ and $\lambda_2$ for the selected epochs, and plot these measures in Fig. \ref{fig:rparam_kcomp_1}. As can be seen, the number of correctly and falsely detected events does not vary across the regularization parameters. However, the average value of precision across all sample points of the selected epochs increases on increasing the regularization parameters. To find a good balance between the performance metrics identified in Fig. \ref{fig:rparam_kcomp_1}, we select $\lambda_1 = 160.0$ and $\lambda_2 = 15.0$. Note that we determined the regularization parameters using selected epochs of the K-complex database \cite{devuyst2010}. For our sample, there were four females and the mean age $27.40 \pm 11.055$ years. The peak-to-peak voltage of the K-complex signal depends on the average age group of the cohort. For instance, an elderly cohort tends to generate a lower average peak-to-peak voltage than a young cohort. In such scenarios, the regularization parameters can be determined using the methodology in described in Appendix \ref{sec:param_kcomplex}, or by rescaling the amplitude of the original input signal to match the average peak-to-peak voltage of the \cite{devuyst2010} database and using the optimal regularization parameters of the K-complex database in \cite{devuyst2010}.
%%%%%%%%%%%%%%%%%%%%%%%%%%%%%%%%%%%%%%%%%%%%%%%%%%%%%%%%%%%%%%%%%%%%%%%%%%%%%%%%%%%%%%%%%%%%%%%%%%%%%%%%%%%%%%%%%%%%%%%%%%%%%%%%%%%%%%%%%%%%%%%%%%%%%%%%%%%
\subsection{Examples}
%%%%%%%%%%%%%%%%%%%%%%%%%%%%%%%%%%%%%%%%%%%%%%%%%%%%%%%%%%%%%%%%%%%%%%%%%%%%%%%%%%%%%%%%%%%%%%%%%%%%%%%%%%%%%%%%%%%%%%%%%%%%%%%%%%%%%%%%%%%%%%%%%%%%%%%%%%%
In Fig. \ref{fig:kcomplex_detoks}, we plot the output of the DETOKS algorithm for K-complex detection using sleep-EEG data from \texttt{excerpt4.edf} dataset. The DETOKS algorithm detects the K-complex regions accurately. However, it cannot separate K-complexes that appear close to each other. In addition, the TEKO is non-zero in regions where there is no K-complex signal. In contrast, as shown in Fig. \ref{fig:kcomplex}, the SAPR algorithm generates individual peaks for every K-complex pattern it detects. Further, the TKEO output generated using the SAPR algorithm is zero when there is no K-complex signal.
\begin{figure}[t!]
\centering
\includegraphics[width=0.48\textwidth, height=0.47\textwidth]{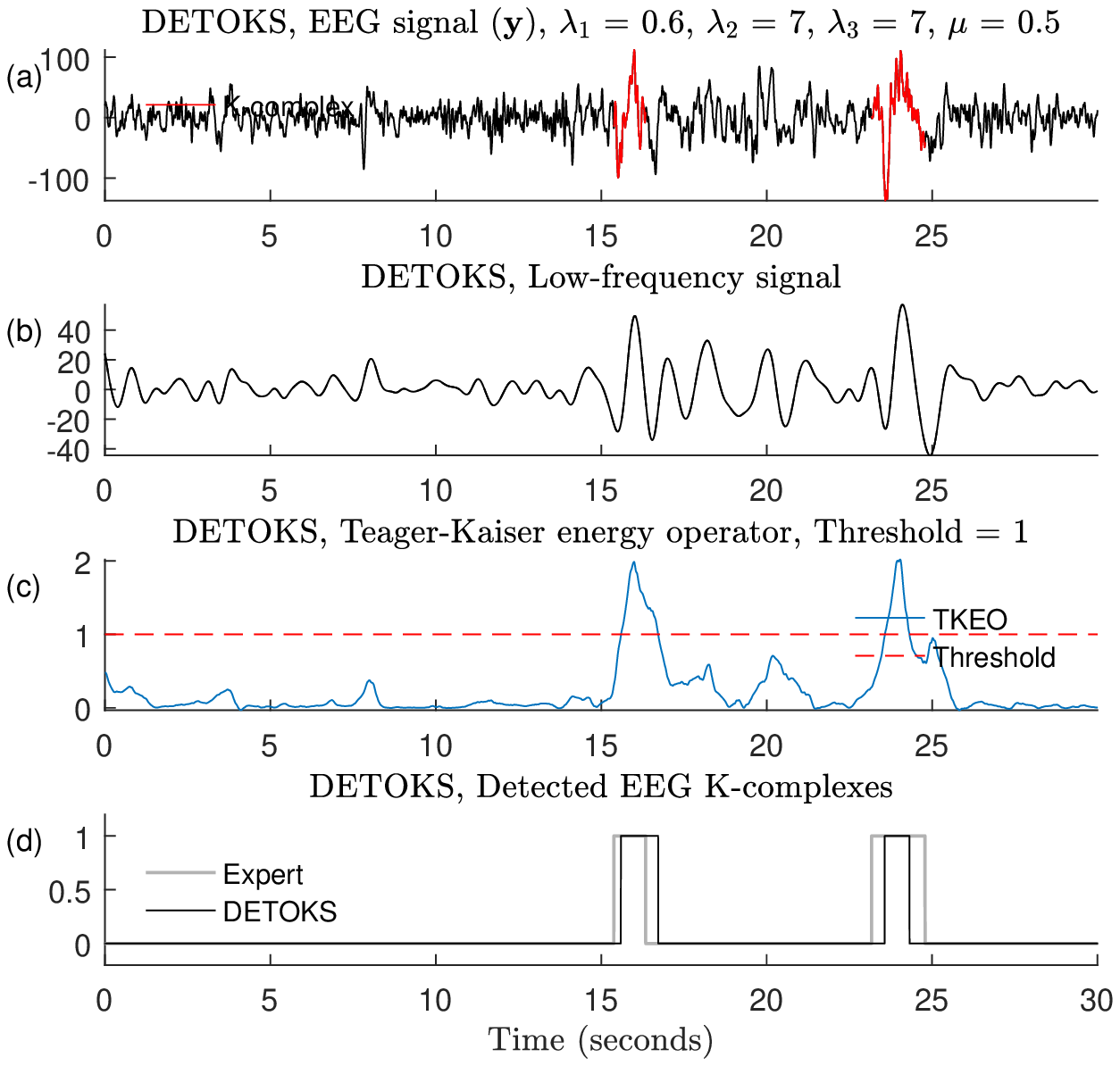}
\caption{DETOKS for K-complex detection. (a) 30 second epoch of sleep-EEG data obtained from \texttt{excerpt4.edf}. The epoch consists of two K-complexes as identified by two experts. (b) Low-pass filtered signal obtained after removing the oscillatory and transient components. (c) Signal obtained by applying the Teager-Kaiser energy operator on low-frequency signal. (d) Expert and algorithm annotated K-complex regions \label{fig:kcomplex_detoks}}
\end{figure}

In Fig. \ref{fig:kcomplex_sapr}, we plot the output of the SAPR algorithm for \texttt{excerpt1.edf}. The sleep-EEG epoch in Fig. \ref{fig:kcomplex}(a) consists of two K-complex regions as annotated by experts and slow wave activity around $6-12$ seconds. The AASM scoring manual defines slow wave activity as ``waves of frequency $0.5-2.0$ Hz and peak-to-peak amplitude $>75$ $\mu$V, measured over the frontal regions referenced to the contralateral ear or mastoid. K-complexes would be considered slow waves if they meet the definition of slow wave activity.'' To minimize the number of false detections caused by slow wave activity, we use an upper threshold of $2.25$ seconds to remove the slow wave activity. In addition, if the $\rm TKEO$ energy is above a certain fixed threshold, and two or more peaks representing the K-complex patterns that are closely separated (within $1.5$ seconds duration) are detected, then we select the first peak. By selecting the first peak, we are minimizing the number of false detections caused due to slow wave sleep which belongs to the same frequency band as the K-complex signal.
\begin{figure}[h!]
\centering
\includegraphics[width=0.48\textwidth, height=0.47\textwidth]{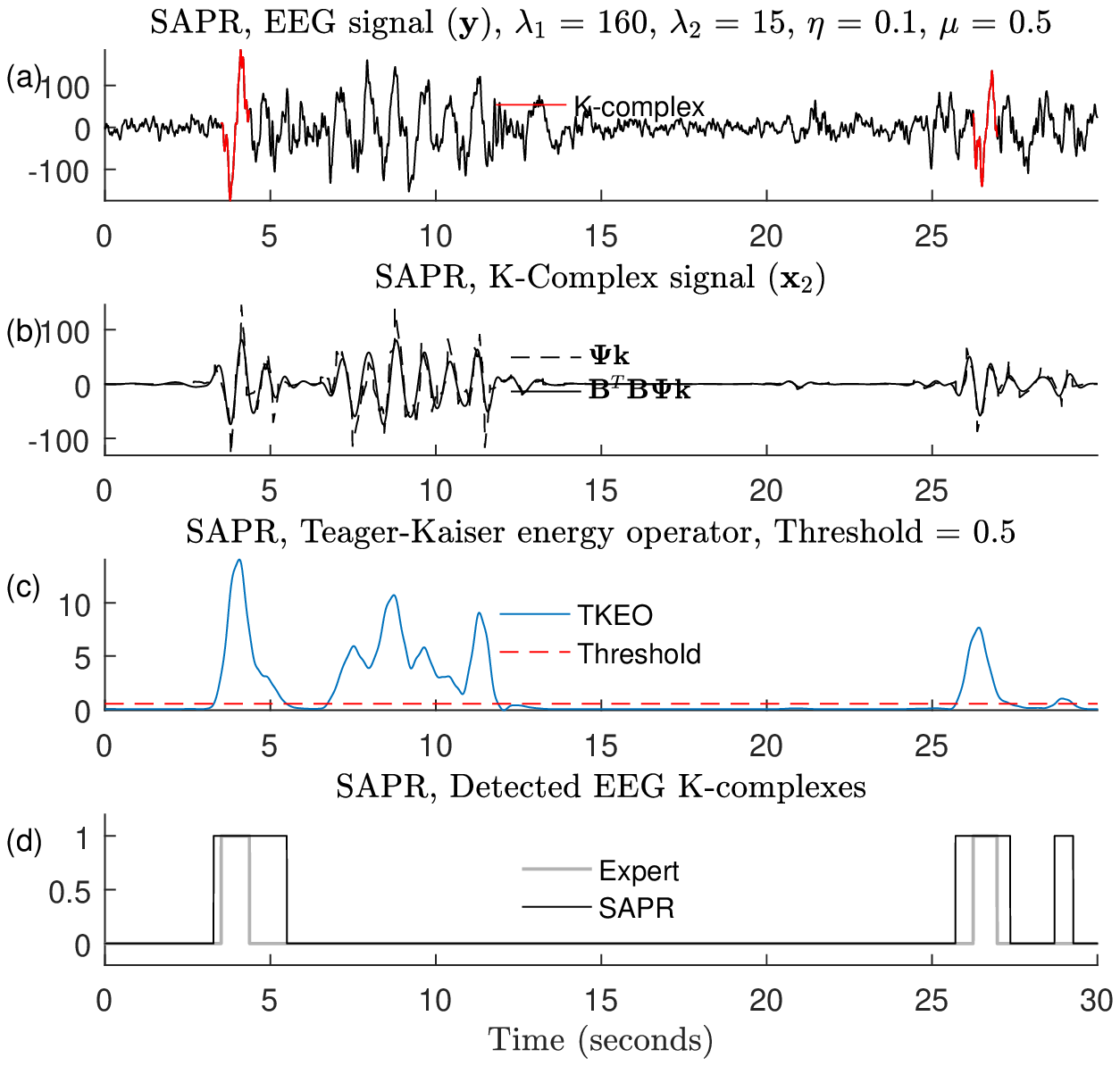}
\caption{SAPR for K-complex detection. (a) 30 second epoch of sleep-EEG data obtained from \texttt{excerpt1.edf}. The epoch consists of two K-complexes as identified by two experts. (b) Reconstructed K-complex signal $\boldsymbol\varPsi \bm{k}$ and its corresponding band-pass filtered component $\bm{B} \bm{B}^\tpose \boldsymbol\varPsi \bm{k}$ obtained using the SAPR method. (c) Signal obtained by applying the Teager-Kaiser energy operator on $\bm{B} \bm{B}^\tpose \boldsymbol\varPsi \bm{k}$. (d) Expert and algorithm annotated K-complex regions \label{fig:kcomplex_sapr}}
\end{figure}
\vfill\null 
%%%%%%%%%%%%%%%%%%%%%%%%%%%%%%%%%%%%%%%%%%%%%%%%%%%%%%%%%%%%%%%%%%%%%%%%%%%%%%%%%%%%%%%%%%%%%%%%%%%%%%%%%%%%%%%%%%%%%%%%%%%%%%%%%%%%%%%%%%%%%%%%%%%%%%%%%%%
\section{\dred Derivation: Sparsity-Assisted Signal Denoising and Pattern Recognition \label{appx_spindle}}
%%%%%%%%%%%%%%%%%%%%%%%%%%%%%%%%%%%%%%%%%%%%%%%%%%%%%%%%%%%%%%%%%%%%%%%%%%%%%%%%%%%%%%%%%%%%%%%%%%%%%%%%%%%%%%%%%%%%%%%%%%%%%%%%%%%%%%%%%%%%%%%%%%%%%%%%%%%
We apply the ADMM \cite[Chapter 4]{boyd2011} to solve the cost function in (\ref{eqn:spin_l1_min_cf_2}). We begin by decoupling the cost function in (\ref{eqn:spin_l1_min_cf_2}) using the `variable splitting' method. The cost can be written as
\begin{align}
\arg \min_{\bm{c}, \bm{x}_3, \bm{u}_1, \bm{u}_2} &\quad \biggl\{ \frac{1}{2} \norm{ \bm{H}^\tpose \bm{H} \left( \bm{y} - \bm{x}_3 \right) - \bm{B}^\tpose \bm{B} \boldsymbol\varPhi \bm{c} }_2^2 + \notag \\ 
                                          & \qquad \qquad \lambda_0 \norm{\bm{c}}_1 + \lambda_1 \norm{ \bm{D} \bm{x}_3}_1 + \lambda_2 \norm{\bm{x}_3}_1 \biggr\}, \notag \\
                                          \mbox{such that} &\quad \bm{u}_1 = \bm{c}, \quad \bm{u}_2 = \bm{x}_3. \label{eqn:var_split_spin_model}
\end{align}
Using scaled augmented Lagrangian, we can further minimize (\ref{eqn:var_split_spin_model}) using an iterative procedure as follows:
\begin{subequations}
\begin{align}
\bm{u}_1, \bm{u}_2 &\gets \arg \min_{\bm{u}_1, \bm{u}_2} \biggl\{ \frac{1}{2} \norm{ \bm{H}^\tpose \bm{H} \left( \bm{y} - \bm{u}_2 \right) - \bm{B}^\tpose \bm{B} \boldsymbol\varPhi \bm{u}_1 }_2^2 + \notag \\
                   &\quad \frac{\mu}{2} \norm{\bm{u}_1 - \bm{d}_1 - \bm{c}}_2^2 + \frac{\mu}{2} \norm{\bm{u}_2 - \bm{d}_2 - \bm{x}_3}_2^2 \biggr\} \label{eqn:u1_u2_cf}\\
\bm{c}, \bm{x}_3   &\gets \arg \min_{\bm{c}, \bm{x}_3} \biggl\{ \lambda_0 \norm{\bm{c}}_1 + \lambda_1 \norm{ \bm{D} \bm{x}_3}_1 + \lambda_2 \norm{\bm{x}_3}_1 + \notag \\
                   &\quad \frac{\mu}{2} \norm{\bm{u}_1 - \bm{d}_1 - \bm{c}}_2^2 + \frac{\mu}{2} \norm{\bm{u}_2 - \bm{d}_2 - \bm{x}_3}_2^2 \biggr\} \label{eqn:c_x3_cf}\\
\bm{d}_1           &\gets \bm{d}_1 - (\bm{u}_1 - \bm{c})    \label{eqn:u1_cf}\\
\bm{d}_2           &\gets \bm{d}_2 - (\bm{u}_2 - \bm{x}_3)  \label{eqn:u2_cf}
\end{align}
\end{subequations}
where $\mu > 0$ is an auxillary variable.  To solve (\ref{eqn:u1_u2_cf}), we make the following substituion:
\begin{align}
\bm{u} = \begin{bmatrix} \bm{u}_1 \\ \bm{u}_2 \end{bmatrix}, \quad \bm{d} = \begin{bmatrix} \bm{d}_1 \\ \bm{d}_2 \end{bmatrix}, \quad \bm{m} = \begin{bmatrix} \bm{c} \\ \bm{x}_3 \end{bmatrix},
\end{align} 
and $\bm{M} = [\bm{B}^\tpose \bm{B} \boldsymbol\varPhi \quad \bm{H}^\tpose \bm{H}]$. The optimization problem in (\ref{eqn:u1_u2_cf}) can be rewritten as
\begin{align}
\bm{u} &\gets \arg \min_{\bm{u}} \biggl\{ \frac{1}{2} \norm{\bm{H}^\tpose \bm{H} \bm{y} - \bm{M} \bm{u} }_2^2 + \norm{\bm{u} - \bm{d} - \bm{m} }_2^2 \biggr\} \label{eqn:sim_arg_min_u}
\end{align}
Equation (\ref{eqn:sim_arg_min_u}) is a standard least squares problem whose solution is given as
\begin{align}
\bm{u} \gets (\mu \bm{I} + \bm{M}^\tpose \bm{M} )^{-1} (\bm{M}^\tpose \bm{H}^\tpose \bm{H} \bm{y} + \mu(\bm{d} + \bm{m}) ). \label{eqn:sol_arg_min_u}
\end{align}
Using the matrix inversion lemma \cite{harville1997matrix} to expand $(\bm{M}^\tpose \bm{M} + \mu \bm{I})^{-1}$, we can further simplify (\ref{eqn:sol_arg_min_u}) as
\begin{align}
\hspace{-0.5em}(\bm{M}^\tpose \bm{M} + \mu \bm{I})^{-1} = \frac{1}{\mu} \left[ \bm{I} - \bm{M}^\tpose (\mu \bm{I} + \bm{M} \bm{M}^\tpose)^{-1} \bm{M} \right] \label{eqn:matrix_inv_lemma}
\end{align}
The outer product of $\bm{M} \bm{M}^T = (\bm{H}^\tpose \bm{H})^2 + (\bm{B}^\tpose \bm{B})^2 $, obtained by using the generalized version of Parseval's identity for the $\mathrm{STFT}$ operator. Define
\begin{align}
\bm{F} = \left[ \mu \bm{I} + (\bm{B}^\tpose \bm{B})^2 + (\bm{H}^\tpose \bm{H})^2 \right]^{-1}. \label{eqn:matrix_inverted}
\end{align}
Using (\ref{eqn:matrix_inv_lemma}), we can simplify (\ref{eqn:sol_arg_min_u}) as a two-step solution, which is implemented as
\begin{subequations}
\begin{align}
\bm{g}_1 &\gets (\bm{c} + \bm{d}_1) + \frac{1}{\mu} (\boldsymbol\varPhi^\tpose \bm{B}^\tpose \bm{B} \bm{H}^\tpose \bm{H} \bm{y} )              \label{eqn:sol_g1} \\
\bm{g}_2 &\gets (\bm{x}_3 + \bm{d}_2) + \frac{1}{\mu} ( (\bm{H}^\tpose \bm{H})^2 \bm{y}) \label{eqn:sol_g2} \\
\bm{g} &\gets (\bm{B}^\tpose \bm{B} \boldsymbol\varPhi \bm{g}_1 + \bm{H}^\tpose \bm{H} \bm{g}_1) \label{eqn:sol_g1_g2} \\
\bm{u}_1 &\gets \bm{g}_1 - \boldsymbol\varPhi^\tpose \bm{B}^\tpose \bm{B} \bm{F} \bm{g} \label{eqn:sol_u1} \\
\bm{u}_2 &\gets \bm{g}_2 - \bm{H}^\tpose \bm{H} \bm{F} \bm{g} \label{eqn:sol_u2}
\end{align}
\end{subequations}
The optimization problem in (\ref{eqn:c_x3_cf}) can be further simplified as
\begin{subequations}
\begin{align}
\bm{c}   &\gets \arg \min_{\bm{c}} \biggl\{ \lambda_0 \norm{\bm{c}}_1 + \frac{\mu}{2} \norm{\bm{u}_1 - \bm{d}_1 - \bm{c}}_2^2 \biggr\}, \label{eqn:c_cf} \\
\bm{x}_3 &\gets \arg \min_{\bm{x}_3} \biggl\{ \lambda_1 \norm{ \bm{D} \bm{x}_3}_1 + \lambda_2 \norm{\bm{x}_3}_1 + \notag \\ 
         &\qquad \qquad \qquad \frac{\mu}{2} \norm{\bm{u}_2 - \bm{d}_2 - \bm{x}_3}_2^2  \biggr\} \label{eqn:x3_cf}.
\end{align} 
\end{subequations}
The solution of (\ref{eqn:c_cf}) is the solution to the least absolute shrinkage and selection operator (LASSO) problem \cite{donoho1995} and expressed as
\begin{align}
\bm{c} &\gets \mathrm{soft}(\bm{u}_1 - \bm{d}_1,\lambda_0/\mu) \label{eqn:c_cf_sol}
\end{align}
where the soft-threshold function is defined as
\begin{align}
\mathrm{soft}(x,T) \triangleq \begin{cases} x - T(x/|x|), &\quad |x| > T \\
                      0,        &\quad |x| \le T \end{cases}.
\end{align}
The solution of (\ref{eqn:x3_cf}) is the solution to the fused LASSO \cite{rinaldo2009,tibshirani2009} and expressed as
\begin{align}
\bm{x}_3 &\gets \mathrm{soft}\left(\mathrm{tvd}\left(\bm{u}_2 - \bm{d}_2,\lambda_1/\mu\right),\lambda_2/\mu\right) \label{eqn:sol_x3}
\end{align}
where $\mathrm{tvd}(\cdot,\cdot)$ represents the solution to the total-variation denoising problem. The details of implementing the SASDPR algorithm is listed as Algorithm \ref{alg:sasdpr_spindle}.
%%%%%%%%%%%%%%%%%%%%%%%%%%%%%%%%%%%%%%%%%%%%%%%%%%%%%%%%%%%%%%%%%%%%%%%%%%%%%%%%%%%%%%%%%%%%%%%%%%%%%%%%%%%%%%%%%%%%%%%%%%%%%%%%%%%%%%%%%%%%%%%%%%%%%%%%%%%
\begin{algorithm}[t!]
\small
\setstretch{1.05}
\caption{\small Sparsity-Assisted Signal Denoising and Pattern Recognition (\ref{eqn:spin_l1_min_cf_2})}\label{alg:sasdpr_spindle}
\begin{algorithmic}
\Procedure{SASDPR}{$\bm{y}$, $\bm{H}$, $\bm{B}$, $\lambda_0$, $\lambda_1$, $\lambda_2$, $\mu$}
\State \textbf{initialize}
\State $\bm{F} = \left[ \mu \bm{I} + (\bm{B}^\tpose \bm{B})^2 + (\bm{H}^\tpose \bm{H})^2 \right]^{-1}$ \Comment{From (\ref{eqn:matrix_inverted})},
\State $\bm{c} \gets \boldsymbol\varPhi^\tpose \bm{B}^\tpose \bm{B} \bm{y}$, $\bm{x}_3 \gets \bm{H}^\tpose \bm{H} \bm{y}$
\State $\bm{d}_1 \gets \bm{0}$, $\bm{d}_2 \gets \bm{0}$
\State $\bm{b}_1 \gets (1/\mu) \boldsymbol\varPhi^\tpose \bm{B}^\tpose \bm{B} \bm{H}^\tpose \bm{H} \bm{y}$, $\bm{b}_2 \gets (1/\mu) (\bm{H}^\tpose \bm{H})^2 \bm{y}$ 
\Repeat
  \State $\bm{g}_1 \gets \bm{b}_1 + \bm{c} + \bm{d}_1$ \Comment{From (\ref{eqn:sol_g1})}
  \State $\bm{g}_2 \gets \bm{b}_2 + \bm{x}_3 + \bm{d}_2$ \Comment{From (\ref{eqn:sol_g2})}
  \State $\bm{g} \gets \bm{B}^\tpose \bm{B} \boldsymbol\varPhi \bm{g}_1 + \bm{H}^\tpose \bm{H} \bm{g}_1$ \Comment{From (\ref{eqn:sol_g1_g2})}
  \State $\bm{u}_1 \gets \bm{g}_1 - \boldsymbol\varPhi^\tpose \bm{B}^\tpose \bm{B} \bm{F} \bm{g}$ \Comment{From (\ref{eqn:sol_u1})} 
  \State $\bm{u}_2 \gets \bm{g}_2 - \bm{H}^\tpose \bm{H} \bm{F} \bm{g}$ \Comment{From (\ref{eqn:sol_u2})}
  \State $\bm{c} \gets \mathrm{soft}\left(\bm{u}_1 - \bm{d}_1, {\lambda_0}/{\mu} \right)$ \Comment{From (\ref{eqn:c_cf_sol})}
  \State $\bm{x}_3 \gets \mathrm{soft}\left(\mathrm{tvd}\left(\bm{u}_2 - \bm{d}_2,{\lambda_1}/{\mu}\right),{\lambda_2}/{\mu}\right)$ \Comment{From (\ref{eqn:sol_x3})}
  \State $\bm{d}_1 \gets \bm{d}_1 - (\bm{u}_1 - \bm{c})$ \Comment{From (\ref{eqn:u1_cf})} 
  \State $\bm{d}_2 \gets \bm{d}_2 - (\bm{u}_2 - \bm{x}_3)$  \Comment{From (\ref{eqn:u2_cf})}
\Until{convergence}
\State \Return $\bm{B}^\tpose \bm{B} \boldsymbol\varPhi\bm{c}$
\EndProcedure
\end{algorithmic}
\end{algorithm}
%%%%%%%%%%%%%%%%%%%%%%%%%%%%%%%%%%%%%%%%%%%%%%%%%%%%%%%%%%%%%%%%%%%%%%%%%%%%%%%%%%%%%%%%%%%%%%%%%%%%%%%%%%%%%%%%%%%%%%%%%%%%%%%%%%%%%%%%%%%%%%%%%%%%%%%%%%%
%%%%%%%%%%%%%%%%%%%%%%%%%%%%%%%%%%%%%%%%%%%%%%%%%%%%%%%%%%%%%%%%%%%%%%%%%%%%%%%%%%%%%%%%%%%%%%%%%%%%%%%%%%%%%%%%%%%%%%%%%%%%%%%%%%%%%%%%%%%%%%%%%%%%%%%%%%%
\section{\dred Sleep Spindle Detection\label{sec:param_spindle}}
%%%%%%%%%%%%%%%%%%%%%%%%%%%%%%%%%%%%%%%%%%%%%%%%%%%%%%%%%%%%%%%%%%%%%%%%%%%%%%%%%%%%%%%%%%%%%%%%%%%%%%%%%%%%%%%%%%%%%%%%%%%%%%%%%%%%%%%%%%%%%%%%%%%%%%%%%%%
%%%%%%%%%%%%%%%%%%%%%%%%%%%%%%%%%%%%%%%%%%%%%%%%%%%%%%%%%%%%%%%%%%%%%%%%%%%%%%%%%%%%%%%%%%%%%%%%%%%%%%%%%%%%%%%%%%%%%%%%%%%%%%%%%%%%%%%%%%%%%%%%%%%%%%%%%%%
\subsection{Parameter Selection}
%%%%%%%%%%%%%%%%%%%%%%%%%%%%%%%%%%%%%%%%%%%%%%%%%%%%%%%%%%%%%%%%%%%%%%%%%%%%%%%%%%%%%%%%%%%%%%%%%%%%%%%%%%%%%%%%%%%%%%%%%%%%%%%%%%%%%%%%%%%%%%%%%%%%%%%%%%%
In this Appendix, we develop a methodology to tune the regularization parameters to solve the optimization problem in (\ref{eqn:spin_l1_min_cf_2}). We begin by the discretizing the regularization parameters so that $\lambda_1 \in [0.30,0.80]$, $\lambda_1 \in [3.0,6.0]$, and $\lambda_2 \in [3.0,6.0]$, in step sizes of $0.1$, $0.2$, and $0.2$, respectively. To minimize the computational load in identifying the feasible operating region, we use the same approach as discussed in Appendix \ref{sec:param_kcomplex}, and select one epoch from every dataset.

\begin{figure}[t!]
\centering
\includegraphics[width=0.49\textwidth, height=0.46\textwidth]{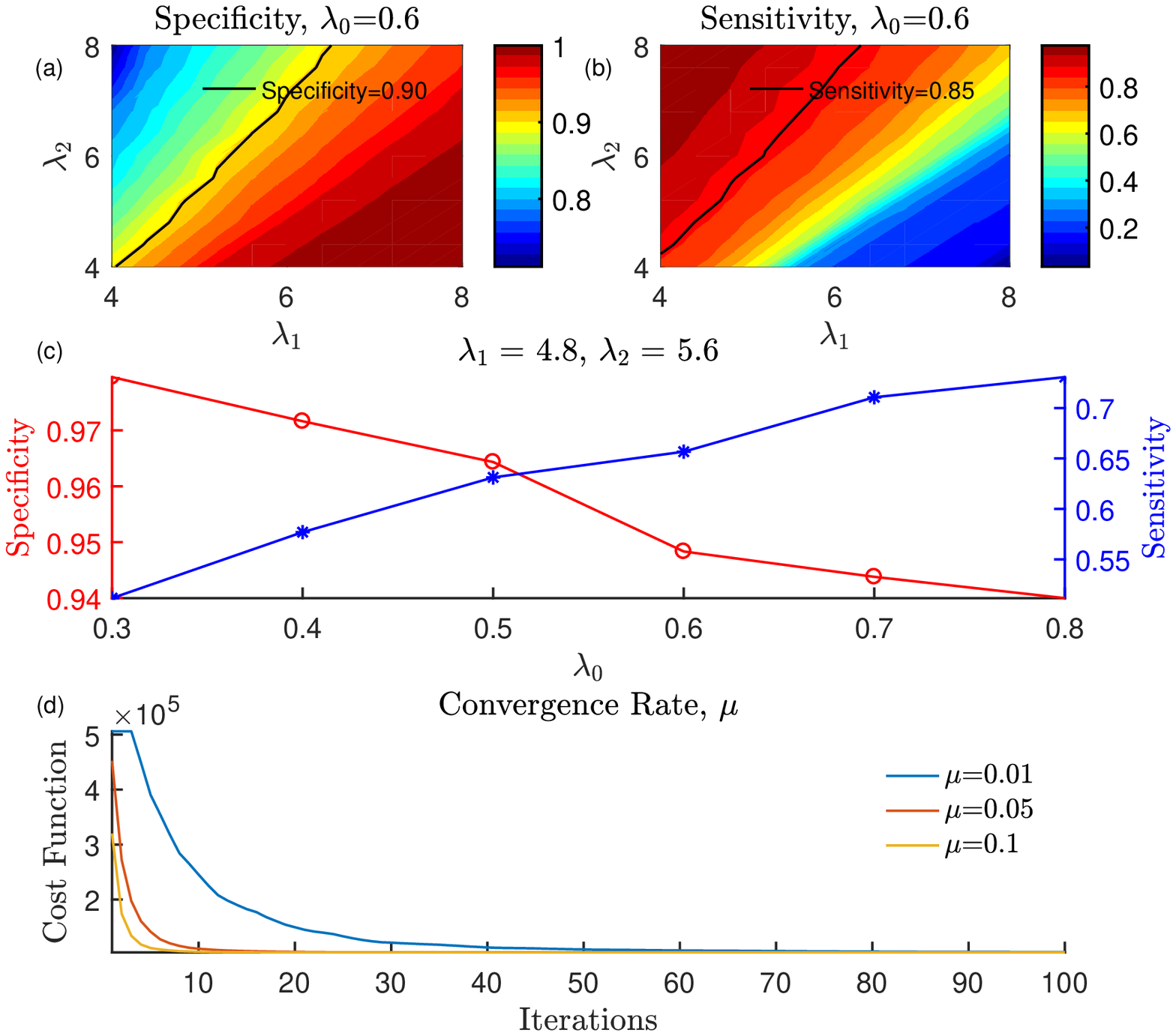}
\caption{Regularization parameters for SASDPR algorithm. (a) Specificity plot across all values of $\lambda_1$ and $\lambda_2$ for fixed $\lambda_0 = 0.6$. (b) Sensitivity plot across all values of $\lambda_1$ and $\lambda_2$ for fixed $\lambda_0 = 0.6$. (c) Specificity and sensitivity across different $\lambda_0$. (d) Cost function across different values of $\mu$. \label{fig:rparam_spindle}}
\end{figure}
Next, to determine a feasible operating region of the regularization parameters, we perform a grid search using the selected epochs (total of four epochs, each of $30$ seconds length, i.e., $N = 6000$ sample points). To evaluate the performance of the SASDPR method, we compute specificity and sensitivity across all sample points of the selected epochs. In Fig. \ref{fig:rparam_spindle}(a) and (b), we plot the average value of the sensitivity and specificity across all $\lambda_1$ and $\lambda_2$ for fixed $\lambda_0 = 0.6$. As can be seen, regions containing high values of specificity demonstrate low sensitivity and vice versa. To find a good balance between specificity and sensitivity, we use contour plot (solid black line in Fig. \ref{fig:rparam_spindle}(a) and (b)) to indicate the regions where the specificity and sensitivity are $0.90$ and $0.85$, respectively. To verify that $\lambda_0 = 0.6$ generates the best performance, we plot the specificity and sensitivity for different values of $\lambda_0$ by selecting one point of $\lambda_1$ and $\lambda_2$, i.e., $\lambda_1 = 4.8$ and $\lambda_2 = 5.6$ which belongs to the feasible operating region of the regularization parameters. As can be seen in Fig. \ref{fig:rparam_spindle}(b), the two curves representing specificity and sensitivity intersect in the interval of $\lambda_0 \in [0.5,0.6]$. Finally, to determine the parameter $\mu$, we compute the average value of the cost function in (\ref{eqn:spin_l1_min_cf_2}) across the selected epochs for different values of $\mu$. Note that $\mu$ only affects the convergence rate of the algorithm and does not affect the final value of the cost function. As can be seen in Fig. \ref{fig:rparam_spindle}(c), the algorithm converges fastest when $\mu = 0.1$.

\begin{figure}[t!]
\centering
\includegraphics[width=0.49\textwidth, height=0.15\textwidth]{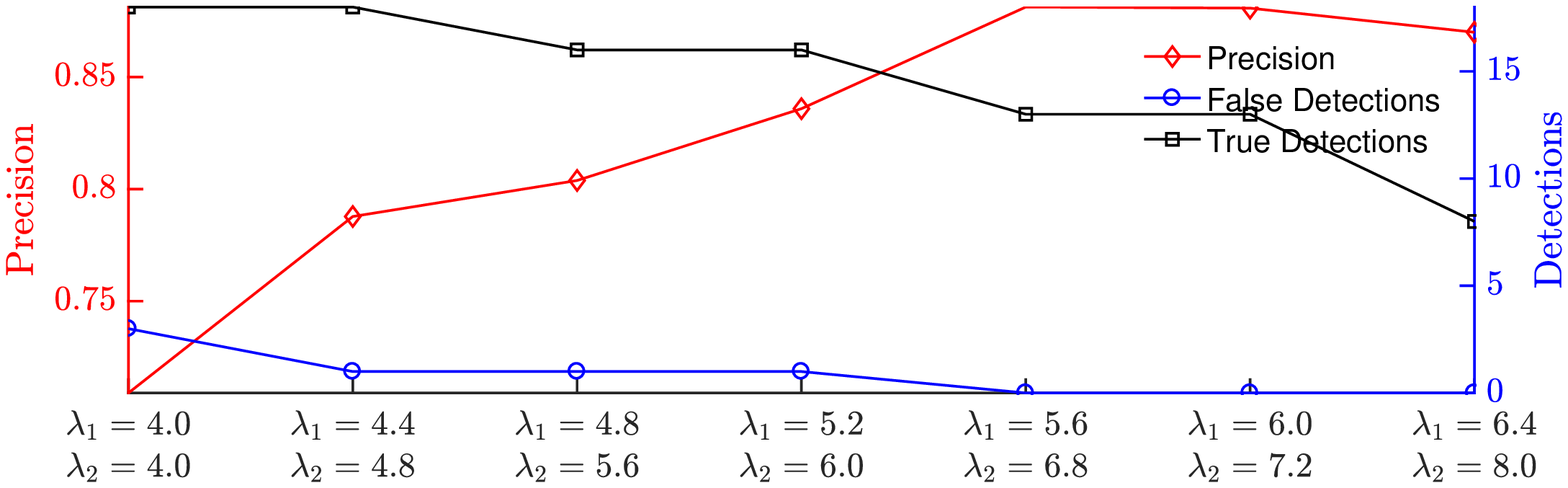}
\caption{Precision, false detections, and true detections across different values of $\lambda_1$ and $\lambda_2$ when $\lambda_0 = 0.6$. \label{fig:rparam_spindle_1}}
\end{figure}
Using Fig. \ref{fig:rparam_spindle}, we reduced the search space of the regularization parameters from a three dimensional space to a feasible operating region between the two solid lines in Fig. \ref{fig:rparam_spindle}(a) and (b). To find the best $\lambda_1$ and $\lambda_2$, we select points in the feasible operating region and plot the performance of the SASDPR algorithm in Fig. \ref{fig:rparam_spindle_1} for the selected epochs. We notice that the number of correctly and falsely detected events decreases on increasing the value of the regularization parameters. However, the average value of precision across all sample points of the selected epochs increases on increasing the value of the regularization parameters. To find a good balance between the performance metrics identified in Fig. \ref{fig:rparam_spindle_1}, we select $\lambda_1 = 4.8$ and $\lambda_2 = 5.6$.
%%%%%%%%%%%%%%%%%%%%%%%%%%%%%%%%%%%%%%%%%%%%%%%%%%%%%%%%%%%%%%%%%%%%%%%%%%%%%%%%%%%%%%%%%%%%%%%%%%%%%%%%%%%%%%%%%%%%%%%%%%%%%%%%%%%%%%%%%%%%%%%%%%%%%%%%%%%
\subsection{Example}
%%%%%%%%%%%%%%%%%%%%%%%%%%%%%%%%%%%%%%%%%%%%%%%%%%%%%%%%%%%%%%%%%%%%%%%%%%%%%%%%%%%%%%%%%%%%%%%%%%%%%%%%%%%%%%%%%%%%%%%%%%%%%%%%%%%%%%%%%%%%%%%%%%%%%%%%%%%
In Fig. \ref{fig:spindle_detoks}, we plot the output of the DETOKS algorithm for sleep spindle detection using sleep-EEG data from \texttt{excerpt5.edf} dataset. The algorithm decomposes the input signal into three components: a) low-frequency signal, b) oscillatory signal, and c) sum of sparse and sparse-derivative signal. A Butterworth band-pass filter with passband $11${--}$16$ Hz is applied to the oscillatory pattern to obtain the output in Fig. \ref{fig:spindle_detoks}(b). In Fig. \ref{fig:spindle_detoks}(c), using the TKEO, we detect the sleep spindle regions so that the TEKO is above a fixed threshold of $0.03$.
\begin{figure}[t!]
\centering
\includegraphics[width=0.47\textwidth, height=0.46\textwidth]{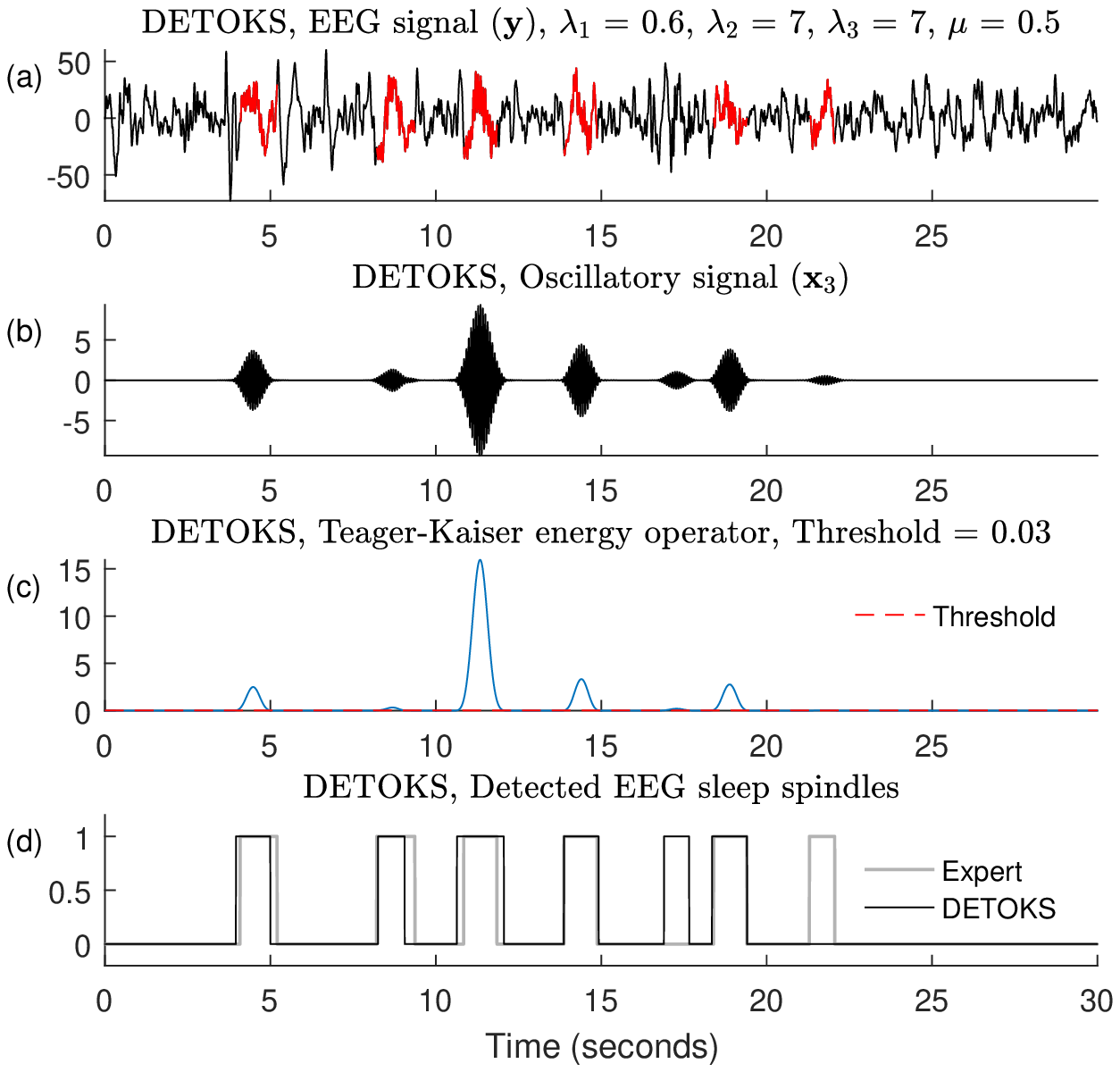}
\caption{Spindle detection. (a) 30 second epoch of sleep-EEG data obtained from \texttt{excerpt5.edf}. The epoch consists of six spindles as identified by two experts. (b) Oscillatory signal component $\bm{x}_2$ detected using the DETOKS algorithm. (c) Signal obtained by applying Teager-Kaiser energy operator on the extracted oscillatory signal component. (d) Expert and algorithm annotated sleep spindle regions. \label{fig:spindle_detoks}}
\end{figure}
\vfill \null
% use section* for acknowledgment
% \section*{Acknowledgment}

% The authors would like to thank...

% Can use something like this to put references on a page
% by themselves when using endfloat and the captionsoff option.
\ifCLASSOPTIONcaptionsoff
  \newpage
\fi

\end{document}